%% file: main.tex
\documentclass[11pt]{article}
\usepackage{amsmath}
\usepackage{amsthm}
\usepackage{amssymb}
\usepackage{algorithm}
\usepackage{color}
\usepackage[english]{babel}
\usepackage{graphicx}
\usepackage{grffile}
\usepackage{wrapfig,epsfig}
\usepackage{epstopdf}
\usepackage{url}
\usepackage{color}
\usepackage{epstopdf}
\usepackage{algpseudocode}
\usepackage[T1]{fontenc}
\usepackage{bbm}
\usepackage{comment}
\usepackage{dsfont}
\usepackage{thm-restate}
\usepackage{bm}
\usepackage{tcolorbox}
\usepackage{subcaption}

\usepackage{enumitem}

\usepackage{tikz}
\usetikzlibrary{arrows}

\usepackage[margin=1in]{geometry}

\graphicspath{{./figs/}}
\usepackage{mathtools}

\newtheorem{theorem}{Theorem}[section]
 
\newtheorem{lemma}[theorem]{Lemma}
\newtheorem{definition}[theorem]{Definition}

\newtheorem{proposition}[theorem]{Proposition}

\newtheorem{observation}[theorem]{Observation}

\newtheorem{remark}[theorem]{Remark}
\newtheorem*{remark*}{Remark}
\newtheorem{claim}[theorem]{Claim}

\newcommand{\wh}{\widehat}
\newcommand{\wt}{\widetilde}

\newcommand{\eps}{\epsilon}
\newcommand{\N}{\mathcal{N}}

\renewcommand{\varepsilon}{\epsilon}
\renewcommand{\tilde}{\wt}
\renewcommand{\hat}{\wh}
\renewcommand{\bar}{\overline}
\renewcommand{\eps}{\epsilon}

\newcommand{\nil}{\mathsf{nil}}

\newcommand{\D}{\mathcal{D}}
\newcommand{\mH}{\mathcal{H}}

\newcommand{\mU}{\mathcal{U}}

\newcommand{\mP}{\mathcal{P}}
\newcommand{\mE}{\mathcal{E}}
\newcommand{\mF}{\mathcal{F}}
\newcommand{\mL}{\mathcal{L}}
\newcommand{\mQ}{\mathcal{Q}}
\newcommand{\mN}{\mathcal{N}}

\newcommand{\mI}{\mathcal{I}}
\newcommand{\mW}{\mathcal{W}}
\newcommand{\suc}{\mathsf{suc}}

\newcommand{\IC}{\mathsf{IC}}

\DeclareMathOperator*{\E}{{\mathbb{E}}}

\DeclareMathOperator{\ALG}{ALG}

\DeclareMathOperator{\poly}{poly}
\DeclareMathOperator{\polylog}{polylog}
\DeclareMathOperator{\argmax}{argmax}

\DeclareMathOperator{\out}{out}

\usepackage{braket}

\algnewcommand\algorithmicforeach{\textbf{for each}}
\algdef{S}[FOR]{ForEach}[1]{\algorithmicforeach\ #1\ \algorithmicdo}

\makeatletter
\newcommand*{\RN}[1]{\expandafter\@slowromancap\romannumeral #1@}
\makeatother

\title{Near Optimal Memory-Regret Tradeoff for Online Learning}
\author{}
\author{  
Binghui Peng\footnote{Supported by NSF CCF-1703925, IIS-1838154, CCF-2106429, CCF-2107187, CCF-1763970, AF2212233, COLL2134095, COLL2212745} \\ Columbia University \\  \texttt{bp2601@columbia.edu} 
\and Aviad Rubinstein\footnote{Supported by NSF CCF-1954927, and a David and Lucile Packard Fellowship} \\ Stanford University \\ \texttt{aviad@cs.stanford.edu}
}
\date{}

\begin{document}
\maketitle

\begin{abstract}
In the experts problem, on each of $T$ days, an agent needs to follow the advice of one of $n$ ``experts''. After each day, the loss associated with each expert's advice is revealed. 
A fundamental result in learning theory says that the agent can achieve vanishing regret, i.e. their cumulative loss is within $o(T)$ of the cumulative loss of the best-in-hindsight expert.

Can the agent perform well without sufficient space to remember all the experts?
We extend a nascent line of research on this question in two directions:
\begin{enumerate}
    \item We give a new algorithm against the oblivious adversary, improving over the memory-regret tradeoff obtained by [PZ23], and nearly matching the lower bound of [SWXZ22].
    \item We also consider an adaptive adversary who can observe past experts chosen by the agent. In this setting we give both a new algorithm and a novel lower bound, proving that roughly $\sqrt{n}$ memory is both necessary and sufficient for obtaining $o(T)$ regret. 
\end{enumerate}
\end{abstract}

\setcounter{page}{0}
\thispagestyle{empty}

\newpage
\input{intro}

\input{preliminary}

\input{oblivious}

\input{adaptive-upper.tex}
\input{adaptive-lower.tex}

%\newpage
\bibliographystyle{alpha}
\bibliography{ref}

\newpage
\appendix
\input{appendix-prob}

\end{document}

%% file: intro.tex
\section{Introduction}
\label{sec:intro}
Consider the {\em online learning} problem of an agent who has to make decisions in an unknown and dynamically changing environment. The standard discrete abstraction of this problem considers a sequence of $T$ days; every morning $n$  {\em experts} offer their ``advice'' to the agent, who chooses to follow one of them; every night the agent observes the loss incurred by her choice as well as all the other choices she could have made%
\footnote{The model where the agent learns the loss-in-hindsight on actions she didn't take is called {\em full feedback} . In this work we do not study the alternative {\em bandit} model where the agent has to explore the different actions to learn their loss distribution.}%
. 
The benchmark for this problem is minimization of {\em regret}, aka
$$ (\text{Loss of agent's actions}) - (\text{Loss of single best-in-hindsight expert}).$$

Amazingly, simple algorithms like multiplicative weights update (MWU) achieve total regret that is asymptotically vanishing as a fraction of the total possible loss, even in the face of an adversarially chosen loss sequence. 
This celebrated result has numerous applications in both theory and practice, including boosting \cite{freund1997decision}, learning in games \cite{freund1999adaptive,cesa2006prediction}, approximate algorithm for max flow \cite{christiano2011electrical} and etc.

Recently, Srinivas, Woodruff, Xu and Zhou~\cite{woodruff2022memory} initiated the study of memory bounds for online learning. They gave a near-tight characterization of the memory-regret tradeoff in the case where each day the loss is drawn independently from the same distribution~(alternatively, there is a fixed set of loss vectors that arrives in a random order). In particular, they gave an algorithm that obtains vanishing regret while only using $\polylog(nT)$ memory. 
\cite{woodruff2022memory} posed as an open question what memory-regret tradeoff is possible without the i.i.d.~or random order assumptions. The first progress on this problem was obtained by Peng and Zhang~\cite{peng2022online} who gave a low memory algorithm with vanishing regret. 
Our first contribution is a new online learning algorithm in this setting, improving over~\cite{peng2022online}'s regret-memory tradeoff, and essentially matching~\cite{woodruff2022memory}'s lower bound.\footnote{
	Our space-efficient algorithms do not need to access the entire loss vector at once ---  (single-pass) streaming access suffices. %; see also previous ~\cite{peng2022online,woodruff2022memory}).
	Meanwhile, our lower bound continues to hold against stronger algorithms that have arbitrary internal memory on each day for processing the loss vector, and are only limited in how much information they store between days.}

\begin{restatable}[Algorithm, oblivious adversary]{theorem}{Oblivious}
	\label{thm:oblivious-main}
	Let $n , T$ be sufficiently large integers and let $\polylog (nT) \leq S \leq n$. 
	There is an online learning algorithm that uses at most $S$ space and achieves $\widetilde{O}\left(\sqrt{\frac{nT}{S}}\right)$ regret against an oblivious adversary, with probability at least $1-1/(nT)^{\omega(1)}$.
\end{restatable}

Note that, up to polylogarithmic factors, this memory-regret tradeoff is tight, hence resolving the open question posed by~\cite{woodruff2022memory}. 
It is also interesting to note that~\cite{woodruff2022memory}'s lower bound holds for the special case of i.i.d.~generated loss. Thus our algorithm implies that the much more general oblivious adversary is actually not strictly harder.

All the low memory algorithms discussed so far hold against an {\em oblivious} adversary who fixes the loss sequence in advance. But an even more amazing property of the classical MWU algorithm is that it achieves vanishing regret even in the face of an {\em adaptive adversary} that can update the current loss vector depending on the agent's past choices. This extension to adaptive adversary is crucial for analyzing settings where the agent {\em interacts} with the environment, as is the case in the most important applications to game theory (see e.g.~\cite[Remark 3.2]{Roughgarden13}).

We give an online learning algorithm with sublinear memory against an adaptive adversary. Specifically, our algorithm can achieve $o(T)$ regret while using only $\tilde{O}(\sqrt{n})$ memory.

\begin{restatable}[Algorithm, adaptive adversary]{theorem}{Adaptive}
	\label{thm:adaptive-upper-main}
	Let $n, T$ be sufficiently large integers. There is an online learning algorithm that uses up to $S$ space and achieves $\wt{O}\left(\max\left\{\sqrt{\frac{nT}{S}}, \frac{\sqrt{n}T}{S}\right\}\right)$ regret against an adaptive adversary, with probability at least $1-1/(nT)^{\omega(1)}$.
\end{restatable}

Finally, we give a novel lower bound for the adaptive adversary case, showing that $\tilde{\Omega}(\sqrt{n})$ memory is also necessary for guaranteeing vanishing regret. 
\begin{theorem}[Lower bound against adaptive adversary]
	\label{thm:lower-main}
	Let $\eps \in (\frac{1}{\sqrt{n}}, 1)$, any online learning algorithm guaranteeing $o(\eps^2 T)$ regret against an adaptive adversary requires space $\tilde{\Omega}(\frac{\sqrt{n}}{\eps})$.
\end{theorem}

\begin{remark*}[Running time]
	The focus of this paper is memory. We note that both algorithms (Theorem \ref{thm:oblivious-main} and Theorem \ref{thm:adaptive-upper-main}) are computationally efficient ---   the running time (per round) scales linearly with the memory. 
\end{remark*}

\subsection{Related work}

\paragraph{Comparison with previous work~\cite{woodruff2022memory, peng2022online}} The work of \cite{woodruff2022memory,peng2022online} are closely related to us and we provide a detailed discussion.
Srinivas, Woodruff, Xu and Zhou \cite{woodruff2022memory} initiate the study of memory bounds for online learning and provide a regret lower bound of $\Omega(\sqrt{\frac{nT}{S}})$ that holds even for i.i.d.~loss sequence. 
Their lower bound reduces from the coin problem, and it is proved via a direct-sum argument and a strong data-processing inequality. 
Both the construction and the proof technique are quite different in comparison with our lower bound for adaptive adversary. \cite{woodruff2022memory} also provide an algorithm with $O(\sqrt{\frac{nT}{S}})$ regret, when the loss sequence arrives in {\em random order}. It is not surprising that their analysis heavily relies on the random nature of loss sequence. 

Peng and Zhang \cite{peng2022online} provide the first sublinear space online learning algorithm for oblivious adversary. 
Their algorithm obtains a total regret of $\wt{O}(n^2T^{\frac{2}{2+\delta}})$ using $n^{\delta}$ space (in concurrent work,~\cite{aamand2023improved} improved this regret bound to $\wt{O}(n^2T^{\frac{1}{1+\delta}})$). Our baseline algorithm shares some similarity with \cite{peng2022online,aamand2023improved}, but as we shall discuss soon, the key component, i.e., the eviction rule is quite different. Our robustifying procedure, amortization technique, and low memory ``monocarpic'' experts 
procedure are all new components that enable us to obtain near-optimal tradeoffs for the oblivious adversary. 

Peng and Zhang \cite{peng2022online} also provide a space lower bound of $\Omega(\min\{\eps^{-1}\log n, n\})$ for any algorithm with $O(\eps T)$ regret, against a {\em strong} adaptive adversary.
As noted in \cite{peng2022online}, the strong adversary could see the distribution/mixed strategy of the algorithm in the {\em current} round, which is stronger than the adaptive adversary usually considered in the literature.
Our lower bound is quantitatively stronger, and holds even against the weaker and more standard notion of adaptive adversary.
In terms of technique, \cite{peng2022online} establish their lower bound via a counting argument and a connection to learning in games (which requires the strong adversary assumption). In contrast, our lower bound instance is constructed based on the standard notion of adaptive adversary who can attack the algorithm's recent actions;  the proof relies on direct-product theorems from communication complexity.
In fact, using our methodology, it is easy to derive an $\Omega(n)$ lower bound for any sub-linear regret algorithm $o(T)$, for the strong adaptive adversary (see Section \ref{sec:strong-adap-app}).

\paragraph{Comparison with work of~\cite{woodruff2023streaming} on adaptive adversary}
\cite{woodruff2023streaming} also provide lower bounds and algorithms for an adaptive adversary, but both the set up and the technique are quite different from ours.
In particular: (1) They focus on the binary prediction task, while we study the general model where an adversary observes the specific experts chosen by algorithm in previous rounds. (2) They consider an additional parameter, $M$, the number of incorrect binary predictions made by the best expert.  In this setting they give a tight bound on the space-regret tradeoff for deterministic algorithms, and give a randomized algorithm that can improve over this tradeoff when $M = o(T/\log^2 n)$. Even though the settings are different, note that our algorithm for the adaptive adversary obtains sub-linear regret with sub-linear memory without any assumption on the loss of the best expert, and that our lower bound holds against randomized algorithms.

\paragraph{Online learning} 
The classic MWU algorithm has shown to achieve optimal regret \cite{littlestone1989weighted,ordentlich1998cost} for online learning with experts advice.
As a meta approach, it has found a wide range of applications in algorithm design and analysis, includes boosting \cite{freund1997decision}, learning in games \cite{freund1999adaptive,cesa2006prediction}, approximate algorithm for max flow \cite{christiano2011electrical} and etc \cite{garber2016sublinear,klivans2017learning,hopkins2020robust}. 
See the survey \cite{arora2012multiplicative} for a complete treatment.

The MWU algorithm falls into the broader framework of following the randomized/perturbed leader \cite{kalai2005efficient,hazan2016introduction}, a general framework that subsumes almost all existing online learning algorithms (except \cite{peng2022online}).
As observed in \cite{woodruff2022memory,peng2022online}, identifying even a (constant) approximately best expert requires $\Omega(n)$ memory, hence new ideas are required for a low memory online learning algorithm. 
The online learning problem has also been studied in other context, examples including oracle efficient algorithm \cite{hazan2016computational,dudik2020oracle}, smooth adversary \cite{nika2022smooth, nika2022oracle, adam2022smooth} and bandit feedback \cite{auer2002nonstochastic,bubeck2012regret}.
Other extended regret notions have been considered, including second order regret \cite{koolen2015second}, interval regret \cite{hazan2007adaptive}, sleeping expert \cite{freund1997using,blum2007external} and others \cite{besbes2015non,chen2020hedging,anagnostides2022near,wei2018more,bubeck2019improved}.

\paragraph{Memory bounds for learning}Memory bounds have been established in {\em stochastic} multi-arm bandit and pure exploration \cite{liau2018stochastic, assadi2020exploration,maiti2021multi, agarwal2022sharp}. They are not worst case or adversarial in nature. 

The role of memory for learning has been extensively studied in the past decade, including time-space trade-off for parity learning \cite{raz2017time,raz2018fast,garg2018extractor,garg2019time,garg2021memory}, memory bounds for convex optimization \cite{marsden2022efficient,blanchard2023quadratic}, continual learning \cite{chen2022memory}, generalization \cite{feldman2020does, brown2021memorization} and many others \cite{steinhardt2016memory, sharan2019memory, gonen2020towards,brown2022strong,dagan2018detecting}.

Communication complexity has been a useful approach for establishing memory lower bound. In particular, our lower bound proof critically utilizes the machinary of direct-product theorem \cite{raz1995parallel,holenstein2007parallel,braverman2015interactive,braverman2013direct,klauck2010strong} and the information cost of set disjointness \cite{bar2004information, sherstov2014communication}.

\input{discussion}

\section{Overview of techniques}

We provide a high level overview on our techniques. 

\subsection{Algorithm for oblivious adversary}

We focus on the case of $S = \polylog(nT)$ and aim for a regret bound of $\wt{O}\left(\sqrt{nT}\right)$, the extension to a general dependence on $S$ is obtained via a grouping trick that is deferred to the end.

\subsubsection{The baseline building block}
We first present a baseline subroutine that obtains a total regret of $\wt{O}(nB + \frac{T}{\sqrt{B}})$. 
The baseline algorithm proceeds epoch by epoch, with a total of $\frac{T}{B}$ epochs and  $B$ days in each epoch. 
The algorithm maintains a pool $\mP$ of experts and runs MWU over them in each epoch; the pool is fixed within the epoch. 
The pool is updated every epoch, it randomly samples and adds $O(1)$ experts at the beginning, and the algorithm prunes the pool at the end.
So far this is the same basic setup as~\cite{peng2022online}:  the algorithm is now competitive to the best experts in the pool and the task is to maintain a good set of experts.

\paragraph{Eviction rule -- A first attempt} 
The first key step and the deviation from previous work of \cite{peng2022online} is a new eviction rule. 
Recall we have only polylogarithmic space and need to keep a small set of core experts.
For each expert $i \in \mP$ we compare its aggregate loss over its lifetime in the pool with the following very powerful benchmark:
\begin{quote}
	{\bf The Covering Benchmark:} Consider a hypothetical algorithm that in each epoch switches to the single best (in hindsight) expert in the pool {\em for that epoch}. Our benchmark is the total aggregate loss of this hypothetical algorithm\footnote{Contrast this benchmark with the standard notion of regret where the benchmark is a single expert for the entire stream.}. 
\end{quote} %--- the aggregate loss of taking the best (in hindsight) expert in the pool {\em in each epoch}. 
If the expert does not have a strictly better loss than this benchmark, we say that it is {\em covered} by the other experts in the pool and it is subject to eviction. 
Intuitively, while this benchmark is strong, it suffices that an expert is covered by the pool to guarantee that running MWU over the pool would have low regret compared to that expert.

\begin{itemize}
	\item {\bf Memory.} The immediate advantage of this powerful benchmark is the saving of memory. 
	For a set of experts that survive the eviction, it is easy to see that the suffix spanned by the last $r$ experts has an overall advantage of at least $2^{r-1}$.
	This indicates $|\mP|\leq O(\log T)$, otherwise the maximum loss exceeds $T$. 
	In other words, only $O(\log T)$ experts survive after the pruning step!
	Finally, we note that to avoid recording the loss of each epoch, it actually suffices to record the pairwise loss and the memory scales quadratically with the size of the pool.
	
	\item {\bf Regret analysis (wishful) } Let $i^{*}$ be the best expert. For the regret analysis, we divide the entire epochs into two types. An epoch is ``good'', if had we sampled  $i^{*}$ in this epoch, it will eventually be covered and evicted from the pool. Consider the interval from the current good epoch until $i^*$ is kicked out of the pool. We want to argue the algorithm's performance on the entire interval is comparable to $i^*$  even if $i^*$ is never sampled. An epoch is ``bad'', if had we sampled  $i^{*}$ in this epoch, it will never be kicked out from the pool. We have no control over the regret in a bad epoch, 
	but intuitively we can only suffer a few bad epochs before we sample $i^{*}$ and keep it forever.
	%but we wish to control its size.
\end{itemize}

There is a fatal issue in formalizing the above outlined approach.
We want to analyse the regret by considering a hypothetical situation that $i^{*}$ is sampled into the pool at a given epoch, but the eviction time could depend on experts enter lately in the pool. 
That is to say, we need to fix the entire sequence of selected experts when defining the good/bad epochs. However, once $i^{*}$ is really sampled into the pool, it interferes with the rest of pool, because adding $i^{*}$ into the pool might kick out other experts, the kicked out expert might allow the pool to keep other experts (that were supposed to be removed), etc. This issue seems subtle but fixing it turns out to be technically very challenging!% part of the paper.

\paragraph{A robust algorithm} We depict the high level intuition for fixing the above issue. The overall idea is to make the entire algorithm insensitive or ``robust'' to the addition of a single new expert in any epoch. %, or in other word, the algorithm should be robust to addition/removal of experts.
Instead of performing the pruning step in every epoch, we maintain $L +1 = O(\log T)$ disjoint sub-pools $\mP = \mP_0 \cup \cdots \cup \mP_{L}$, where the $\ell$-th sub-pool $\mP_{\ell}$ ($\ell \in [0:L]$) contains experts added in the most recent $2^{\ell}$ epochs, and it is merged into the higher sub-pool $\mP_{\ell+1}$ every $2^{\ell}$ epochs. This modification is critical for a robust algorithm as it ensures an expert can potentially participate in the eviction process at most $O(\log T)$ times (instead of every epoch). 

Given two adjacent sub-pools $\mP_{\ell}, \mP_{\ell+1}$, the merging step needs to ensure (1) the output sub-pool has small size; and (2) the outcome should be insensitive to a single change\footnote{At this stage, it seems like we need a ``differential private merging/pruning'' procedure. Differential privacy techniques were indeed an inspiration, but they do not directly apply. 
	%Our eviction rule is very sensitive to a single change of input and a point-wise perturbation guarantee of DP seems to be impossible.  
	The memory analysis is very delicate and sensitive to even a small perturbation caused by a standard DP mechanism. 
}.
The final merging and pruning step is an iterative procedure. 
At each iteration, we estimate the size of $\mP_{\ell} \cup \mP_{\ell+1}$ (in an insensitive way), and if it is too large, then we sample $\frac{1}{\polylog (nT)}$-fraction of experts and use them to filter the sub-pool. This is repeated for $O(\log(nT))$ times and we prove that it removes $\Omega(\frac{1}{\log (nT)})$-fraction of experts at each iteration w.h.p.

\paragraph{A more detailed description of the baseline algorithm} In summary, the baseline algorithm maintains $L + 1$ sub-pools $\mP$, and for  $\ell \in [0:L]$ the $\ell$-th sub-pool $\mP_{\ell}$ is merged with $\mP_{\ell+1}$ every $2^{\ell}$ epochs.  The merge and pruning step proceeds in at most $K = O(\log (nT))$ iterations. At each iteration, it first estimates the size of $\mP_\ell\cup \mP_{\ell + 1}$ by sub-sampling and proceeds only if the size is large; 
it then samples $\frac{1}{\polylog(nT)}$-fraction of the experts as the {\em filter} set $\mF$. An expert is removed if it is covered by $\mF$.

\begin{itemize}
	\item {\bf Memory analysis. } The memory analysis now becomes involved and our goal is to prove $\Omega(\frac{1}{\log(nT)})$-fraction of experts are removed per iteration. To do so, we {\em recursively} identify (at most) $O(\log T)$ pivotal experts in $\mF$ such that the suffix of the last $r$ pivotal experts is at least $2^{r-1}$ better than a large fraction of remaining experts in the pool, unless $\Omega(\frac{1}{\log(nT)})$-fraction of experts has already been removed during the recursive procedure.
	\item {\bf Regret analysis. } An expert is {\em passive} if it never ``participates'' in the eviction process, i.e., it is never sampled in the filter set $\mF$ and the size estimation. It is active otherwise. The probability of being a passive expert is large, i.e., at least $(1 - \frac{1}{\polylog(nT)})^{2L\times 2K} \geq 1/2$.
	In the analysis, we fix the set of {\em sampled and active} experts at each epoch and we divide the entire epochs into two types, defined similarly as before. 
	The set of {\em sampled and passive} experts are not fixed, but they do not influence the execution of the baseline algorithm (i.e., the estimate size, the filter set, the alive set of active experts).
	The regret is controlled well for good epochs. 
	While for a bad epoch, if the expert $i^{*}$ is sampled and at the same time, being passive, then it would never be kicked out. 
	As the expert $i^{*}$ has probability $\Omega(\frac{1}{n})$ being {\em sampled and passive}, and we have proved the pool size rarely exceeds $\polylog(nT)$, the number of bad epoch is bounded by $\wt{O}(n)$.
\end{itemize}

In summary, the regret of baseline comes from three source (1) $\wt{O}(\sqrt{B})$ per epoch due to MWU; (2) $1$ for each interval containing good epochs; (3) $B$ for each bad epoch, hence the regret is upper bounded as
$
\wt{O}(\sqrt{B}\cdot (T/B) + (T/B) + B \cdot n) = \wt{O}(T/\sqrt{B} + nB).
$

\paragraph{A simple but sub-optimal bootstrapping } The baseline approach obtains $\wt{O}(n^{1/3}T^{2/3})$ regret after optimizing the epoch size.
A simple bootstrapping approach (similar and even simpler than the counterpart of \cite{peng2022online}) can be used to decrease the regret and achieve optimal dependence on $T$. 
Concretely, we have used a fairly crude regret bound of $B$ for bad epochs, this can be reduced by running a baseline (with sequence length $B$ instead of $T$) separately for each epoch, and the final decision is determined by running MWU over this epoch-wise baseline (sequence length $B$) and the entire baseline (sequence length $T$).
%(i.e., with sequence length $B$ instead of $T$) and running a MWU on top of that. 
By doing so, the regret of a bad epoch is reduced to $\wt{O}(n^{1/3}B^{2/3})$ instead of $B$, and the regret of good epochs remains the same order. 
The total regret becomes $\wt{O}(T/\sqrt{B} + n^{4/3}B^{2/3})$, and after balancing the epoch size, one obtains an improved regret of $\wt{O}(n^{4/7}T^{4/7})$.
The bootstrap method can be repeated for multiple times, attaining an algorithm of regret $\wt{O}(n\sqrt{T})$.
This already obtains the optimal dependence on $T$, but suboptimal by a factor of $\sqrt{n}$.

\subsubsection{Optimal boosting against an oblivious adversary}
The final oblivious adversary algorithm maintains multiple threads of the baseline with different parameters, and amortizes the regret carefully. 
The new ideas (comparing with \cite{peng2022online}) are (1) inheriting pools from a high frequency baseline to a low frequency baseline (instead of only the regret bound); (2) separating the pool maintenance and regret control.

\begin{figure}[!ht]
	\centering
	\includegraphics[scale=0.48]{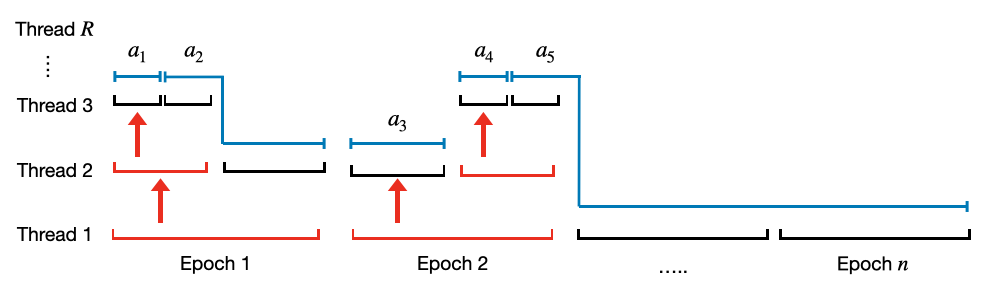}
	\caption{Epoch assignment. Red segments are bad epochs of each thread, blue segments are intervals generated by the walk. \\\\
		{\bf Detailed description of the walk in the figure:} As the first epochs in threads 1 and 2 are bad, the walk moves to the first epoch of thread 3, where $i^*$ can be covered and evicted (interval $a_1$). In the next epoch, $i^*$ is not immediately covered so it moves down to thread 2 where it is covered (interval $a_2$). The walk moves to the next epoch (epoch 2 in thread 1), which is bad. The walk moves up to thread 2 where it finds a good epoch that immediately covers $i^*$ (interval $a_3$). The next thread-2 is epoch is bad, so the walk moves up to thread 3; there $i^*$ is covered immediately (interval $a_4$). After that, $i^*$ survives until it is finally covered in the last epoch (interval $a_5)$.
	}
	\label{fig:walk}
\end{figure}

The full algorithm maintains $R = O(\log T)$ threads of baseline algorithms. The lowest thread has $n$ epochs and $B_1 = \frac{T}{n}$ days per epoch, the $r$-th ($r \in [2:R]$) thread has $2$ epochs only and each epoch contains $B_r = \frac{T}{2^{r-1} n}$ days -- it restarts every $T_{r} = B_{r-1}$ days. 

Every time we restart the $r$-th thread, we also start a new epoch of a lower thread (longer epochs); instead of discarding the experts in the $r$-th thread's pool, we move them to the lowest thread with a new epoch.
This is a key difference compared to the sub-optimal bootstrapping method. % is that the $r$-th level baseline inherits the pool from the $(r+1)$-th level baseline in the previous restart (in addition to the $O(1)$ randomly sampled experts). 
It makes intuitive sense because one may hope the optimal expert has a better chance of being retained after a few rounds of sampling/pruning in the higher thread (though the final analysis turns out to be quite different). The memory bounds largely follow from the baseline and we outline the regret analysis.

\paragraph{Amortizing regret} In the regret analysis, as before, we fix the set of {\em active and sampled} experts for each thread.
To account for all the regret incurred by the algorithm, our analysis splits the entire sequence into a collection of intervals, with different regret guarantees.
The splitting is defined by a walk over epochs at the different threads.
We note again that this walk is for analysis purposes only. 

The walk starts from the first epoch of the bottom thread (longest epochs). We say that an epoch considered in the walk is {\em bad} if the set of {\em alive active experts} (union over all threads) would never cover $i^{*}$ until the end of the algorithm, had the expert $i^{*}$ entered the pool at the beginning of the epoch.
\begin{enumerate}
	\item Even if the current epoch is bad, we wouldn't want to pay maximal loss for the entirety of the epoch. Instead, the walk moves up to a higher thread with shorter epochs, in hope that we could at least cover $i^{*}$ for part of the epoch. (It continues moving up while the epoch is bad and while there are higher threads with shorter epochs, if the walk reached the maximum thread, we can afford a short bad epoch.) 
	\item On good epochs, we can consider the interval from that epoch to the epoch where $i^{*}$ would be kicked out; it is guaranteed that the algorithm suffers a small regret over this interval (from the current epoch to the kicked-out epoch). By the definition of our tight boosting algorithm, if $i^{*}$ stays for sufficiently long, it will go down to lower threads (every time we restart $i^{*}$'s current thread, $i^{*}$  goes to the lowest thread with a new). Our walk follows $i^{*}$'s trajectory across epochs from different threads, until it is evicted. 
	\item Once $i^{*}$ is evicted, we continue the walk moves to the next epoch that $i^{*}$ would have gone to if it were not evicted. %(i.e.~the epoch that $i^{*}$'s pool-mates that evicted it go to). 
	We see if that epoch is good or bad, and then continue with Steps 1 or 2 above. 
\end{enumerate}

In the analysis, we consider the intervals defined by Step 2 of the walk. On each interval $\mI$, the regret scales like $\sqrt{|\mI|}$, so we particularly want to bound the total time spent on short intervals. To do this, we show that every lower thread epoch that fully contains $\mI$ must be bad, and as we argued before there cannot be too many bad epochs.

\paragraph{Interval regret guarantee for monocarpic experts} 
One important detail is missing -- the regret analysis goes through if we can guarantee an $O(\sqrt{|\mI|})$ regret (compared with the ``epoch-wise'' minimum of alive active experts) over any intervals $\mI$. 
This is non-trivial and no longer guaranteed by ``running MWU over each epoch''. 
In particular, an interval could consist of epochs from different threads, and more importantly, the set of alive active experts need not come from the same thread (they could come from all threads). 
We resolve the issue with a general procedure {\sc MonocarpicExpert}, which obtains {\em interval regret} over any alive expert.

%\footnote{%One can find monocarpic at \url{https://www.google.com/search?q=monocarpic&source=lnms&tbm=isch}}
We consider a variant of the sleeping experts problem~\cite{freund1997using,hazan2007adaptive} that we call {\em monocarpic experts}\footnote{The name is inspired by {\em monocarpic plants}, aka plants that only bloom once.}. In the monocarpic expert problem, the experts could enter and exit the pool at any time (but each expert only ever wakes up once) and we want an interval regret guarantee w.r.t.~any alive expert. 
That is, the algorithm is competitive with an expert at any time interval $\mI$ during its lifetime, with a regret of $O(\sqrt{|\mI|})$. 
The monocarpic experts problem is closely related to the sleeping expert problem, but the difference is that a sleeping expert could re-enter the pool (i.e., its lifetime is not a consecutive period). 
In this paper, we give a low memory algorithm for monocarpic experts with interval regret guarantees, whose memory scales linearly with the maximum number of alive experts. 
This is obtained by binary division tricks and utilizes the recent advance of second-order regret guarantee (i.e., the {\sc Squint} algorithm \cite{koolen2015second}).

\paragraph{Grouping trick} Finally, to obtain $\wt{O}(\sqrt\frac{nT}{S})$ regret with $S$ space, we partition the experts into $\wt{O}(S)$ groups, and each group contains $\wt{O}(n/S)$ experts. We maintain a separate meta-thread of algorithm for each group and run the MWU on top of these $\wt{O}(S)$ threads  of algorithms. 
This reduces the expert size from $n$ to $\wt{O}(n/S)$ and attains the  $\wt{O}(\sqrt{\frac{nT}{S}})$ regret.

\subsection{Lower bound for adaptive adversary}

We first delineate the high-level principle of constructing an adaptive adversary against a low memory (online learning) algorithm.
The adversary 
\begin{enumerate}
	\item Attacks the most recent $S_r = \Omega(S)$ actions picked by the algorithm; 
	\item Maintains a set of special experts $I$ ($ |I| = \Omega(S) \gg S_r$) that are good most of time.
\end{enumerate}
Intuitively, the first principle demonstrates the memory advantage of an adaptive adversary. 
While a low space online learning algorithm keeps only $S$ ``good'' experts in memory, an adaptive algorithm can make them bad once the algorithm releases them. The second principle ensures the existence of one good expert in $I$, this is because at any time, most experts in $I$ are not attacked (not played in the last $S_r$ rounds) and receive small loss.

For simplicity of exposition, we focus on the regime when the space $S = o(\sqrt{n})$ and we take the number of special experts $N = \sqrt{n}$. 
The adversary first draws $(x_{A}, x_{B}) \in \{0,1\}^n\times \{0, 1\}^n \sim \mu^{N}$, where $\mu$ is some hard distribution for the set disjointness problem, with $\frac{n}{N} = \sqrt{n}$ coordinates.
The loss sequence is divided into $\frac{10T}{N}$ epochs and each epoch contains $\frac{N}{10}$ days. 
The loss sequence of each epoch is constructed separately. During each epoch, the loss vector is set to $1 - x_{A}$ or $1- x_B$ with equal probability $1/2$, except for those coordinates that have been played by the algorithm during the current epoch -- they are fixed to be $1/2$.
It is easy to see that the best expert receives at most $\frac{T}{20}$ loss over the entire sequence.
To obtain $o(T)$ regret, an online learning algorithm needs to commit at least $\Omega(N)$ different intersecting coordinates (of the $N$ set disjointness instance) at some epoch (otherwise the loss received is $\frac{T}{2} - o(T)$).

A natural three-party communication problem can be formulated as follow. 
There is a prophet, Charlie, who knows all the $x_{A}$ and $x_B$, but it is only allowed to send $S$ bits of information at the beginning (this captures the initial memory state of an epoch). 
Alice and Bob receive the input $x_A$ and $x_B$ separately, and they are allowed to communicate at most $\frac{N}{10}\cdot S$ bits.
The goal is to output $\Omega(N)$ intersecting coordinate at the end.

The challenging part is the $S$ bits advice from the prophet.
Our key observation is that one can get rid of the advice by seeking for a direct-product theorem.
In particular, a too-good-to-be-true three-party communication protocol would imply a two-party communication protocol that solves the $N$-copy set disjointness problem, with a total of $S\cdot \frac{N}{10} = o(n)$ bits communication, and success probability at least $\exp(-\wt{O}(S))$ (without advice!). 
Roughly speaking, this holds from the fact that they can use the public randomness to guess the advice and it succeeds with probability $2^{-S}$ on average.
The set disjoint problem requires $\Omega(\sqrt{n})$ internal information cost per copy and
the direct-product theorem of \cite{braverman2015interactive} implies $\wt{\Omega}(n)$ bits of communication are necessary to resolve the problem with success probability at least $\exp(-\wt{O}(N)) \ll \exp(-\wt{O}(S))$.

\subsection{Algorithm against adaptive adversary}
\label{sec:tech-adaptive}

%Provided with our algorithm against an oblivious adversary, one may tempt to transform it into an algorithm against adaptive adversary in a black box manner. 
%For example, the recent work of \cite{hasidim2020adversarially} provides a block box method using differential privacy (DP) to turn a streaming algorithm designed for oblivious adversary into a streaming algorithm against adaptive adversary, by hiding internal randomness of the algorithm. 
%However, it is not applicable in our setting since the adversary observes the actual action played by the algorithm, instead of only a single prediction value that can be easily hidden via a DP median mechanism.\footnote{We note if the goal is to output a binary prediction, then the DP approach indeed obtains the same regret guarantee of Theorem \ref{thm:adaptive-upper-main}.} 
%We note that this seems to be an inherent limits of the DP approach and it is much more challenging to hide the actual action, which is critical for broader applications of the online learning algorithm.

Our algorithm for adaptive adversary completely differs from the one for oblivious adversary and it requires a set of new ideas.
Recall that we aim for an algorithm using $O(\sqrt{n}/\eps)$ space and obtaining $\eps T$ regret. 
A natural idea is to observe (but not commit) an expert for some time before actually committing the action (or adding it to the pool). 
However, the challenge is that an adaptive adversary can easily counter this type of algorithm by making an expert perform bad once it is released by the algorithm. 
In other words, it seems that an expert becomes useless as soon as it is revealed, despite the enormous effort and difficulty of finding such a good expert.

Our solution is conceptually simple and it consists of two components.
We provide a sub-optimal version first. 
\begin{itemize}
	\item {\sc RandomExpert}. The {\sc RandomExpert} procedure sub-samples a new set of $\sqrt{n}/\eps$ experts per epoch, where each epoch contains $1/\eps^2$ days.
	The goal of {\sc RandomExpert} is to be $\eps$-competitive with the top $\sqrt{n}/\eps$-th expert within the epoch (i.e., to suffer at most $\eps$ more average loss than the $\sqrt{n}/\eps$-th expert of the epoch).
	To handle the aforementioned challenge of adversary making the expert bad immediately after it is selected, we apply a common trick to make an algorithm work against adaptive adversary: The {\sc RandomExperts} divides the set into $1/\eps^2$ groups, with $\eps \sqrt{n}$ experts per group.  
	It runs MWU separately for each group of experts, and at day $b$ ($b \in [1/\eps^2]$), it follows the decision from the $b$-th group.
	The algorithm is guaranteed to be $\eps$-competitive with the top $\sqrt{n}/\eps$-th expert, because w.h.p. the {\sc RandomExpert} samples one from the top $\sqrt{n}/\eps$ experts for each group.
	\item {\sc LongExpert}. The {\sc LongExpert} procedure maintains a pool $\mP$ of experts and runs MWU over $\mP$ during each epoch. 
	At each epoch, it samples a new random set of $\sqrt{n}/\eps$ experts and only {\em observes} their loss. 
	It compares their performance with the {\sc RandomExpert} at the end of the epoch, and an expert is added to $\mP$, if its performance is $\Omega(\eps)$ better than the {\sc RandomExpert}. The expert is kept for $\sqrt{n}$ epochs.
\end{itemize}
The final algorithm runs MWU over the {\sc RandomExpert} and the {\sc LongExpert} procedure.

The analysis is fairly intuitively
\begin{itemize}
	\item {\bf Memory bounds.} The {\sc LongExpert} reserves $O(1/\eps^2)$ experts per epoch, since only the top $\sqrt{n}/\eps$ experts would be kept (due to the regret guarantee of {\sc RandomExpert}), and it samples $\sqrt{n}/\eps$ experts per epoch. 
	An expert is kept in $\mP$ for $O(\sqrt{n})$ epochs, so the total memory used is $O(\sqrt{n}/\eps^2)$.
	
	\item {\bf Regret analysis.} 
	For any expert $i\in [n]$, if it is not among the top $\sqrt{n}/\eps$ experts or it is not $\eps$ worse than the {\sc RandomExpert}, then it causes at most $\eps$-regret in an epoch.
	On the other hand (i.e., it is a top $\sqrt{n}/\eps$ expert and performs significantly better than the {\sc RandomExpert}), with probability $\frac{1}{\eps\sqrt{n}}$, it is sampled and kept for $\sqrt{n}$ epochs (and afflicts $\eps$ average regret over these $\sqrt{n}$ epochs). 
	Intuitively speaking, this means, in order to make our algorithm perform bad on one epoch (and causes at most 1 average regret / epoch),
	the adversary should punish this top expert for $\sqrt{n}$ epochs, worsening the loss of the average top expert by 
	$\frac{1}{\eps\sqrt{n}} \cdot \sqrt{n} = 1/\eps$ epochs in expectation, so an averaged of $\eps$ regret is obtained.
\end{itemize}

\paragraph{Optimizing the regret} 
The above algorithm obtains $O(\eps T)$ regret using $O(\sqrt{n}/\eps^2)$ space. 
The sub-optimality comes from that experts that are $\alpha$-better than the {\sc RandomExpert} are treated equally, for any $\alpha \in (\eps, 1]$. 
To circumvent this inefficiency, we maintain $R = \log_2(1/\eps)$ threads of {\sc RandomExpert}. The epoch size is still $1/\eps^2$ and let $\eps_r = 2^{r}\eps$ ($r\in [R]$). 
The $r$-th thread {\sc RandomExpert} samples $\eps_r^2\sqrt{n}/\eps$ experts, runs MWU for $1/\eps_r^2$ days and  repeats for $\eps_r^2/\eps^2$ times per epoch. 
It is guarantee to be $\eps_r$-competitive with the top $\eps \sqrt{n}/\eps_r^2$-th expert (a worse regret guarantee, but comparable against top performance experts). 
The {\sc LongExpert} maintains $R$ sub-pools $\mP_1, \ldots, \mP_{R}$, the $r$-th pool samples $\sqrt{n}/\eps$ experts and observes their performance for $1/\eps_r^2$ days (it also restarts  for $\eps_r^2/\eps^2$ times per epoch). It only keeps experts that are $\eps_r$-better than the $r$-th thread {\sc RandomExpert}, and an expert is kept for $\eps \sqrt{n}$ epochs (so $\sqrt{n}/\eps$ days in total).
A {\sc MonocarpicExpert} procedure is executed on $\mP$.
We note an expected number of $O(1/\eps_r^2) \cdot (\eps_r^2/\eps^2) = O(1/\eps^2)$ experts are added to thread $r$ pool per epoch $\times$ they stay there for $\eps \sqrt{n}$ epochs, so the memory is bounded by $O(\sqrt{n}/\eps)$.
The regret analysis becomes a bit more complicated and we only sketch the intuition here.
If an expert is $\Theta(\eps_r)$ better than the $\eps \sqrt{n}/\eps_r^2$-th expert, then it incurs an average of $\eps_r$ loss over $1/\eps_r^2$ days. At the same time, with probability $1/\eps\sqrt{n}$, it is sampled and kept in the sub-pool $\mP_{r}$ for $\sqrt{n}/\eps$ days. Hence it has $(1/\eps_r) \cdot (\eps\sqrt{n}) / (\sqrt{n}/\eps) \leq \eps$ regret on average. 
The complication  arises from the expert might not be $\Omega(\eps_r)$ worse in the first $1/\eps_r^2$ days, in that case, the {\sc RandomExpert} procedure only guarantees $O(\eps_r)$ (instead of $O(\eps)$) regret and we need to carefully attribute the epoch in the analysis. We leave the technical details to Section \ref{sec:adaptive-upper}.

\begin{remark}[Beyond black-box reduction via differential privacy]
	%Provided with our algorithm against an oblivious adversary, 
	One may tempt to obtain an algorithm against adaptive adversary using the recent developed technique of differential privacy (DP) \cite{hasidim2020adversarially,beimel2022dynamic} on top of our oblivious algorithm.
	%make our oblivious algorithm work against an adaptive adversary using the recent developed technique of differential privacy (DP) \cite{hasidim2020adversarially,beimel2022dynamic}.
	This approach does not apply to our main setting because the algorithm is required to output {\em the actual action}, instead of a single binary prediction value which can be hidden via a DP median mechanism.
	This seems to be an inherent limit of the DP approach and it is much more challenging to hide the actual action, which is critical for broader applications of the online learning algorithm.\footnote{For example, in the application of game theory, the opponent observes the actual action; in many applications of machine learning, the opponent observes the hypothesis outputs by the algorithm }
	
	Meanwhile, if we narrow down the online learning task and only require to output a binary prediction (e.g., the task is to predict whether it is going to rain or not, and the expert offers 0/1 advice), then the DP technique is indeed applicable, but only obtains a sub-optimal regret bound of $\wt{O}(\frac{n^{1/4}T}{\sqrt{S}})$.  
	We provide a short explanation here. 
	
	The typical setup of the DP approach is to maintain $K$ copies of the oblivious algorithm, each uses $S'$ space. The total space is $S = KS'$ and we need to balance $K$ and $S'$. The typical choice is to take $K = O(\sqrt{T})$ and apply a DP median mechanism each round. However, there is substantial difference between online learning and the typical streaming applications. In a typical streaming application, each copy is guaranteed to be correct with high probability so taking median is a good idea, while for online learning, there is no such guarantee and it is easy to construct an example showing the median approach suffers $\Omega(1)$ regret. Instead, we should apply a DP mean mechanism and choose $K = \sqrt{T}/ \eps$ to obtain $\eps$ regret on average. Balancing these terms, the total regret is of order $\wt{\Theta}(\frac{n^{1/4} T}{\sqrt{S}})$, which is worse than our Theorem \ref{thm:adaptive-upper-main}.
\end{remark}

%\newpage

\subsection{Organization of the paper}
Section \ref{sec:pre} provides formal definitions and introduces the background. We provide the algorithm for oblivious adversary in Section \ref{sec:oblivious}. Section \ref{sec:adaptive-upper} and Section \ref{sec:adaptive-lower} provide the upper and lower bound for adaptive adversary.

%% file: discussion.tex
\subsection{Discussion}
\label{sec:future}

In this work we study the regret-memory tradeoffs for the experts problem which is a fundamental problem on its own. 
The experts problem is also used as a sub-routine in numerous applications in game theory \cite{freund1999adaptive}, machine learning \cite{klivans2017learning} and optimization \cite{plotkin1995fast}, it is interesting to see if our results lead to sub-space algorithms for any of those applications.

In many settings there are further structures over the set of experts (e.g. bounded Littlestone dimension) or the loss vector has succinct representation. These structural assumptions open opportunities of better memory-regret trade-off, even beyond our worst-case bounds. Understanding the memory-regret tradeoffs of such important special cases of the experts problem is an exciting direction for future work.\footnote{Note that even with structural assumptions that allow storing the payouts on any day succinctly, naively remembering those payouts over $T$ days would not be space efficient for large $T$.}

%% file: preliminary.tex
\section{Preliminary}
\label{sec:pre}

We consider the standard model of online learning with expert advice.

\begin{definition}[Online learning with expert advice]
In an online learning problem, there are $n$ experts and the algorithm is initiated with a memory state $M_0$. At the $t$-th day ($t\in [T]$),
\begin{itemize}
\item The algorithm picks an expert $i_t \in [n]$ bases on its internal state $M_{t-1}$;
\item The nature reveals the loss vector $\ell_{t} \in [0,1]^n$;
\item The algorithm receives loss $\ell_t(i_t)$, observes the loss vector $\ell_t$ and update its memory state from $M_{t-1}$ to $M_{t}$.
\end{itemize}
The loss vector is chosen adversarially:
\begin{itemize}
\item{\textbf{Oblivious adversary.}} An oblivious adversary (randomly) chooses the loss sequence $\ell_1, \ldots, \ell_T$ in advance.
\item{\textbf{Adaptive adversary.}} An adaptive adversary chooses the loss $\ell_t$ at the beginning of day $t$ and it could depend on the past action $i_1, \ldots, i_{t-1}$.
\end{itemize}
An algorithm uses at most $S$ words of memory if $\max_{t\in [T]}|M_t| \leq S$.
\end{definition}

The performance of the algorithm is measured as regret, which is defined as
\[
R(T) := \E\left[\sum_{t=1}^{T}\ell_t(i_t) - \min_{i^{*} \in [n]} \sum_{t=1}^{T}\ell_{t}(i^{*})\right]
\]
where the expectation is taken over the randomness of algorithm and the choice of $\ell_t$. The average regret is defined as $R(T)/T$.

\paragraph{Notation} In an online learning problem, let $T$ be the total number of days, $n$ be the number of experts. Let $[n] = \{1,2, \ldots, n\}$ and $[n_1:n_2] = \{n_1,\ldots, n_2\}$. We sometimes use $\N$ to denote the set of experts. For any integer $t$, let $\mathsf{pw}(t)$ be the largest integer such that $t$ is a multiple of $2^{\mathsf{pw}(t)}$. Let $\Delta_n$ contain all probability distributions over $[n]$.

\paragraph{Regret guarantee of existing online learning algorithm}

We make use of existing algorithms in online learning literature. The multiplicative weight update (see Algorithm \ref{algo:mwu}) is the most fundamental algorithm.

\begin{algorithm}[!htbp]
\caption{MWU \cite{arora2012multiplicative}}
\label{algo:mwu}
\begin{algorithmic}[1] 
\For{$t=1,2, \ldots, T$}
\State Compute $p_t \in \Delta_n$ over experts such that $p_t(i) \propto \exp(-\eta \sum_{\tau=1}^{t-1}\ell_\tau(i))$ for $i\in [n]$
\State Sample an expert $i_t \sim p_t$ and observe the loss vector $\ell_t \in [0, 1]^n$
\EndFor
\end{algorithmic}
\end{algorithm}

\begin{lemma}[\cite{arora2012multiplicative}]
Let $n, T \geq 1$, $\eta > 0$ and the loss $\ell_t \in [0,1]^n$ ($t\in [T]$). The MWU algorithm guarantees
\begin{align*}
    \sum_{t=1}^{T} \langle p_t, \ell_t\rangle - \min_{i^{*}\in [n]}\sum_{t=1}^{T}\ell_t(i^{*}) \leq \frac{\log n}{\eta} + \eta T,
\end{align*}
Taking $\eta = \sqrt{\log (n/\delta)/T}$, with probability at least $1-\delta$,
\begin{align*}
    \sum_{t=1}^{T}  \ell_t(i_t) - \min_{i^{*}\in [n]}\sum_{t=1}^{T}\ell_t(i^{*}) \leq O\left(\sqrt{T \log (n/\delta)}\right).
\end{align*}
\end{lemma}

Define $v_{t}(i) = \langle p_t, l_t\rangle - \ell_t(i)$ be the loss difference between expert $i$ and the algorithm, the {\sc Squint} algorithm (See Algorithm \ref{algo:squint}) provides second order guarantee.
\begin{algorithm}[!htbp]
\caption{{\sc Squint} \cite{koolen2015second}}
\label{algo:squint}
\begin{algorithmic}[1] 
\For{$t=1,2, \ldots, T$}
\State Compute $p_t \in \Delta_n$ over experts such that 
%\begin{align*}
$p_t(i)  \sim \E_{\eta}\left[\eta \cdot \exp\left(\eta \sum_{\tau=1}^{t-1}v_{\tau}(i) - \eta^2 \sum_{\tau=1}^{t-1}v_{\tau}^2(i)\right)\right] $
%\end{align*}
\State Sample an expert $i_t \sim p_t$ and observe the loss vector $\ell_t \in [0, 1]^n$
\EndFor
\end{algorithmic}
\end{algorithm}

\begin{lemma}[Theorem 3 of \cite{koolen2015second}]
\label{lem:squint}
Suppose the learning rate $\eta$ is sampled from the prior distribution of $\gamma$ over $[0, 1/2]$ such that $\gamma(\eta) \propto \frac{1}{\eta\log^2(\eta)}$. For any expert $i \in [n]$, the {\sc Squint} guarantees that
\begin{align*}
\E\left[\sum_{t=1}^{T}\ell_t(i_t) - \sum_{t=1}^{T}\ell_{t}(i)\right] \leq O\left(\sqrt{V_T^i \log (nT)}\right).
\end{align*}
where $V_T^i = \sum_{t=1}^{T}(\E_{i_t}[\ell_t(i_t)] - \ell_t(i))^2$ is the loss variance of expert $i$.
\end{lemma}

\paragraph{Information theory} We use standard notation from information. Let $X, Y$ be random variables, $H(X)$ be the entropy and $I(X; Y)$ be the mutual information. In a two-party communication problem, the internal information cost is defined as follows.
\begin{definition}[Internal information]
The internal information cost of a communication protocol $\Pi$ over inputs drawn from $\mu$ on $X\times Y$ is defined as
\[
\IC_{\mu}(\Pi) := I(X; \Pi |Y) + I(Y; \Pi |X).
\]
\end{definition}

%% file: oblivious.tex
\section{Near optimal bound for oblivious adversary}
\label{sec:oblivious}

We have the following regret guarantee for an oblivious adversary.
\Oblivious*

\subsection{The baseline-subroutine building block}
\label{sec:base}

We first provide a baseline subroutine that achieves the following regret guarantee.
\begin{proposition}
\label{prop:oblivious-base}
Let $n , T\geq 1$ be sufficiently large integers, for any $B \in [T]$, there is an online learning algorithm that uses at most $\polylog(nT)$ space and achieves $\widetilde{O}\left(nB + \frac{T}{\sqrt{B}}\right)$ regret against an oblivious adversary.
\end{proposition}

The pseudocode of the baseline subroutine is presented in Algorithms \ref{algo:base}--\ref{algo:filter}.
The \textsc{Baseline} (i.e. Algorithm \ref{algo:base}) proceeds epoch by epoch, and each epoch contains $B$ days.
It maintains a pool $\mP$ of experts and at the beginning of each epoch, it adds $O(1)$ random experts into the pool $\mP$ (Line \ref{line:base_new} of Algorithm \ref{algo:base}). 
It runs MWU over experts in $\mP$ within the epoch.
The pool size grows quickly and the algorithm prunes the pool at the end of each epoch.

The entire pool $\mP$ is divided into $L+1 =  \log_2(T)  + 1$ disjoint sub-pools $\mP = \mP_{0}\cup\cdots \cup \mP_{L}$. Roughly speaking, $\mP_{\ell}$ only contains experts sampled in the most recent $2^{\ell}$ epochs and it is merged into $\mP_{\ell+1}$ every $2^{\ell}$ epochs.

The \textsc{Merge} procedure (Algorithm \ref{algo:merge}) merges the sub-pool $\mP_{\ell}$ into $\mP_{\ell+1}$ and performs subsequent pruning step to control the pool size.
In particular, it proceeds in $K = 16  \log(nT) $ iterations, and at the $k$-th iteration, whenever the sub-pool size is greater than a threshold of (roughly) $\log^9 (nT)$, it samples $\frac{1}{\log^4 (nT)}$-fraction of experts $\mF_{k}$ and use them to filter $\mP_{\ell+1}$. 
Looking ahead, we shall prove that each iteration removes at least $\Omega(\frac{1}{\log (nT)})$-fraction of experts in $\mP_{\ell+1}$, unless  $\mP_{\ell+1}$ is already small.
We note that both the size estimation and the sampling step (i.e.~Algorithms \ref{algo:estimate} and \ref{algo:sample})  are performed in an {\em insensitive} way, which is critical to the memory analysis.

The core idea lies in the {\sc Filter} procedure (Algorithm \ref{algo:filter}), it takes a set of experts $\mF$ and a pool $\mQ$, and an expert $i$ is removed from $\mQ$ if it is {\em covered} by experts in $\mF$ (see Definition \ref{def:cover}).

\paragraph{Notation} To formally introduce the concept of covering, we need some notations. 
For any epoch $t\in [T/B]$, let $\mP_{\ell}^{(t)}$ ($\ell \in [0: L]$) be the sub-pool $\mP_{\ell}$ at the beginning of epoch $t$ (after Line \ref{line:base_new} of Algorithm \ref{algo:base}) and 
similarly let  $\mP^{(t)}$ denote the entire pool at the same time.
Our algorithm samples experts with replacement; even if it samples the same expert twice (in different epochs), we will treat it as two separate experts. 
%An expert $i \in \mP_{t}$ is uniquely identified by a tuple of $(i, E_{i}) \in [n]\times [T/B]$, where $E_{i} \leq t$ is the entering time of $i$. 
For expert $i \in \mP^{(t)}$, let $E_i \in [T/B]$ denote the time when expert $i$ was sampled and entered the pool. 
For experts $i, j \in \mP^{(t)}$, we write $i \preceq j$ if $E_{i} \leq E_{j}$, i.e., the expert $i$ enters the pool no later than the expert $j$. 
Let $\Gamma_{i, j}$ be the time interval $\{E_{i}, \ldots, E_{j} - 1\}$ ($\Gamma_{i, j} = \emptyset$ if $j$ enters the pool no later than $i$), and slightly abuse of notation, let $\Gamma_{i, t}$ be the time interval $\{E_i, \ldots,  t\}$.
For any interval $\mI \subseteq [T/B]$, the {\em cumulative} loss of expert $i$ over $\mI$ is defined as
\begin{align*}
\mL_{\mathcal{I}}(i) := \sum_{t\in \mathcal{I}}\sum_{b=1}^{B}\ell_{(t-1)B +b}(i).
\end{align*}
In particular, define
\begin{align*}
\mL_{\Gamma_{i, j}}(i) := \left\{
\begin{matrix}
\sum_{t\in \Gamma_{i, j}}\sum_{b=1}^{B}\ell_{(t-1)B +b}(i) & \Gamma_{i, j} \neq \emptyset \\
+\infty & \Gamma_{i, j} = \emptyset
\end{matrix}
\right.
\end{align*}
be the cumulative loss of expert $i$ during interval $\Gamma_{i, j}$. Slightly abusing notation, we set $\mL_{t}(i) = \sum_{b=1}^{B}\ell_{(t-1)B +b}(i)$ to be the total loss of expert $i$ during epoch $t$.

Now we can formally define the {\em cover}.
\begin{definition}[Cover]
\label{def:cover}
Given a set of experts $\mF = \{i_1, \ldots, i_{|\mF|}\}$ and an expert $j\in [n]$. Let $i_1 \preceq i_2 \cdots \preceq i_r \preceq j \preceq i_{r+1} \cdots \preceq i_{|\mF|}$ for some $r \in \{0, 1,2,, \ldots, |\mF|\}$.
At epoch $t$, we say expert $j$ is covered by $\mF$, if 
\begin{align}
\mathcal{L}_{\Gamma_{j,t}}(j) \geq  \min_{i_{r}^{*} \in \mF}\mL_{\Gamma_{j, i_{r+1}}}(i_r^{*}) + \left(\sum_{s = r+1}^{|\mF|}\min_{i_{s}^{*} \in \mF}\mL_{\Gamma_{i_{s}, i_{s+1}}}(i_s^{*})\right)  - 1.
\end{align}
Here we slightly abuse notation and denote $\Gamma_{i_{|\mF|}, i_{|\mF|+1}} = \Gamma_{i_{|\mF|}, t}$.
\end{definition}

Finally, it is worth noting that the {\sc Baseline} only needs to record the loss $\mathcal{L}_{\Gamma_{i, j}}(i)$ for $i, j \in \mP$, which scales quadratically with the size of pool.

\begin{algorithm}[!htbp]
\caption{\textsc{Baseline}}
\label{algo:base}
\begin{algorithmic}[1]
\State Initiate the pool $\mP \leftarrow \emptyset$ and sub-pools $\mP_{\ell} \leftarrow \emptyset$ ($\ell \in [0:L]$) \Comment{$L = \log_2 (T)$}
\For{$t = 1,2, \ldots, T/B$} \Comment{Epoch $t$}
\State $\mP_0 \leftarrow \textsc{Sample}(\N, 1/n)$ \Comment{Sample new experts into the pool}\label{line:base_new}
\State Run MWU over $\mP$ for the next $B$ days \Comment{$\mP = \mP_0 \cup \cdots \cup \mP_{L}$}
\For{$\ell = 0,1, \ldots, \mathsf{pw}(t)$} \Comment{$\mathsf{pw}(t)$ is the largest integer such that $(t/2^{\mathsf{pw}(t)}) \in \mathbb{N}$}
\State $\mP_{\ell+1} \leftarrow \textsc{Merge}(\mP_{\ell+1}, \mP_{\ell})$
\State $\mP_{\ell} \leftarrow \emptyset$
\EndFor 
\EndFor
\end{algorithmic}
\end{algorithm}

\begin{algorithm}[!htbp]
\caption{\textsc{Merge}($\mQ_A, \mQ_B$)}
\label{algo:merge}
\begin{algorithmic}[1]
\State $\mQ_C\leftarrow \mQ_A \cup \mQ_B$
\For{$k=1,2, \ldots, K $} \Comment{$K = 16  \log(nT) $}
\State $s \leftarrow \textsc{EstimateSize}(\mQ_C, \frac{1}{\log^4 (nT)})$   \Comment{Estimate the size of $\mQ_C$}\label{line:merge-estimate}
\If{$s \leq \log^5 (nT)$}
\State Go to Line \ref{line:merge-return} \label{line:threshold}
\Else
\State $\mF_k \leftarrow \textsc{Sample}(\mQ_C, \frac{1}{\log^4 (nT)})$ \Comment{Sample $\frac{1}{\log^4 (nT)}$-fraction of experts}\label{line:merge-sample}
\State $\mQ_C \leftarrow \textsc{Filter}(\mF_{k},\mQ_C) \cup \mF_{k}$
\EndIf
\EndFor
\State Return $\mQ_C$ \label{line:merge-return}
\end{algorithmic}
\end{algorithm}

\begin{algorithm}[!htbp]
\caption{\textsc{EstimateSize}($\mQ, p$)}
\label{algo:estimate}
\begin{algorithmic}[1]
\State $p_{i} \leftarrow \text{Bernoulli}(p)$ for every $i \in \mQ$
\State Return $\sum_{i \in \mQ}p_i$
\end{algorithmic}
\end{algorithm}

\begin{algorithm}[!htbp]
\caption{\textsc{Sample}($\mQ, p$)}
\label{algo:sample}
\begin{algorithmic}[1]
\State $p_{i} \leftarrow \text{Bernoulli}(p)$ for every $i \in \mQ$
\State Return $\{i\in \mQ: p_{i} = 1\}$
\end{algorithmic}
\end{algorithm}

\begin{algorithm}[!htbp]
\caption{\textsc{Filter}($\mF, \mQ$)}
\label{algo:filter}
\begin{algorithmic}[1]
\For{each expert $i \in \mQ$}
\If{$\mF$ covers $i$} \Comment{Definition \ref{def:cover}}
\State $\mQ \leftarrow \mQ \setminus \{i\}$
\EndIf
\EndFor
\State Return $\mQ$
\end{algorithmic}
\end{algorithm}

\subsubsection{Memory bound}

We first provide the memory guarantee for \textsc{Baseline}.
\begin{lemma}[Memory bounds for {\sc Baseline}]
\label{lem:memory-baseline}
With probability at least $1-1/(nT)^{\omega(1)}$, the total memory used by Algorithm \ref{algo:base} is at most $O(\log^{20} (nT))$ words.
\end{lemma}

The key observation comes from the following guarantee of \textsc{Merge} procedure.
\begin{lemma}
\label{lem:merge_size}
With probability at least $1-1/(nT)^{\omega(1)}$, the output $\mQ_{C}$ of {\sc Merge} procedure satisfies \begin{align*}
|\mQ_{C}| \leq \max\left\{2\log^9 (nT), \frac{1}{4}(|\mQ_A| + |\mQ_{B}|)\right\}.
\end{align*}
\end{lemma}
\begin{proof}
Suppose the {\sc Filter} procedure has been called for $k_{\max} \in [0:K]$ times and let $\mQ_{C,k}$ be the status of pool $\mQ_{C}$ at the beginning of $k$-th loop.
Note that $|\mQ_{C,1}| = |\mQ_{A}| + |\mQ_{B}|$. 
For any $k \in [k_{\max}]$, our goal is to prove, with probability at least $1-1/(nT)^{\omega(1)}$,
\begin{enumerate}
\item Either $|\mQ_{C, k}| \leq 2 \log^9 (nT)$; or
\item $|\mQ_{C, k+1}| \leq (1 - \frac{1}{6\log (nT)})|\mQ_{C,k}|$.
\end{enumerate}
Suppose we have proved the claim, then taking a union bound over $k \in [k_{\max}]$, it guarantees that 
\begin{align*}
|\mQ_{C,k_{\max}+1}| \leq &~ \max\left\{ 2\log^9 (nT), \left(1 - \frac{1}{6\log (nT)}\right)^{16\log(nT)} |\mQ_{C,1}|\right\}\\
\leq &~ \max\left\{ 2\log^9 (nT), \frac{1}{4}(|\mQ_A| + |\mQ_{B}|)\right\}.
\end{align*}

To prove the claim, fix an iteration $k \in [k_{\max}]$ and assume $|\mQ_{C,k}| > 2\log^9 (nT)$, it suffices to prove $|\mQ_{C, k+1}| \leq (1 - \frac{1}{6\log (nT)})|\mQ_{C, k}|$ holds with high probability.
First, it passes the threshold test (Line \ref{line:threshold} of {\sc Merge}) with high probability, since
\begin{align*}
\Pr[s \leq \log^{5}(nT)] = \Pr\left[\sum_{i\in \mQ_{C,k}}p_i \leq \log^{5}(nT)\right] \leq \exp(-\log^4(nT)/4) = 1/(nT)^{\omega(1)}.
\end{align*}
The second step follows from the Chernoff bound and the fact that (1) $\{p_i\}_{i \in \mQ_{C,k}}$ are independent Bernoulli variables with mean $\frac{1}{\log^4(nT)}$; (2) $|\mQ_{C, k}| \geq 2\log^{9}(nT)$.

We next analyse the {\sc Filter} procedure and prove with high probability (over the choice of $\mF_k$), at least $\frac{1}{6\log nT}$-fraction of experts would be removed.  We prove that this set of experts exists iteratively. 
\begin{itemize}
    \item At the $r$-th iteration of the analysis we identify a subset $R_r$ that are removed with high probability over the choice of $\mF_k$. If $R_r$ is not too small ($\Omega(1/\log(nT))$-fraction of the sub-pool), we are guaranteed to reduce the size of the sub-pool so we're done.
    \item Otherwise, we identify $r$ pivot experts of $\mF_k$ and a large ($1-O(r/\log(nT))$-fraction of the subset $\mW_{r+1}$ whose loss is quite large. Specifically, consider the suffix of the epochs spanned by the $r$ experts in $\mF_k$. Any expert in $\mW_{r+1}$ is covered by those $r$ experts over this suffix --- and furthermore this continues to hold even if we give the experts in $\mW_{r+1}$ an additive $2^{r-1}$ loss reduction.  This cannot continue to hold if the loss reduction exceeds the total possible loss of $T$, hence the previous case of large $R_r$ must hold for some $r \le \log(nT)$.
\end{itemize}

\paragraph{Setting up notation for the analysis:} Let $M = |\mQ_{C, k}|$. For any $r \geq 1$, we iteratively define 
\begin{align*}
\mW_{r+1} := \mW_{r} \backslash (R_{r} \cup O_{r} \cup D_{r}(j_{r})) \quad \text{and} \quad \mW_{1} = \mQ_{C,k}.
\end{align*}
The definitions of $j_r, R_{r}, O_{r}$ and $D_{r}(j_{r})$ are presented subsequently.

We first define the set $D_r(i)$ for every expert $i \in \mW_r$, let
\begin{align*}
D_r(i) := \{j \in \mW_r \setminus \{i\}, \mL_{\Gamma_{i, j_{r-1}}}(i) \geq \mL_{\Gamma_{i, j_{r-1}}}(j) - 2^{r-1} \}, \forall i \in \mW_{r}.
\end{align*}
Here we slightly abuse notation and take $\Gamma_{i, j_0} = \Gamma_{i,t}$.
In other words, the set $D_r(i)$ contains experts that are comparable with expert $i$ over the interval $\Gamma_{i, j_{r-1}}$.

We next define the set $R_r$, let
\begin{align*}
R_r = \{i: i \in \mW_r, |D_r(i)| \geq  \log^6(nT)\}.
\end{align*}
The set $R_r$ contains every expert $i \in \mW_r$ that has a large size $D_r(i)$; i.e. we shall argue these are the experts that are easily removed because they are covered by the combination of $j_1, \ldots, j_{r-1}$ and any single expert from $D_r$.

Now, define the expert $j_r$ as the latest expert in $\mW_{r}\backslash R_{r}$ that is also sampled into $\mF_{k}$, i.e.,
\begin{align*}
j_r := \argmax_{j\in (\mW_r\backslash R_r) \cap \mF_{k}} E_{j}.
\end{align*}
Note if no such expert exists, i.e.~$(\mW_r\backslash R_r) \cap \mF_{k} = \emptyset$, then we leave it undefined.

The set $O_r$ includes experts enter the pool later than $j_r$, i.e.
\begin{align*}
O_r = \{i: i \in \mW_r\backslash R_r, E_{i} \geq E_{j_r} \}.
\end{align*}

\paragraph{Inductive hypothesis.}
Let $r \geq 1$, our inductive hypothesis conditions on the event that $R_{m} \leq \frac{1}{4\log(nT)}M$ for all $m \in [r]$. Looking ahead, these are the easy cases because the set of experts $R_m$ are easily removed.
We inductively prove 
\begin{enumerate}
\item $|\mW_{r+1}| \geq (1 - \frac{r}{2\log (nT)})M$;
\item The experts $j_{1}, \ldots, j_{r} \in \mF_{k}$ and follow the order of $j_{r} \preceq j_{r-1} \preceq \cdots \preceq j_1$ ; 
\item For any $m \in [r]$ and expert $i \in \mW_{r+1}$,
\begin{align}
\mL_{\Gamma_{j_{m}, j_{m-1}}}(i) \geq \mL_{\Gamma_{j_{m}, j_{m-1}}}(j_m) + 2^{m-1}.\label{eq:ind1}
\end{align}
\end{enumerate}
We prove the claims by induction and we note the base case of $r=0$ holds trivially. Suppose the induction holds up to $r-1$.

\paragraph{The analysis of $r$-th iteration.}
 We now consider two cases: in the first case, the set of easily removable experts $R_r$ is large, so we can just remove them!
In the second case, while it is hard to cover the typical expert with a single expert, we could continue the induction and prove the typical expert can eventually be covered by pooling together several experts.

\vspace{+2mm}
\textbf{Case 1.} If $|R_r| \geq \frac{1}{4\log (nT)}M$. We prove with high probability that experts in $R_r$ are removed by the {\sc Filter} procedure.
For any expert $i \in R_r$, each expert $j \in D_r(i)$ is included into $\mF_k$ with (independent) probability $\frac{1}{\log^4(nT)}$, hence, 
\begin{align*}
\Pr\left[j \notin \mF_{k}, \forall \, j \in D_r(i)\right] \leq \left(1 - \frac{1}{\log^4 (nT)}\right)^{\log^6 (nT)} \leq 1/(nT)^{\omega(1)}.
\end{align*}
The first step follows from $|D_r(i)| \geq \log^{6}(nT)$ holds for any $i \in R_r$.

Therefore, with probability at least $1 - 1/(nT)^{\omega(1)}$, there exists at least one expert $j \in D_r(i)\cap \mF_k$. We claim that the expert $i$ is covered by $j, j_{r-1}, j_{r-2}, \ldots, j_{1}$. In particular, we have
\begin{align*}
\mL_{\Gamma_{i,t}}(i) = &~  \sum_{m=1}^{r-1}\mL_{\Gamma_{j_m, j_{m-1}}}(i) + \mL_{\Gamma_{i, j_{r-1}}}(i)\\
\geq &~ \sum_{m=1}^{r-1}(\mL_{\Gamma_{j_m, j_{m-1}}}(j_{m}) + 2^{m-1}) + \mL_{\Gamma_{i, j_{r-1}}}(j) - 2^{r-1}\\
= &~ \sum_{m=1}^{r-1}\mL_{\Gamma_{j_m, j_{m-1}}}(j_{m}) + \mL_{\Gamma_{i, j_{r-1}}}(j) -1,
\end{align*}
where we split the loss of $i$ in the first step, the second step holds due to the inductive hypothesis (see Eq.~\eqref{eq:ind1}) and $j \in \D_r(i)$.

Hence, with probability at least $1 - 1/(nT)^{\omega(1)}$, the set of experts $R_r\setminus \mF_{k}$ are removed. Note that $|R_r| \geq \frac{1}{4\log (nT)}M$ and the size of $\mF_k$ satisfies
\begin{align*}
\Pr\left[|\mF_{k}| \geq \frac{2}{\log^4(nT)} M \right] \leq \exp(-M/4\log^4(nT)) \leq 1/(nT)^{\omega(1)},
\end{align*}
where the first step holds due to Chernoff bound.

Taking a union bound, we have $|R_r\setminus \mF_{k}| \geq \frac{1}{6\log (nT)}M$ with probability at least $1 - 1/(nT)^{\omega(1)}$ and we have already identified a large set to be removed.

\vspace{+2mm}
\textbf{Case 2.} If $|R_r| < \frac{1}{4\log (nT)}M$.
Recall $j_{r}$ is defined as the latest expert of $\mW_{r}\setminus R_{r}$ that is sampled into $\mF_{k}$. 
With probability at least $1-1/(nT)^{\omega(1)}$, at most $\log^6 (nT)$ experts arrive later than $j_{r}$, i.e., $|O_r|\leq \log^6 (nT)$. Recall the set $\mW_{r+1}$ is recursively defined as $\mW_{r+1} = \mW_{r} \setminus (R_{r}\cup O_r \cup D_r(j_{r}))$. We prove our inductive hypothesis continues to hold.
First, the size of $\mW_{r+1}$ satisfies
\begin{align*}
|\mW_{r+1}| \geq &~ |\mW_{r}| - |R_r| - |O_r| - |D_r(j_r)| \\
\geq &~ \left(1 - \frac{r-1}{2\log (nT)}\right) M - \frac{1}{4\log (nT)}M - \log^6 (nT) - \log^6 (nT) \\
\geq &~ \left(1- \frac{r}{2\log (nT)}\right)M.
\end{align*}
Second, it is clear that $j_r \in \mF_k$ by definition and it comes earlier than $j_{r-1}, \ldots, j_1$.
The third claim holds since $D_{r}(j_r) \cap \mW_{r+1} = \emptyset$, we have proved the induction.

Finally, if the event $R_{r} < \frac{1}{4\log(nT)} M$ continues to hold to $r =\log(nT) + 1$, then one has $D_{r}(i) = \mW_r$ for every expert $i \in \mW_{r}$, since the loss is at most $T$. Consequently, one as $|R_r| = |\mW_{r}| \geq \frac{1}{2}M$. 
We conclude the proof here.
\end{proof}

Now we can bound the pool size. 
\begin{lemma}[Memory bounds for sub-pool]
\label{lem:pool_size}
For any epoch $t\in [T/B]$ and sub-pool $\ell \in [0: L]$, with probability at least $1-1/(nT)^{\omega(1)}$, one has $|\mP_{\ell}^{(t)}| \leq 2\log^9 (nT)$.
\end{lemma}
\begin{proof}
We prove by induction on $\ell$. For the base case, at any epoch $t$, the sub-pool $\mP_{0}^{(t)}$ consists of the experts sampled at the beginning of epoch $t$ and one has
\begin{align*}
\Pr[|\mP_{0}^{(t)}| \geq 2\log^9(nT)] \leq \exp(- \log^9(nT)/3) \leq 1/(nT)^{\omega(1)},
\end{align*}
this comes from Chernoff bound and the fact that each expert is sampled with probability $1/n$.

Suppose the induction holds up to $\ell-1$ ($\ell \geq 1$), the sub-pool $\mP_{\ell}$ is empty at beginning, and it is merged with $\mP_{\ell-1}$ every $2^{\ell-1}$ epochs. That is to say, the input sub-pools ($\mP_{\ell}$ and $\mP_{\ell-1}$) of {\sc Merge} both have size no more than $2\log^9(nT)$, hence by Lemma \ref{lem:merge_size}, with probability at least $1-1/(nT)^{\omega(1)}$, the output sub-pool $|\mP_{\ell}^{(t)}| \leq 2\log^9 (nT)$.
\end{proof}

Now we are ready to prove Lemma \ref{lem:memory-baseline}.
\begin{proof}[Proof of Lemma \ref{lem:memory}]
By Lemma \ref{lem:pool_size}, at epoch $t$, the size of the pool is bounded by:
\[
|\mP^{(t)}| \leq \sum_{\ell=0}^{L}|\mP_{\ell}^{(t)}| \leq 2\log^9 (nT) \cdot 2\log(nT) = O(\log^{10} (nT)).
\]
To execute the {\sc Filter} procedure,  it suffices to record loss of expert $i \in \mP^{(t)}$ over interval $\Gamma_{i, j}$ for any $i, j \in \mP^{(t)}$, it takes $O(\log^{10} (nT))$ words per expert and up to $O(\log^{20} (nT))$ words of memory in total. We have finished the memory analysis of {\sc Baseline}.
\end{proof}

\subsubsection{Regret analysis}
\label{sec:base-regret}

Let $i^{*}$  denote the best expert.
At a high level, the regret analysis considers two types of epochs, which can be informally described as follows:
\begin{itemize}
    \item An epoch is ``good'', if had we sampled  $i^{*}$ in this epoch, it will eventually be covered and kicked out of the pool. Consider the interval from the current good epoch until $i^*$ is kicked out of the pool: We argue that the algorithm's performance on the entire interval  is comparable to $i^*$ even if $i^*$ is never sampled.
    \item The set of ``bad'' epochs, denoted by $\mH$, are ones where $i^*$ would survive until the end of the stream.
\end{itemize}
Our plan is to argue that number of bad epochs must be small --- otherwise the pool would overflow with many copies of $i^*$; but we had already proved in Lemma~\ref{lem:pool_size_full} that the pool size is polylogarithmic.  For each bad epoch we have no control over the regret, but the total regret from bad epochs is small because there aren't many of them. (And as we explained above, on good epochs the algorithm is guaranteed to have low regret.)

There is a subtle issue with the above approach%
\footnote{The issue is subtle but fixing was perhaps the most difficult part in coming up with the algorithm and analyzing it.}:
We want to analyze the regret by considering a hypothetical situation that $i^*$ is sampled into the pool at a given epoch. However, sampling $i^*$ also interferes with the rest pool going forward because $i^*$ may kick out other experts; kicking out some experts may allow the pool to keep other experts that should have been kicked out, etc. This makes it very difficult to compare the pool the algorithm actually uses and the hypothetical pool with $i^*$.

To handle the above subtlety, we carefully designed the algorithm so that even conditioned on $i^*$ having been sampled into the pool, w.h.p.~it is never considered in the {\sc Sample} and {\sc EstimateSize} subroutines. Hence it does not interfere with the rest of the pool.

\paragraph{Random bits} 
It is useful to understand the random bits used for pool selection first.
%We formally define the probability space of {\sc Baseline}.
%Let $\mu$ be the probability measure over sample space $\Omega$, and an event $\xi = \{\xi_{\adv}\}\cup \{\xi_{t}\}_{t\in [T/B]} \in \Omega$ contains all random bits used by the adversary and the algorithm.
For each epoch $t \in [T/B]$, let $\xi_{t} = (\xi_{t, 1}, \ldots, \xi_{t, n})$, where $\xi_{t, i}$ is used for expert $i$ ($i \in [n]$).
The random bits $\xi_{t, i} = (\xi_{t, i,1}, \xi_{t, i,2})$ consists of two parts, where $\xi_{t, i, 1} \in \{0, 1\}$ is a Bernoulli variable with mean $1/n$, and it is used for sampling expert $i$ into the pool (i.e. Line \ref{line:base_new} of Algorithm \ref{algo:base}). {\sc Baseline} adds expert $i$ at epoch $t$ if and only if $\xi_{t, i,1} = 1$.
The second part of the random bits $\xi_{t, i,2} \in \{0, 1\}^{2L \times 2K}$ are used for {\sc EstimateSize} and {\sc Sample} when calling the {\sc Merge} procedure. Each coordinate of $\xi_{t, i,2}$ is a Bernoulli random variable with mean $\frac{1}{\log^4 (nT)}$ and the expert $i$ gets sampled for {\sc Sample} and {\sc EstimateSize} if the corresponding random bit equals $1$.

We introduce the concept of passive/active expert.
\begin{definition}[Active/Passive expert]
At any epoch $t\in [T/B]$, an expert $i$ is said to be {\em passive}, if $\xi_{t, i, 2} = \vec{0}$. It is said to be an {\em active} expert otherwise.
\end{definition}

\paragraph{Epoch assignment} 
Fix the loss sequence $\ell_1, \ldots, \ell_{T}$ and the set of {\em sampled and active} experts $Y_t \subseteq [n]$ of each epoch $t$ (note that the set of sampled and passive experts are still not determined yet).
The key observation is that the estimate size $s$, the filter set $\mF$ as well as the set of alive active experts are fixed at any time during the execution of {\sc Baseline}. They are completely determined by the loss sequence $\{\ell_t\}_{t\in [T]}$ and the sampled while active experts $\{Y_t\}_{t\in [T/B]}$.

We would divide the entire epochs into two types.  Consider the following epoch assignment mechanism. We start with $\tau \leftarrow 1$, $a_{1} \leftarrow 1$ and $\mH \leftarrow \emptyset$. Here the set $\mH \subseteq [T/B]$ collects all bad epochs that the algorithm has no controls over the regret.
Let $t(a_\tau) \in [a_{\tau}: T/B] \cup \{+\infty\}$ be the first time such that $i^{*}$ (with entering time $a_{\tau}$) is covered by the set of alive active experts at epoch $t(a_{\tau})$.
We note that the notion of covering is well-defined regardless of $i^{*}$ being sampled or not. We divide into two cases.
\begin{itemize}
\item Expert $i^{*}$ is never covered (i.e., $t(a_\tau) =+\infty$). Then we augment the set of bad epochs $\mH \leftarrow \mH \cup\{a_{\tau}\}$, and update $a_{\tau+1} = a_{\tau} + 1$, $\tau \leftarrow \tau + 1$;
\item Otherwise, update $a_{\tau + 1} \leftarrow t(a_{\tau}) + 1$ and $\tau \leftarrow \tau + 1$.
\end{itemize}

The assignment process terminates once $a_{\tau}$ reaches $\frac{T}{B}$ and let $\tau_{\max}$ be the total number of steps.

We decompose the regret of {\sc Baseline} based on the above assignment process and bound it in terms of the size of bad epochs $\mH$. 
\begin{lemma}
\label{lem:regret-decompose}
Fixing the loss sequence as well as the set of {\em sampled and active} experts of each epoch, with probability at least $1 - 1/(nT)^{\omega(1)}$ over the choice of MWU, one has
\begin{align*}
\sum_{t=1}^{T/B} \sum_{b=1}^{B} \ell_{(t-1)B + b}(i_{(t-1)B + b}) - \sum_{t=1}^{T/B}\mL_t(i^{*}) \leq O\left(\frac{T}{\sqrt{B}}\log (nT)\right)+ B \cdot |\mH|.
\end{align*}
\end{lemma}
\begin{proof}
Let $i_{t}^{*}$ be the best expert in the pool $\mP^{(t)}$ during epoch $t$, then with probability at least $1 - 1/(nT)^{\omega(1)}$, we have
\begin{align*}
&~ \sum_{t=1}^{T/B} \sum_{b=1}^{B} \ell_{(t-1)B + b}(i_{(t-1)B + b}) - \sum_{t=1}^{T/B}\mL_t(i^{*})\\
= &~ \sum_{t=1}^{T/B} \sum_{b=1}^{B} \ell_{(t-1)B + b}(i_{(t-1)B + b}) - \sum_{t=1}^{T/B} \mL_t(i^{*}_t) + \sum_{t=1}^{T/B} \mL_t(i^{*}_t) - \sum_{t=1}^{T/B}\mL_t(i^{*}) \\
\leq &~ \sum_{t=1}^{T/B}\mL_{t}(i_t^{*}) - \sum_{t=1}^{T/B} \mL_t(i^{*}) + \frac{T}{B}\cdot O\left(\sqrt{B}\log (nT)\right) \\
= &~ \sum_{t \in \mH}\mL_{t}(i_t^{*}) - \mL_t(i^{*}) + \sum_{\tau=1}^{\tau_{\max}}\sum_{t\in [a_{\tau}:a_{\tau+1}-1]\backslash \mH} \mL_{t}(i_t^{*}) - \mL_t(i^{*}) + O\left(\frac{T}{\sqrt{B}}\log (nT)\right)\\
\leq &~ B \cdot |\mH| + \tau_{\max} + O\left(\frac{T}{\sqrt{B}}\log (nT)\right)\\
\leq &~ B \cdot |\mH| + O\left(\frac{T}{\sqrt{B}}\log (nT)\right).
\end{align*}
The second step follows from the regret guarantee of MWU. We decompose the regret according to the assignment in the third step. The fourth step follows from the naive bound of $\mL_t(i_t) \leq B$ and the definition of a cover. The last step follows from $\tau_{\max} \leq \frac{T}{B}$. This finishes the proof.
\end{proof}

We next prove the size of $\mH$ is small with high probability.

\begin{lemma}
\label{lem:base-size}
With probability at least $1 - 1/(nT)^{\omega(1)}$, $|\mH| \leq O(n\log^{10}(nT))$.
\end{lemma}
\begin{proof}
We first fix the loss sequence $\ell_1, \ldots, \ell_{T}$ as well as the set of sampled and active experts $Y_t$ of each epoch $t$. 
We prove the size of pool $\mP$ at the end satisfies 
\[
\Pr\left[|\mP|\geq \frac{|\mH|}{4n} \mid Y_1, \ldots,Y_{T/B}, \ell_1, \ldots, \ell_T \right] \geq \frac{1}{2},
\]
whenever $|\mH| \geq \Omega(n)$. Here the probability is taken over the randomness of the remaining  {\em sampled and passive} experts. 

Let $W_t$ be the set of sampled and passive experts of epoch $t$. For any epoch $t \in \mH$, by the epoch assignment mechanism, the expert $i^{*} \notin Y_{t}$ (otherwise it is covered by itself at epoch $t$).
Condition on $\{Y_t\}_{t\in [T/B]}$, the probability that $i^{*} \in W_t$ obeys
\begin{align*}
\Pr[i^{*}\in W_t | Y_1, \ldots, Y_{T/B}, \ell_1, \ldots, \ell_T] = &~  \Pr[i^{*}\in Y_t | i^{*}\notin Y_{t}] \\
\geq &~ \Pr[i^{*}\in W_t] = \frac{1}{n} \cdot \left(1 - \frac{1}{\log^4 (n)}\right)^{2L \times 2K} \geq \frac{1}{2n}.
\end{align*}
The key observation is that once $i^{*} \in W_{t}$, i.e., $i^{*}$ is sampled while passive at epoch $t$, it would survive till the end. This is because the filter set $\mF$ and the estimate size $s$ are fully determined by the sampled and active experts $\{Y_t\}_{t\in [T/B]}$ and the loss sequence, which has already been fixed. 
Therefore, if a passive expert $i^{*}$ enters at epoch $t$, it survives till the end.

The event of $i^{*}\in W_t$ are independent for $t \in \mH$ (condition on the loss sequence and $\{Y_t\}_{t\in [T/B]}$), hence, by Chernoff bound, the pool $\mP$ at the end of epoch $T/B$ satisfies
\begin{align}
\Pr\left[|\mP|\leq \frac{|\mH|}{4n} \mid Y_1, \ldots, Y_{T/B}, \ell_1, \ldots, \ell_T \right] 
\leq &~ \Pr\left[\sum_{t\in \mH}\mathsf{1}\{i^{*}\in W_t\}\leq \frac{|\mH|}{4n} \mid Y_1, \ldots, Y_{T/B}, \ell_1, \ldots, \ell_T \right] \notag \\
\leq &~ \exp(-|\mH|/16n). \label{eq:bad-size-baseline}
\end{align}

Note we have already proved in Lemma \ref{lem:memory-baseline} that $|\mP| \geq  \Omega(\log^{10}(nT))$ happens with probability at most $1/(nT)^{\omega(1)}$, by Eq.~\eqref{eq:bad-size-baseline}, this implies 
\[
\Pr[|\mH|\geq \Omega(n\log^{10}(nT))] \leq 1/(nT)^{\omega(1)}.
\] 
We conclude the proof here.
\end{proof}

Combining Lemma \ref{lem:regret-decompose} and Lemma \ref{lem:base-size}, we obtain the regret bound of Proposition \ref{prop:oblivious-base}. The memory guarantee of Proposition \ref{prop:oblivious-base} has already been established in Lemma \ref{lem:memory-baseline}.

\subsection{Boosting}
\label{sec:bootstrap-optimal}

We achieve the optimal regret guarantee of Theorem \ref{thm:oblivious-main} by maintaining multiple threads of {\sc Baseline} and {\em amortizing} the regret carefully. 
The new algorithmic ingredient includes
(1) leveraging the pool selection from a high frequency run of {\sc Baseline} for a lower frequency run of {\sc Baseline}; and (2) a low memory monocarpic expert algorithm with interval regret guarantees.

We first obtain the optimal rate when the space $S$ is small.
\begin{proposition}
\label{prop:oblivious-n}
Let $n , T\geq 1$ be sufficiently large, there is an online learning algorithm that uses at most $\polylog(nT)$ space and achieves $\widetilde{O}\left(\sqrt{nT}\right)$ regret against an oblivious adversary.
\end{proposition}

%We first introduce the ``grouping'' trick, which utilizes $S$ space (instead of only polylogarithmic space as in Proposition \ref{prop:oblivious-base}). 
%We group experts into $n_s$ groups, with at most $G_{s} = \lceil \frac{S}{\log^{20}(nT)}\rceil$ experts per group. For each meta experts $i \in \mN_s$, it runs MWU over experts $\{(i-1)G_s+1, \ldots, i G_s \}$ once it is sampled by the full algorithm, starting from uniform weights. It terminates once it is removed from the pool. It is w.l.o.g. to view the meta expert as one expert in the original problem and we effectively reduce the input size from $n$ to $n_s$. There is one subtle difference, the loss of a meta expert is only well-defined if it is sampled -- we will see it is fine since our analysis only uses its loss when it is alive.

\paragraph{Full algorithm}
The full algorithm is presented as Algorithm \ref{algo:full}. WLOG, we assume $T/n$ is a tower of $2$.
Let $R = \log_2(T/n) $ be the total number of threads.
Let $B_{r} = \frac{T}{n 2^{r-1}}$ ($r \in [R]$), $T_1 = T$ and $T_{r} = B_{r-1}$ for $r\in [2:R]$.
The full algorithm maintains $R$ different threads of the $\textsc{Baseline}_{+}$ algorithm, where the $r$-th thread ($r \in [R]$) $\textsc{Baseline}_{+}(r)$ restarts every $T_{r}$ days, with epoch size $B_{r}$. 
The $\textsc{Baseline}_{+}$ differs slightly from $\textsc{Baseline}$ by (1) at the beginning of each epoch, in addition to sampling $O(1)$ experts, it also inherits experts from pools of higher threads; (2) it only maintains its pool (running {\sc Merge} procedure) but does not run MWU over the epoch. 
In other words, $\textsc{Baseline}_{+}(r)$ ($r\in [R]$) determines the ``alive'' experts, while the regret is controlled by a separate procedure {\sc MonocarpicExpert}.
Instead of the epoch-wise regret guarantee, the {\sc MonocarpicExpert} procedure guarantees an $\wt{O}(\sqrt{|\mI|})$ regret over {\em any} interval $\mI$ of the life time of an expert.

\begin{algorithm}[!htbp]
\caption{Full algorithm}
\label{algo:full}
\begin{algorithmic}[1]
\State Initialize $\textsc{Baseline}_{+}(r)$ ($\forall\, r\in [R]$) \Comment{$\textsc{Baseline}_{+}(r)$ maintains the pool $\mP_{r,\cdot}$}
\State $\mP \leftarrow \mP_{1,\cdot}\cup \cdots \cup \mP_{R,\cdot}$ \Comment{$\mP$ aggregates pools from all threads}
\For{$t = 1,2, \ldots, T$}
\State Run $\textsc{Baseline}_{+}(r)$ ($r\in [R]$) 
\State Run {\sc MonocarpicExpert} over $\mP$ 
\EndFor
\end{algorithmic}
\end{algorithm}

\begin{algorithm}[!htbp]
\caption{$\textsc{Baseline}_{+}(r)$ \Comment{$r \in [R]$}
%\Comment{$\mathsf{pw}(t)$ is the largest integer such that $t$ is a multiple of $2^{\mathsf{pw}(t)}$}\\
%\Comment{$\mathsf{lr}(r, s, t) := R - \mathsf{pw}((s-1)T_r + (t-1)B_r) \in [R]$ is the lowest level with a completed epoch at day $(s-1)T_r + (t-1)B_r$}
}
\label{algo:base+}
\begin{algorithmic}[1]
\For{$s = 1,2,\ldots, T/T_{r}$} \Comment{$s$-th restart}
\State Initiate the pool $\mP_{r,\cdot} \leftarrow \emptyset$ and sub-pools $\mP_{r, \ell} \leftarrow \emptyset$ ($i \in [0:L]$)
\For{$t = 1,2, \ldots, T_r/B_r$} \Comment{Epoch $t$} 
\State \texttt{// At the beginning of epoch $t$  } %$\triangleright$ At the beginning of epoch $t$ 
\If{$r$ is the lowest thread with a new epoch}\label{line:cont-merge1} %\Comment{$r = R - \mathsf{pw}((s-1)T_r + (t-1)B_r)$}
\For{$r' = R, \ldots, r+1$} 
\State $\mP_{r,0} \leftarrow \textsc{Merge}(\mP_{r, 0} \cup \mP_{r', \cdot}$) \Comment{Inherit from higher thread pools}
\EndFor 
\EndIf\label{line:cont-merge2}
\State $\mP_{r,0} \leftarrow \mP_{r,0} \cup \textsc{Sample}(\N, 1/n)$ \Comment{Sample new experts into pool} \label{line:base_new+}\\
\State \texttt{// At the end of epoch $t$  }%$\triangleright$ At the end of epoch $t$
\For{$\ell = 0,1, \ldots, \mathsf{pw}(t)$}
\State $\mP_{r,\ell+1} \leftarrow \textsc{Merge}(\mP_{r,\ell+1}, \mP_{r,\ell})$
\State $\mP_{r,\ell} \leftarrow \emptyset$
\EndFor 
\EndFor
\EndFor
\end{algorithmic}
\end{algorithm}

\subsubsection{{\sc MonocarpicExpert}}
\label{sec:sleep}

An important component of our full algorithm is the  {\sc MonocarpicExpert} procedure. In the problem of monocarpic experts, instead of having all experts presented at the beginning, an expert could enter and exit at any time (but each expert only ever wakes up once). The {\sc MonocarpicExpert} procedure achieves low regret with respect to every interval during the live period of an expert, while its memory scales with the maximum number of alive experts. 

\begin{theorem}
\label{thm:sleeping}
Let $T \geq 1$. For any expert $i$ that is alive over interval $\mathcal{I} \subseteq [T]$, the {\sc MonocarpicExpert} procedure (Algorithm \ref{algo:sleep}) guarantees that with probability at least $1-1/(nT)^{\omega(1)}$, 
\begin{align*}
\sum_{t \in \mathcal{I}'}\ell_t(i_t) - \sum_{t \in \mathcal{I}'}\ell_t(i) \leq O\left(\sqrt{|\mathcal{I'}|\log (nT)}\right)
\end{align*}
holds for every interval $\mI'\subseteq  \mI$, against an adaptive adversary.
Furthermore, suppose the maximum number of alive experts at most $M$, then the {\sc MonocarpicExpert} procedure uses up to $O(M\log^2(nT))$ words of memory.
\end{theorem}

It is worthy noting that the memory/regret guarantee of {\sc MonocarpicExpert} works even against an adaptive adversary.
The pseudocode of {\sc MonocarpicExpert} is presented as Algorithm \ref{algo:sleep}.
The set of experts is divided into $L$ subsets $\mU_{1} \cup \cdots \cup \mU_{L}$ ($L = \log_2 (T)$), and $\mU_{\ell}$ contains experts that live for at most $2^{\ell}$ days. 
The subset $\mU_{\ell}$ is updated every $2^{\ell-1}$ days (either receives experts from $\mU_{\ell-1}$ or gives away experts to $\mU_{\ell+1}$) and removes inactive expert. The membership of $\mU_{\ell}$ is fixed between two updates.
{\sc MonocarpicExpert} relies heavily on the {\sc IntervalRegret} subroutine. It runs {\sc IntervalRegret} over $L$ experts $\textsc{Exp}_{\ell}$ ($\ell \in [L]$), where $\textsc{Exp}_{\ell}$ itself runs $\textsc{IntervalRegret}$ on $\mU_{\ell}$ every $2^{\ell-1}$ days (between two updates, where membership of $\mU_{\ell}$ is fixed).
 If an expert in $\mU_{\ell}$ becomes inactive between two updates (i.e., it is alive at the beginning but exits in the middle), then we assign unit loss to it till the next update.

\begin{algorithm}[!htbp]
\caption{\textsc{MonocarpicExpert} 
}
\label{algo:sleep}
\begin{algorithmic}[1]
\State Initialize $\mU_\ell \leftarrow \emptyset$, $\textsc{Exp}_{\ell}$ ($\ell \in [L]$) \Comment{$L = \log_2(T) $}
\State $\textsc{Exp}\leftarrow \textsc{IntervalRegret}(\textsc{Exp}_1, \ldots, \textsc{Exp}_{L}, T)$
\For{$t=1,2, \ldots, T$}
\State Add newly activated experts to $\mU_1$  
\State Follow the decision of $\textsc{Exp}$
\For{$\ell = 1,2, \ldots, \mathsf{pw}(t)$} \Comment{Update membership}
\State $\mU_{\ell+1} \leftarrow \mU_{\ell+1} \cup \mU_{\ell}, \mU_{\ell} \leftarrow \emptyset$
\State Remove inactive experts in $\mU_{\ell+1}$
\EndFor
\EndFor\\ 

\Procedure{$\textsc{Exp}_{\ell}$}{} \Comment{$\ell \in [L]$}
\For{$s=1,2,\ldots, T/2^{\ell-1}$}\Comment{$s$-th restart}
\State $\textsc{IntervalRegret}(\mU_{\ell}, 2^{\ell-1})$
\EndFor
\EndProcedure
\iffalse
\Procedure{UpdateMembership}{$t$}
\State Add newly activated expert to $\mU_1$
\For{$\ell = 1,2, \ldots, \mathsf{pw}(t)$} 
\State $\mU_{\ell+1} \leftarrow \mU_{\ell+1} \cup \mU_{\ell}, \mU_{\ell} \leftarrow \emptyset$
\State Remove inactive experts in $\mU_{\ell+1}$
\EndFor
\EndProcedure
\fi
\end{algorithmic}
\end{algorithm}

The regret and memory guarantee largely follows from the {\sc IntervalRegret} subroutine. 
\begin{lemma}
\label{lem:interval}
Let $T \geq 1$. For any expert $i \in \mU$ and time interval $\mI \subseteq [T]$, Algorithm {\sc IntervalRegret} guarantees that with probability at least $1 - 1/(nT)^{\omega(1)}$,
\begin{align*}
\sum_{t \in \mathcal{I}}\ell_t(i_t) - \sum_{t \in \mathcal{I}}\ell_t(i) \leq O\left(\sqrt{|\mathcal{I}|\log (nT)}\right)
\end{align*}
holds against an adaptive adversary.
Moreover, {\sc IntervalRegret} uses up to $O(|\mU|\log(nT))$ words of memory.
\end{lemma}

The {\sc IntervalRegret} maintains a set of meta experts $\textsc{SingleInterval}_{a, b}$ ($a \in [L], b \in [T/2^a]$). 
The $\textsc{SingleInterval}_{a, b}$ is {\em effective} over the time interval $[2^a(b-1)+1: 2^a b]$. 
It runs MWU over $\mU$ in $[2^a(b-1)+1: 2^a b]$ (initiates with uniform weight) and does nothing outside of the interval. 
Let $h(t, a, b)$ denote the effectiveness of $\textsc{SingleInterval}_{a, b}$, that is, it $h(t, a, b) = 1$ when $t\in [2^a(b-1)+1: 2^a b]$ and $h(t, a,b) = 0$ otherwise.
%and runs MWU over $\mU$ during the time interval $[2^a(b-1)+1: 2^a b]$. 
We use $i_{t, a, b}$ to denote the action of $\textsc{SingleInterval}_{a, b}$ at day $t$.
%It indicates the effectiveness of $\textsc{SingleInterval}_{a, b}$ at day $t$. 
{\sc IntervalRegret} maintains a set of weights $\{w_{a, b}\}_{a \in [L], b \in [T/2^a]}$ over $\textsc{SingleInterval}_{a, b}$.
At day $t$, there are $L$ effective meta experts and the algorithm follows the advice from one of them, by sampling proportional to $\{w_{a, b}\}_{h(t, a, b) =1}$.
The loss vector is constructed as follow: For an effective meta expert $\textsc{SingleInterval}_{a, b}$ (i.e., $h(t, a, b) = 1$), it simply equals $\ell_{t}(i_{t, a, b})$, while for non-effective expert (i.e., $h(t, a, b)=0$), it is set to $\bar{\ell}_t$, where $\bar{\ell}_t$ is the ``expected'' loss computed by Eq.~\eqref{eq:loss-construct}. It equals the expected loss received by {\sc IntervalRegret} at day $t$.
The weight $w_{a, b}$ is updated by the {\sc Squint} algorithm \cite{koolen2015second}.

\begin{algorithm}[!htbp]
\caption{\textsc{IntervalRegret}($\mU, T$)}
\label{algo:interval}
\begin{algorithmic}[1]
\State Initialize $w_{a, b} \leftarrow 1$ over $\textsc{SingleInterval}_{a,b}$ ($a \in [L], b \in [T/2^a]$)
\For{$t=1,2, \ldots, T$}
\State Sample action $i_{t, a,b}$ from $\{w_{a, b}\}_{h(t, a, b) =1}$  \Comment{$h(t, a, b) = 1$ if $t \in [2^a(b-1)+1: 2^a b]$}
\State Compute the expected loss
\begin{align}
%\sum_{a, b: h(t, a, b) = 1}w_{a, b}\ell_{t}(i_{t, a, b}) + \sum_{a, b: h(t, a, b) = 0}w_{a, b}\bar{\ell}_{t} = \sum_{a, b}w_{a, b}\bar{\ell}_t 
\bar{\ell}_t \leftarrow \sum_{a, b: h(t, a, b) = 1}\frac{w_{a, b}}{\sum_{a, b: h(t, a, b) = 1} w_{a, b} } \cdot \ell_{t}(i_{t, a, b})
\label{eq:loss-construct}
\end{align}
\State Assign loss 
\[
\hat{\ell}_{t}( a, b ) = \left\{
\begin{matrix}
\ell_{t}(i_{t, a, b}) &  h(t, a, b) =1\\
\bar{\ell}_t & h(t, a, b) =0
\end{matrix}
\right.
\]
\State Update the weight distribution using {\sc Squint}
\begin{align}
w_{a, b}  \leftarrow \E_{\eta}\left[\eta \cdot \exp\left(\eta \sum_{\tau=1}^{t-1}v_{\tau}(a, b) - \eta^2 \sum_{\tau=1}^{t-1}v_{\tau}^2(a, b)\right)\right] \label{eq:squint}
\end{align}
where $v_{\tau}(a, b) = \bar{\ell}_{\tau} - \hat{\ell}_{\tau}(a, b)$
\EndFor\\
\Procedure{$\textsc{SingleInterval}_{a, b}$}{}
\For{$t = 2^{a}(b-1) +1, \ldots, 2^{a}b$}
\State Run MWU over $\mU$.
\EndFor
\EndProcedure
\end{algorithmic}
\end{algorithm}

\begin{proof}
For any $a \in [L], b \in [T/2^{a}]$ and expert $i \in \mU$, due to the regret guarantee of MWU, one has, with probability at least $1- 1/(nT)^{\omega(1)}$,
\begin{align}
\sum_{t=2^a(b-1)+1}^{2^a b}\ell_{t}(i_{t, a, b}) -\sum_{t=2^a(b-1)+1}^{2^a b}\ell_{t}(i) \leq O\left(\sqrt{2^a \log (nT)}\right) \label{eq:interval1}.
\end{align}

Recall $\bar{\ell}_t$ is the solution of Eq.~\eqref{eq:loss-construct} at day $t$, it equals the expected loss received by {\sc IntervalRegret}, since
\begin{align*}
\bar{\ell}_t  = \sum_{a,b: h(t,a, b) = 1}\frac{w_{t, a, b}}{\sum_{a, b: h(t, a, b) = 1}w_{t, a, b}}\ell(i_{t, a, b}) = \sum_{a,b: h(t,a, b) = 1}p_{t,a, b}\ell(i_{t, a, b}) 
\end{align*}
where $p_{t,a, b} = \frac{w_{t, a, b}}{\sum_{a, b: h(t, a, b) = 1}w_{t, a, b}}$ is the probability of following $i_{t, a, b}$.

By the regret guarantee of {\sc Squint} (see Lemma \ref{lem:squint}), one has
\begin{align}
&~\E\left[\sum_{t\in [2^a (b-1) + 1: 2^a b]}\ell_t(i_t) -  \sum_{t\in [2^a (b-1) + 1: 2^a b]}\ell_t(i_{t,a, b})\right]\notag \\
= &~ \E\left[\sum_{t=1}^{T}\ell_t(i_t) - \sum_{t=1}^{T}\hat{\ell}_t(a, b)\right]\notag \\
\leq &~  O\left(\sqrt{\left(\sum_{t\in [T]}(\E[\ell_t(i_t)] - \hat{\ell}_t(a, b))^2\right)\cdot \log (nT)}\right) =  O\left(\sqrt{2^a\log (nT)}\right) \label{eq:interval2}.
\end{align}
The first and third steps hold since $\hat{\ell}_t(a, b) = \ell_t(i_{t, a, b})$ for $t \in [2^a (b-1) + 1: 2^a b]$ and loss $\hat{\ell}_t(a, b) = \bar{\ell}_t = \E[\ell_t(i_t)]$ for any $t \notin [2^a (b-1) + 1: 2^a b]$, the second step follows from the regret guarantee of {\sc Squint}.

Applying the Azuma-Hoeffding inequality to Eq.~\eqref{eq:interval2} and note $|\E[\ell_t(i_t)] - \ell_t(i_t)| \leq 2$, one obtains
\begin{align}
\sum_{t\in [2^a (b-1) + 1: 2^a b]}\ell_t(i_t) -  \sum_{t\in [2^a (b-1) + 1: 2^a b]}\ell_t(i_{t,a, b}) \leq O\left(\sqrt{2^a \log (nT)}\right) \label{eq:interval3}
\end{align}
holds with probability at least $1-1/(nT)^{\omega(1)}$.

Combining Eq.~\eqref{eq:interval1} and Eq.~\eqref{eq:interval3}, we have that with probability at least $1-1/(nT)^{\omega(1)}$,
\begin{align*}
\sum_{t\in [2^a (b-1) + 1: 2^a b]}\ell_t(i_t) - \sum_{t\in [2^a (b-1) + 1: 2^a b]}\ell_{t}(i)\leq O\left(\sqrt{2^a \log (nT)}\right).
\end{align*}

To conclude the regret analysis, we note for any interval $\mI = [t_1: t_2]\subseteq T$, one can split $\mI$ into $X \leq 2\log_2(|\mI|)$ disjoint intervals $\mI = \mI_1 \cup \mI_{2}\cup \cdots \cup \mI_X$, such that  (1) $\mI_{x}$ ($x \in [X]$) exactly spans the lifetime of some meta expert $\textsc{SingleInterval}_{a_x, b_x}$ and (2) there are at most two length-$2^{x}$ intervals. Then we conclude 
%one splits the interval as follow. Let $a = a_1\ldots a_{\log T}$ and $b_{1}\ldots b_{\log T}$ be the binary representation of $a$ and $b$, and let $d \in [\log T]$ be the first index such that $a_d \neq b_d$, then we know that $a_{d} = 0, b_d = 1$. Let $c=a_1\ldots a_d \ldots 000$, let $x = a - c = x_1 \ldots x_{\log n}$ and $y = c - b = y_1\ldots y_{\log n}$, then we split the interval as $\{I_{1, u}\}_{x_u = 1}$ and $\{I_{2, u}\}_{y_u = 1}$ base on $x, y$, and $I_{1, u}, I_{2, u}$ is of length $2^u$, then we have
\begin{align*}
\sum_{t\in \mI}\ell_t(i_t) - \sum_{t\in \mI}\ell_{t}(i) = &~ \sum_{x=1}^{X}\sum_{t\in \mI_x}(\ell_t(i_t) - \ell_{t}(i)) \leq  \sum_{x=1}^{X} O\left(\sqrt{|\mI_x|\log(nT)}\right)\\
\leq &~ \sum_{x=1}^{\log (|\mI|)}O\left(\sqrt{2^{x}\log (nT)}\right) = O\left(\sqrt{|\mI|\log(nT)}\right).
\end{align*}

\paragraph{Memory usage of {\sc IntervalRegret}}
A naive implementation takes $O(T|\mU|\log(nT))$ as one needs to maintain $O(T\log T)$ $\textsc{SingleInterval}_{a, b}$ procedures.
However, note that at day $t$, the algorithm only needs to know the weights $\{w_{a, b}\}_{h(t, a, b)=1}$ to determine its action and compute the average loss $\bar{\ell}_t$. 
That is, it suffices to know the weight $w_{a, b}$ of effective experts.
There are $L = \log_2 (T)$ effective meta experts at any time, and $\textsc{SingleInterval}_{a, b}$ is effective for a consecutive period with weight $w_{a, b}$ starting from $\E_{\eta}[\eta]$.
Hence {\sc IntervalRegret} requires $O(|\mU|\log (nT))$ words of memory in total. We conclude the proof here.

%and there are at most 
%Note naively there are $T\log(nT)$ experts that need to maintain and each expert specifies weight distribution over $M$ experts.
%For constructing the loss and sampling the expert, one only needs to know $\frac{p_{a, b}}{\sum_{t(a, b)=1}p_{a, b}} \propto \exp(\eta \sum_{\tau=1}^{t-1}v_{\tau}(i) - \eta^2 \sum_{\tau=1}^{t-1}v_{\tau}^2(i))$ by Eq.~\eqref{eq:squint}.
%Note by our construction rule, $v_t(i) = 0$ whenever the expert $i$ is not active. Hence it suffices to understand the loss of active experts. Since there are at most $O(\log nT)$ live $\textsc{SingleInterval}_{a, b}$ at one, the total memory required is $|M|\log(T)$. We finish the proof here.
\end{proof}

We can now conclude the analysis of {\sc MonocarpicExpert}.
\begin{proof}[Proof of Theorem \ref{thm:sleeping}]
We first provide the regret analysis. 
For any expert $i$ that is alive over interval $\mI$. Let $\mI' = [t_1: t_2] \subseteq \mI$ be any sub-interval. One can split the interval $\mI'$ into $X \leq \log(|\mI'|)$ consecutive intervals $\mI' = \mI'_1 \cup \mI'_2 \cup \ldots \cup \mI'_X$, such that expert $i$ resides in $\mU_{\ell_x}$ during the interval of $\mI'_x$ ($\ell_1 < \ell_2 < \cdots < \ell_{X}$).
Moreover, the size of $\mI'_x$ is exponentially increasing, except for $\mI'_1$ and $\mI'_X$
($i$ may not stay for an entire update at $\mU_{\ell_1}$ and $\mU_{\ell_X}$).
Let $i_{t,\ell}$ be the action of $\textsc{Exp}_{\ell}$ at day $t$, Then due to the {\sc IntervalRegret} guarantee, with probability at least $1-1/(nT)^{\omega(1)}$, one has
\begin{align*}
\sum_{t \in \mI'}\ell_{t}(i_t) - \sum_{t \in \mI'}\ell_t(i) =  &~ \sum_{x=1}^{X}\sum_{t \in \mI'_{x}}\ell_{t}(i_t) - \ell_t(i) \\
= &~ \sum_{x=1}^{X}\sum_{t \in \mI'_{x}}\ell_{t}(i_{t}) - \ell_t(i_{t, \ell_x}) +  \sum_{x=1}^{X}\sum_{t \in \mI'_{x}}\ell_{t}(i_{t, \ell_x}) - \ell_t(i)\\
\leq &~ \sum_{x=1}^{X} O(\sqrt{|\mI_{x}'|\log (nT)}) \leq O(\sqrt{|\mI'|\log (nT)}).
\end{align*}
Here the third step follows from the interval guarantee (see Lemma \ref{lem:interval}) of $\textsc{Exp}_{\ell_x}$, the interval regret guarantee of $\textsc{Exp}_{\ell_x}$ on $\mU_{\ell_{x}}$,  and the fact that $\textsc{Exp}_{\ell_x}$ restarts at most twice when $i$ resides in $\mU_{\ell_x}$.
The last step follows from the increasing size of $\mI_{x}'$.

Finally, for the memory usage, there are $L = \log_2(T)$ procedures $\textsc{Exp}_{\ell}$ ($\ell \in [L]$), and $\textsc{Exp}_{\ell}$ runs on $\mU_{\ell}$.
It suffices to bound the size of $\mU_{\ell}$. 
Note that $\mU_{\ell}$ is updated every $2^{\ell-1}$ days, and at the time of update, all experts in $\mU_{\ell}$ are alive and the set remains the same during the next $2^{\ell-1}$ days. Hence $|\mU_{\ell}| \leq M$ and the memory usage of $\textsc{Exp}_{\ell}$ is at most $O(M\log (nT))$.  
The total memory is at most $O(M\log^2(nT))$.
\end{proof}

\subsubsection{Analysis of full algorithm}
We next analyse the memory and regret guarantee of the full algorithm.

\subsubsection*{Memory bound}
The memory bound is largely inherited from {\sc Baseline}.
\begin{lemma}
\label{lem:pool_size_full}
With probability at least $1 - 1/(nT)^{\omega(1)}$, for any $r \in [R]$, the pool $\mP_{r, \cdot}$ never exceeds $O(\log^{10}(nT))$.
\end{lemma}
\begin{proof}
We prove with probability at least $1- 1/(nT)^{\omega(1)}$, at any time, the size of $\mP_{r, \ell}$ is at most $2\log^9 (nT)$ ($r \in [R]$, $\ell \in [L]$) and $\mP_{r, 0}$ is at most $3\log^9 (nT)$.  
The case of $r = R$ holds trivially and suppose the induction holds for thread $r+1$. 

For the $r$-th thread, $\mP_{r, 0}$ is initiated as $\textsc{Sample}(\mN, 1/n)$, or the union of $\textsc{Sample}(\mN, 1/n)$ and the result of merging $\mP_{r+1, \cdot}, \ldots, \mP_{R, \cdot}$ (Line \ref{line:cont-merge1} -- \ref{line:cont-merge2} of Algorithm \ref{algo:base+}).
We have by induction $|\mP_{r', \cdot}| = |\mP_{r', 2}| \leq 2\log^9 (nT)$ for any $r' > r$ (there are only two epochs of  $\textsc{Baseline}_{+}(r')$ per restart). Hence by Lemma \ref{lem:merge_size}, the pool obtained from merging  $\mP_{r+1, \cdot}, \ldots, \mP_{R, \cdot}$ has size at most $2\log^9(nT)$, and by Chernoff bound, the sample set $\textsc{Sample}(\mN, 1/n)$ is at most $\log^9(nT)$ with high probability. 

While for $\ell \geq 1$, by Lemma \ref{lem:merge_size}, the \textsc{Merge} procedure guarantees that the sub-pool $\mP_{r, \ell}$ satisfies
\begin{align*}
|\mP_{r, \ell}| \leq &~ \max\left\{2\log^9(nT), \frac{1}{4}(|\mP_{r, \ell}| + |\mP_{r, \ell-1}|)\right\}\\
\leq &~ \max\left\{2\log^9(nT), \frac{1}{4}(2\log^9(nT) + 3\log^9(nT))\right\} = 2\log^9(nT).
\end{align*}
%We conclude the proof here.
\end{proof}

Consequently, the memory usage of the full algorithm is not large.
\begin{lemma}
\label{lem:memory}
With probability at least $1 - 1/(nT)^{\omega(1)}$, the pool size of $\mP$ never exceeds $O(\log^{11}(nT))$. 
Furthermore, the memory usage of Algorithm \ref{algo:full} is bounded by $O(\log^{22}(nT))$ words.
\end{lemma}
\begin{proof}
By Lemma \ref{lem:pool_size_full}, we know that the size of pool $\mP$ satisfies $|\mP| = |\mP_{1, \cdot} \cup \cdots \cup \mP_{R, \cdot}| \leq O(\log^{11}(nT))$. 
The {\sc Merge} procedure needs to know the loss on each interval $\Gamma_{i, j}$ of $i, j\in \mP$, this takes $O(\log^{22}(nT))$ words of memory. 
By Theorem \ref{thm:sleeping}, the {\sc MonocarpicExpert} requires memory $O(|\mP|\log^2(nT)) = O(\log^{13}(nT))$ . In summary, the total memory is bounded by $O(\log^{22}(nT))$.
%\[
%O(\log^{22}(nT) + \log^{11}(nT) \cdot G_s) = O(S).
%\]
We conclude the proof here.
\end{proof}

\subsubsection*{Regret analysis}
The main departure (and complication) comes from the regret analysis.

\paragraph{Notation} 
Let $i^{*} \in [n]$ be the best expert and let
\[
K_1 = [0:n-1], \quad K_2 = [0:n-1]\times \{0,1\}  \quad \ldots  \quad K_R=  [0:n-1] \times \underbrace{\{0,1\}\times \cdots \times \{0,1\}}_{R-1}
\]
and $K$ be the union of $K_1, \ldots, K_R$, i.e.,
\[
K = K_1 \cup \cdots \cup K_{R}.
\]
Given any timestep $a = (a_{1}, \ldots, a_{r(a)}) \in K$ (where $r(a)$ is defined such that $a \in K_{r(a)}$), the timestep $a$ uniquely identifies an epoch of $\textsc{Baseline}_{+}(r(a))$, i.e., it refers to the $a_{r(a)}$-th epoch of the $(\sum_{r =1}^{r(a)-1}a_r 2^{r(a)-r-1})$-th restart.

\begin{definition}[Operator $\oplus$]
For any timestep $a\in K$, we write 
\[
a'  = (a_1', a_2', \ldots a_{r(a')}') = a\oplus 1 \in K 
\]
as the unique timestep that satisfies
\[
\sum_{i=1}^{r(a')} a_i'B_i = \sum_{i=1}^{r(a)} a_i B_i + B_{r(a)} \quad \text{and} \quad a'_{r(a')} \neq 0. 
\]
That is to say, $a' = a \oplus 1$ is the next number under $2$-base with $0$ truncated at the end  (except the first coordinate, which belongs to $[0:n-1]$).
\end{definition}
Intuitively, $a\oplus 1$ is the next complete epoch after epoch $a$, if there are multiple epochs of different threads, it refers to the one at the lowest thread.

Let $K(a)$ contain all timesteps that succeed $a$ under the $\oplus$ operation, i.e., 
\[
K(a) := \{a\} \cup \{a \oplus 1\} \cup \{(a \oplus 1) \oplus 1\} \cup \cdots  \subseteq K 
\]
Roughly speaking, $K(a)$ includes all timesteps/epochs that an expert enters at $a$ would reside. 

\paragraph{Random bits} 
Again, it is useful to understand the random bits used for pool selection before defining an active/passive expert.
For any timestep $a \in K$, let the random bits $\xi_{a} = (\xi_{a, 1}, \ldots, \xi_{a, n})$, where $\xi_{a, i} = (\xi_{a,i, 1}, \xi_{a, i,2})$ is used for expert $i$ ($i \in [n]$).
The first coordinate $\xi_{a, i, 1} \in \{0, 1\}$ is a Bernoulli variable with mean $1/n$. It is used for sampling new expert into the pool at the beginning of epoch $a$ (i.e. Line \ref{line:base_new+} of Algorithm \ref{algo:base+}).
The second part of random bits $\xi_{a, i, 2} \in \{0, 1\}^{R \times 2L \times 2K}$ are used for estimating size and filtering in the {\sc Merge} procedure of each sub-pool, they are of mean $\frac{1}{\log^4 (nT)}$.
%Concretely, the random bits $\xi_{r, s,t,i,2, r}$ ($r \in [R]$) are used for the {\sc Merge} procedure when $i$ is in pool $\mP_{r}$. 
%When expert $i$ lies in pool $\mP_{r}$ ($r\geq 2$), it subjects to (at most) $4$ calls of {\sc Merge}, and during each call of {\sc Merge}, it subjects to (at most) $2K= 32 \lceil\log (nT)\rceil$ calls of the {\sc Sample} and {\sc EstimateSize}.
%While when expert $i$ lies in $\mP_1$, it subjects to at most $L = \lceil\log(nT)\rceil$ calls of {\sc Merge} and $2K$ calls of {\sc Sample} and {\sc EstimateSize}.
%Hence, in total, let $S = 8KR + 2KL = O(\log^{2}(nT))$ and random bits $\xi_{t, s,t,i, 2} \in \{0, 1\}^{S}$, each coordinate is a Bernoulli random variable with mean $\frac{1}{\log^3 (nT)}$, and the expert $i$ gets sampled for {\sc Sample} and {\sc EstimateSize} if the corresponding coordinate equals $1$.
%Note there are many random bits that never get used, e.g., we actually do not need $\xi_{r, s,t,i,2, r'}$ for $r' \geq r+1$ since experts only go from high level to low level, we define for convenience of proof.
%Finally, let $\xi = \{\xi_{\adv}\} \cup \{\xi_{r,s,t}\}_{r\in [R], s\in [T/T_r], t\in [T_r/B_r]}$, $\Omega$ be the sample space over $\xi$ and $\mu$ be the probability measure.

Recall the concept of passive/active expert.
\begin{definition}[Active/Passive expert]
At any timestep $a \in K$, an expert $i\in [n]$ is said to be {\em passive}, if $\xi_{a, i, 2} = \vec{0}$. It is said to be an {\em active} expert otherwise.
\end{definition}

\paragraph{Epoch Assignment}  Our accounting argument will split the entire sequence of epochs into a  collection of disjoint subsequences, and assign each subsequence to a different set of experts in the pool. 

Fix the loss sequence $\ell_1, \ldots, \ell_{T}$ and the set of {\em sampled and active} experts $Y_a \subseteq [n]$ for each timestep $a \in K$. 
%The key observation is that the estimate size $s$, the filter set $\mF$, as well as the set of alive active experts are fully determined (at any time), regardless of the sets of {\em sampled but passive} experts.
Our key observation is
\begin{observation}
\label{obs:fix2}
Suppose the loss sequence $\{\ell_t\}_{t\in [T]}$ and the set of sampled and active experts $\{Y_a\}_{a \in K}$ are fixed, then at any time during the execution of Algorithm \ref{algo:full}, the estimate size $s$, the filter set $\mF$ and the set of alive active experts are also fixed, regardless of the set of sampled and passive experts.
\end{observation}

%\begin{proof}
%A passive expert never ``appears'' in the {\sc EstimateSize} and {\sc Filter} 
%\end{proof}

%We omit the proof of Lemma \ref{lem:fix2} as it is straightforward.

\begin{definition}[Eviction time, full algorithm]
For any timestep $a \in K$, consider the expert $i^{*}$ with entering time $a$, the eviction time $t(a) \in K(a)\cup \{+\infty\}$ is defined as the first timestep of $K(a)$, such that $i^{*}$ is covered by the set of alive active experts at the end of $t(a)$.
If $i^{*}$ would not be covered, then set $t(a) = +\infty$.
\end{definition}

\begin{algorithm}[!htbp]
\caption{Epoch assignment (Note: only used in analysis)}
\label{algo:sequence}
\begin{algorithmic}[1]
\State Initialize $\mH_{r} \leftarrow \emptyset$ ($r\in [R]$), $\tau \leftarrow 1$, $a_{1} \leftarrow 0$
\While{$\cup_{\tau'\leq \tau}\mI_{\tau'}  \neq [T]$}
\If{$t(a_{\tau}) = +\infty$}\Comment{$i^{*}$ survives till the end} 
\State $\mH_{r(a_{\tau})} \leftarrow \mH_{r(a_{\tau})} \cup a_{\tau}$ \Comment{$a_\tau$ is a bad epoch at thread $r(a_\tau)$}\label{line:augment}
\If{$r(a_{\tau}) = R$} \Comment{Stop at the top thread}
\State $\mI_{\tau} \leftarrow a_\tau$, $a_{\tau + 1} \leftarrow a_{\tau} \oplus 1$, $\tau \leftarrow \tau + 1$, 
\Else 
\State $a_{\tau} \leftarrow (a_{\tau}, 0)$ \Comment{Move to next thread}
\EndIf
\Else
\State $a_{\tau + 1} \leftarrow t(a_{\tau}) \oplus 1$, $I_{\tau} \leftarrow [a_{\tau}: t(a_{\tau})]$, $\tau \leftarrow \tau + 1$
\EndIf
\EndWhile
\end{algorithmic}
\end{algorithm}

The key ingredient in our analysis is the epoch assignment mechanism, whose pseudocode is presented as Algorithm \ref{algo:sequence} (note again that this algorithm is only for analysis). Algorithm \ref{algo:sequence} splits the entire sequence and determines the bad epochs $\mH_{r}$ of each thread $r$ ($r\in [R]$) through a bottom-up walk.
It starts from the bottom thread $r=1$. 
At each step $\tau$, if the expert $i^{*}$ could survive till the end with entering time $a_{\tau}$, then it is a bad epoch at thread $r(a_{\tau})$ and the algorithm augments the collection of bad epochs $\mH_{r(a_{\tau})}$ of thread $r(a_{\tau})$ (Line \ref{line:augment}). 
Instead of moving to the next epoch of the same thread, Algorithm~\ref{algo:sequence} walks to the next thread, unless it is already at the top thread $R$. 
If the expert $i^{*}$ would be covered until timestep $t(a_{\tau}) \in K(a)$, 
then Algorithm~\ref{algo:sequence} sets the $\tau$-th interval as $\mI_{\tau} = [a_{\tau}:t(a_{\tau})]$ and moves to the next epoch $a_{\tau+1} = t(a_{\tau})\oplus 1$.  
Here $a_{\tau}$ and $t(a_{\tau})$ do not need to be at the same thread, and we write $\mI_{\tau} = [a_{\tau}:t(a_{\tau})]$ to denote the time interval between the beginning of $a_\tau$ and the end of $t(a_{\tau})$.
See Figure \ref{fig:walk} for an illustration.

Let $\tau_{\max}$ be the total number of intervals in the partition generated by Algorithm~\ref{algo:sequence}.
We first make a few simple observations about the intervals $\{\mI_{\tau}\}_{\tau \in [\tau_{\max}]}$.

\begin{lemma}
\label{lem:obv-interval}
We have
\begin{itemize}
\item The intervals $\{\mI_{\tau}\}_{\tau \in [\tau_{\max}]}$ are disjoint and $\bigcup_{\tau \in [\tau_{\max}]}\mI_{\tau} = [T]$.
\item Let $L_1 = \{\mI_{\tau}\}_{\tau \in [\tau_{\max}] }$ and let $L_{r}:= \{\mI_{\tau}:  |\mI_{\tau}| < B_{r-1}\}$ ($r \in [2:R]$) contain intervals of length less than $B_{r-1}$, then 
\begin{align*}
\sum_{\mI\in L_{r}} |\mI| \leq |\mH_{r-1}| \cdot B_{r-1}.
\end{align*}
\end{itemize}
\end{lemma}
\begin{proof}
The first claim follows directly from the assignment process and we focus on the second claim.
For any thread $r \in [2:R]$ and any interval $\mI_{\tau}\in L_{r}$, we prove $\mI_{\tau}$ is contained in a bad epoch of thread $r-1$.

Recall $\mI_{\tau}$ starts from the beginning of epoch $a_\tau$ and terminates at the end of epoch $t(a_{\tau})$ ($a_{\tau}, t(a_{\tau})$ are not necessarily at the same thread).
We first observe $t(a_\tau)$ is of thread at least $r$, otherwise $\mI$ spans at least an epoch of thread $r-1$.
Suppose the timestep $t(a_\tau)$ is contained in epoch $a'$ of thread $r-1$ (i.e., $r(a') = r-1$), it suffices to prove (1) $\mI$ starts within $a'$ and (2) $a'$ is a bad epoch of thread $r-1$ (i.e., $a' \in \mH_{r-1}$). 
%we must have $r(a_{\tau}) \geq r$, otherwise $\mI_{\tau}$ at least spans $a_{\tau}$, which is of length at least $B_{r-1}$.
The first claim follows from the epoch assignment procedure.
For the second claim, suppose $a'$ is not a bad epoch of thread $r-1$, then let $a_1'$ be the closest thread $r-1$ bad epoch before $a'$ and $a_2'$ be the closest thread $(r-1)$ bad epoch after $a'$. The ``walk''  defined by Algorithm \ref{algo:sequence} would not go above thread $r-1$ between $a_1'$ and $a_2'$, which contradicts our assumption of $\mI\in L_{r}$.
We conclude the proof here.

%The key observation is that the assignment process (Algorithm \ref{algo:sequence}) guarantees that $a_{\tau}$ lies in a bad epoch of $\mH_{r-1}$, i.e., there exists $\tau' \in [\tau_{\max}]$ such that (1) $a_{\tau'} \in \mH_{r-1}$ and (2) $a_{\tau} \subseteq a_{\tau'}$.
%Here we write $a_1 \subseteq a_2$ ($a_1, a_2 \in K$) if epoch $a_1$ is a subset of epoch $a_2$.
%Meanwhile, the ending epoch $t(a_{\tau})$ must be contained in $a_{\tau'}$ (i.e., $t(a_\tau) \subseteq a_{\tau'}$), otherwise $\mI_{\tau}$ at least contains $a_{\tau'} \oplus 1$, which is of length at least $B_{r-1}$.
%In summary, we have proved $\mI_{\tau} \subseteq a_{\tau'}$, that is, intervals of $L_r$ must be a subset of $\mH_{r-1}$, and therefore, their total length is at most $|\mH_{r-1}|\cdot B_{r-1}$.
\end{proof}

We next prove the size of bad epochs $\mH_{r}$ is at most $O(n\log^{11} (nT))$ with high probability.

\begin{lemma}
\label{lem:bad-size-full}
With probability at least $1 - 1/(nT)^{\omega(1)}$, $|\mH_r| \leq O(n\log^{11}(nT))$ holds for any $r\in [R]$.
\end{lemma}
\begin{proof}
Recall we first fixed the loss sequence $\ell_1, \ldots, \ell_{T}$ as well as the set of sampled and active experts $\{Y_a\}_{a\in K}$.
For any $r\in [R]$, we prove the size of pool $\mP$ at the end satisfies 
\[
\Pr\left[|\mP|\geq \frac{|\mH_r|}{4n} \mid \{Y_a\}_{a\in K}, \{\ell_t\}_{t\in[T]} \right] \geq \frac{1}{2},
\]
whenever $|\mH_r| \geq \Omega(n)$. Here the probability is taken over the randomness of the remaining  {\em sampled and passive} experts. 

The proof is similar to Lemma \ref{lem:base-size}.
Let $W_a$ be the set of sampled and passive experts of timestep $a$. For any timestep $a \in \mH_r$, by our assignment mechanism (Algorithm \ref{algo:sequence}), the expert $i^{*} \notin Y_{a}$ (otherwise $i^{*}$ is covered by itself.)
The probability that $i^{*} \in W_a$ obeys
\begin{align*}
\Pr\left[i^{*}\in W_a | \{Y_a\}_{a\in K}, \{\ell_t\}_{t\in[T]} \right] = &~  \Pr[i^{*}\in W_a | i^{*}\notin Y_{a}] \\
\geq &~ \Pr[i^{*}\in W_a] = \frac{1}{n} \cdot \left(1 - \frac{1}{\log^4 (n)}\right)^{R\times 2L \times 2K} \geq \frac{1}{2n}.
\end{align*}

The key observation is that once $i^{*} \in W_{a}$, i.e., $i^{*}$ is sampled while passive at timestep $t$, it would survive till the end. This is because the filter set $\mF$, the estimate size $s$ and the set of alive active experts are all fully determined by the sampled and active experts $\{Y_a\}_{a\in K}$ and the loss sequence $\{\ell_t\}_{t\in [T]}$, which has already been fixed when defining the eviction time $t(a)$. 
Therefore, if a passive expert $i^{*}$ enters at timestep $a$, it survives till the end.

The event of $i^{*}\in W_a$ are independent for $a \in \mH_r$ (condition on the loss sequence and $\{Y_a\}_{a\in [K]}$), hence, by Chernoff bound, one has
\begin{align*}
\Pr\left[|\mP|\leq \frac{|\mH_r|}{4n} \mid \{Y_a\}_{a\in K}, \{\ell_t\}_{t\in[T]}\right] \leq &~ \Pr\left[\sum_{a\in \mH_r}\mathsf{1}\{i^{*}\in W_{a}\} \leq \frac{|\mH_r|}{4n} \mid \{Y_a\}_{a\in K}, \{\ell_t\}_{t\in[T]}\right]\\
\leq &~ \exp(-|\mH_r|/16n). 
\end{align*}

Note we have already proved in Lemma \ref{lem:memory} that $|\mP| \geq  \Omega(\log^{11}(nT))$ happens with probability at most $1/(nT)^{\omega(1)}$, this implies 
\[
\Pr[|\mH_r|\geq \Omega(n\log^{11}(nT))] \leq 1/(nT)^{\omega(1)}.
\] 
We conclude the proof here.
\end{proof}

Now we apply the regret guarantee of {\sc MonocarpicExpert} and prove that Algorithm \ref{algo:full} achieves low regret over each interval $\mI_{\tau}$.
\begin{lemma}
\label{lem:interval-regret-apply}
With probability at least $1-1/(nT)^{\omega(1)}$, for any $\tau \in [\tau_{\max}]$ and $a_{\tau} \notin \mH_{R}$, 
\begin{align*}
\sum_{t \in \mI_{\tau}}\ell_{t}(i_t) - \sum_{t \in \mI_{\tau}}\ell_{t}(i^{*}) \leq O\left(\sqrt{|\mI_{\tau}|} \cdot \log^6(nT)\right).
\end{align*}
%and for $a_\tau \in \mH_R$,
%\begin{align*}
%\sum_{t \in \mI_{\tau}}\ell_{t}(i_t) - \sum_{t \in \mI_{\tau}}\ell_{t}(i^{*}) \leq O(1).
%\end{align*}
\end{lemma}
\begin{proof}
We condition on the event of Lemma \ref{lem:pool_size_full}.
Given any interval $I_{\tau}$ starting with $a_\tau$, ending with $t(a_{\tau})$ and $a_{\tau} \notin \mH_{R}$, the expert $i^{*}$ with entering time $a_\tau$ is covered by $\mP$ the end of $t(a_{\tau})$. Let $i_1^{*}, \ldots, i_{s}^{*}$ be the set of experts that cover $i^{*}$.
By Definition \ref{def:cover}, we can partition the interval $\mI_{\tau} = \mI_{\tau, 1} \cup \cdots \cup \mI_{\tau, s}$ and obtain
\begin{align}
\label{eq:mg1}
\sum_{t \in \mI_{\tau}}\ell_{t}(i^{*}) \geq \sum_{j=1}^{s}\sum_{t\in \mI_{\tau, j}}\ell_t(i_{j}^{*}) - 1
\end{align}
By the regret guarantee of {\sc MonocarpicExpert}, with probability at least $1 - 1/(nT)^{\omega(1)}$, one has
\begin{align}
\sum_{t \in \mI_{\tau}}\ell_{t}(i_t) - \sum_{t \in \mI_{\tau}}\ell_{t}(i^{*}) = &~ \sum_{t \in \mI_{\tau}} \ell_{t}(i_t) - \sum_{j=1}^{s}\sum_{t\in \mI_{\tau, j}}\ell_t(i_{j}^{*}) +  \sum_{j=1}^{s}\sum_{t\in \mI_{\tau, j}}\ell_t(i_{j}^{*}) - \sum_{t \in \mI_{\tau}}\ell_{t}(i^{*})\notag\\
\leq &~ \sum_{j=1}^{s}O\left(\sqrt{|\mI_{\tau, j}|\log(nT)}\right) + \sum_{j=1}^{s}\sum_{t\in \mI_{\tau, j}}\ell_t(i_{j}^{*})  - \sum_{t \in \mI_{\tau}}\ell_{t}(i^{*}) \notag  \\
\leq &~ \sum_{j=1}^{s}O\left(\sqrt{|\mI_{\tau, j}|\log(nT)}\right) + 1\notag \\
\leq &~ \sqrt{|\mI_{\tau}|} \cdot \log^{6}(nT).\notag%\label{eq:mg2}
\end{align}
The second step follows from Theorem \ref{thm:sleeping}, the third step follows from Eq.~\eqref{eq:mg1} and the last step follows from $s \leq \log^{11}(nT)$ (Lemma \ref{lem:pool_size_full}), $\mI_{\tau} = \mI_{\tau, 1} \cup \cdots \cup \mI_{\tau, s}$ and Cauchy-Schwarz. We conclude the proof here.

%Meanwhile, recall the meta expert $i^{**}$ runs MWU over a set of experts that contains $i^{*}$, we have that
%\begin{align} 
%\sum_{t \in \mI_{\tau}}\ell_{t}(i^{**}) - \sum_{t \in \mI_{\tau}}\ell_{t}(i^{*}) \leq O\left(\sqrt{|\mI_{\tau}|\log(nT)}\right)\label{eq:mg3}
%\end{align}
%holds with probability at least $1-1/(nT)^{\omega(1)}$. 
%Combining Eq.~\eqref{eq:mg2} \eqref{eq:mg3}, we complete the proof.
\end{proof}

Now we can bound the regret of full algorithm.

\begin{lemma}
\label{lem:regret-full}
With probability at least $1- 1/(nT)^{\omega(1)}$, one has
\[
\sum_{t \in [T]}\ell_{t}(i_t) - \sum_{t \in [T]}\ell_{t}(i^{*}) \leq \wt{O}\left(\sqrt{nT}\right).
\]
\end{lemma}
\begin{proof}
In the proof, we first fix the loss sequence $\{\ell_t\}_{t\in [T]}$ and the set of sampled and active experts $\{Y_a\}_{a\in K}$. 
Now we can use the epoch assignment mechanism (Algorithm \ref{algo:sequence}) and split the entire sequence $[T]$ into a collection of intervals $\{I_{\tau}\}_{\tau \in [\tau_{\max}]}$.

We first categorize the intervals based on their length. With probability at least $1-1/(nT)^{\omega(1)}$, we have
\begin{align}
\sum_{t \in [T]}\ell_{t}(i_t) -\sum_{t \in [T]}\ell_{t}(i^{*}) = &~ \sum_{r=1}^{R-1}\sum_{\mI \in L_{r}\setminus L_{r+1}}\sum_{t\in \mI}(\ell_{t}(i_t) -  \ell_{t}(i^{*})) + \sum_{\mI \in L_R}\sum_{t\in \mI}(\ell_{t}(i_t) -  \ell_{t}(i^{*}))\notag \\
\leq &~ \sum_{r=1}^{R-1}\sum_{\mI \in L_{r}\setminus L_{r+1}}\sum_{t\in \mI}(\ell_{t}(i_t) -  \ell_{t}(i^{*})) + O(|\mH_{R-1}|)\notag \\
\leq &~ \sum_{r=1}^{R-1}\sum_{\mI \in L_{r}\setminus L_{r+1}}O\left(\sqrt{|\mI|}\log^6 (nT)\right) + O(|\mH_{R-1}|).\label{eq:regret-last1}
\end{align}
The first step holds from $[T] = \cup_{\mI\in L_1} \mI$, the second step uses the naive bound of
\[
\sum_{\mI \in L_R}\sum_{t\in \mI}(\ell_{t}(i_t) -  \ell_{t}(i^{*})) \leq \sum_{\mI\in L_{R}}|\mI| \leq |\mH_{R-1}|\cdot B_{R-1} = O(|\mH_{R-1}|),
\]
and the last step follows from the interval regret bound of Lemma \ref{lem:interval-regret-apply}. 

We bound the regret of each group separately.
%To bound the RHS of Eq.~\eqref{eq:regret-last1}, we 
Note for any $r\in [R-1]$, one has 
\begin{align}
\sum_{\mI \in L_{r}\setminus L_{r+1}}\sqrt{|\mI|} \leq &~ \sqrt{| L_{r}\setminus L_{r+1}| \cdot \sum_{\mI \in L_{r} \backslash L_{r+1}}|\mI|}\notag \\
\leq &~ |\mH_{r-1}| B_{r-1} \cdot  \sqrt{\frac{1}{B_{r}}} \notag \\
\leq &~ \left\{
\begin{matrix}
|\mH_{r-1}| \cdot \sqrt{2T/n}  & r\geq 2\\
\sqrt{nT}& r = 1
\end{matrix}
\right..\label{eq:regret-last2}
\end{align}
The first step follows from Cauchy-Schwarz, the second step follows from (1) the length of an interval $\mI \in L_{r} \backslash L_{r+1}$ satisfies $|\mI| \in [B_{r}: B_{r-1})$ and (2) $\sum_{\mI \in L_{r} \backslash L_{r+1}}|\mI| \leq |\mH_{r-1}| \cdot B_{r-1}$ (see Lemma \ref{lem:obv-interval}). 
The last step follows from the choice of $\{B_r\}_{r\in [R]}$. Note we slightly abuse notation and set $|\mH_{0}| = 1$ and $B_{0} = T$.

We can bound the RHS of Eq.~\eqref{eq:regret-last1} using Eq.~\eqref{eq:regret-last2}. With probability at least $1-1/(nT)^{\omega(1)}$, one has
\begin{align}
 &~ \sum_{r=1}^{R-1}\sum_{\mI \in L_{r}\setminus L_{r+1}}O\left(\sqrt{|\mI|}\log^6 (nT)\right) + O(|\mH_{R-1}|)\notag\\
\leq &~ O\left( \sqrt{n T}\log^6(nT) + \sum_{r=2}^{R-1} |\mH_{r-1}|\cdot \sqrt{T/n}\log^6(nT)  + |\mH_{R-1}|\right)\notag\\
\leq &~ \wt{O}\left(\sqrt{T/n}\right),
\label{eq:regret-last3}
\end{align}
where the second step follows from Lemma \ref{lem:bad-size-full}, i.e., $|\mH_{r}| \leq \wt{O}(n)$ holds with high probability for any $r \in [R]$.

Combining Eq.~\eqref{eq:regret-last1} and Eq.~\eqref{eq:regret-last3}, we conclude 
\begin{align*}
\sum_{t \in [T]}\ell_{t}(i_t) -\sum_{t \in [T]}\ell_{t}(i^{*})  \leq \wt{O}\left(\sqrt{nT}\right) 
\end{align*}
holds with probability at least $1-1/(nT)^{\omega(1)}$.
We complete the proof here.
\end{proof}

Combining the regret analysis (Lemma \ref{lem:regret-full}) and the memory bounds (Lemma \ref{lem:memory}), we conclude the proof of Proposition \ref{prop:oblivious-n}. We complete the proof of Theorem \ref{thm:oblivious-main} by a grouping trick.

\begin{proof}[Proof of Theorem \ref{thm:oblivious-main}]
We partition the experts $[n]$ into $G = \wt{O}(S)$ groups, with $n' = \wt{O}(n/S)$ experts per group. 
We instantiate a meta-thread%
\footnote{here we use {\em meta-thread} to distinguish from the $R$ {\em threads} internal to Algorithm \ref{algo:full}.} of Algorithm \ref{algo:full} over each group $j \in [G]$, and each meta-thread is viewed as an individual meta expert. 
The final decision comes from running the MWU over the set of meta experts $[G]$. 
Let the optimal expert $i^{*} \in [n]$ be contained in the $j^{*}$-th group ($j^{*} \in [G]$), the regret guarantee follows from 
\begin{align*}
\sum_{t=1}^{T}\ell_t(i_t) - \sum_{t=1}^{T}\ell_t(i^{*}) = &~ \sum_{t=1}^{T}\ell_t(i_t) - \sum_{t=1}^{T}\ell_t(j^{*}) + \sum_{t=1}^{T}\ell_t(j^{*}) - \sum_{t=1}^{T}\ell_t(i^{*})\\
\leq &~ \wt{O}(\sqrt{T}) + \wt{O}(\sqrt{n' T}) =  \wt{O}\left(\sqrt{\frac{nT}{S}}\right),
\end{align*}
where the second step follows from the regret guarantee of MWU and Proposition \ref{prop:oblivious-n}.
To bound the memory, note the algorithm maintains $G$ threads of Algorithm \ref{algo:full} , each takes $\polylog(nT)$ space, so the total space used is at most $G \cdot \polylog(nT) \leq S$. We conclude the proof here.

\end{proof}

%% file: adaptive-upper.tex
\section{Vanishing regret against adaptive adversary}
\label{sec:adaptive-upper}

We have the regret guarantee against an adaptive adversary.

\Adaptive*

We focus on the case of $T \geq S$ and provide an algorithm that obtains a total regret of $\wt{O}\left(\frac{\sqrt{n}T}{S}\right)$ when $T \geq S$. 
This implies the regret guarantee of Theorem \ref{thm:adaptive-upper-main} over all range of $T$, using the same grouping trick as Theorem \ref{thm:oblivious-main}.
Concretely, we prove
\begin{proposition}
\label{prop:adaptive-upper}
Let $n, T \geq 1, \eps \in (1/\sqrt{n}, 1)$ and $S = \sqrt{n}/\eps$. Suppose $T \geq S$, then here is an online learning algorithm that uses up to $S \cdot\polylog(nT)$ space and achieves $\eps T \cdot \polylog(nT)$ regret against an adaptive adversary.
\end{proposition}

The pseudocode of the main algorithm is presented as Algorithm \ref{algo:adaptive}. It maintains two components, $\textsc{RandomExpert}$ and $\textsc{LongExpert}$, and runs {\sc IntervalRegret} over them.

Let $B = 1/\eps^2$ be the epoch size. 
For each epoch, the {\sc RandomExpert} (Algorithm \ref{algo:random-expert}) runs $\textsc{IntervalRegret}$ over $R = \lceil\log_4 (1/\eps^2)\rceil = \lceil\log_2 (1/\eps)\rceil$ procedures $\textsc{Rexp}_{r}$ ($r\in [R]$).
For any $r \in [R]$, $\textsc{Rexp}_r$ restarts every $1/\eps_r^2$ days, and during each restart, it runs MWU over $N_r = \eps_r^2\sqrt{n}\log^2(nT)/\eps$ randomly selected experts. 
We apply a common trick to make it work against an adaptive adversary:
$\textsc{Rexp}_{r}$ maintains $1/\eps_r^2$ different pools of experts and the decision of day $b$ comes from the $b$-th pool ($b \in [1/\eps_r^2]$).
This reduces an adaptive adversary to an oblivious adversary. 
From a high level, $\textsc{Rexp}_{r}$ is guaranteed to be $\eps_r$-competitive to all but the top $\eps\sqrt{n}/\eps_r^2$ experts during each restart.
%Roughly speaking, {\sc RandomExpert} is guaranteed to be at most $\eps_r$ worse than the top $\eps\sqrt{n}/\eps_r^2$-th expert.

The $\textsc{LongExpert}$ (Algorithm \ref{algo:random-expert}) aims to hedge the top experts. 
It maintains a pool $\mP$ that contains experts perform well during some past epochs and runs {\sc MonocarpicExpert} over $\mP$.
For each epoch, $\textsc{LongExpert}$ calls $\textsc{MaintainPool}_r$ ($r\in [R]$). For each $r \in [R]$, $\textsc{MaintainPool}_r$ restarts every $1/\eps_r^2$ days. During each restart, it samples a set of experts $O_r$ ($|O_r| = O(\sqrt{n}/\eps)$) and observes their loss for $1/\eps_r^2$ days. An expert $i \in O_r$ is said to survive and added to $\mP_r$ if its loss is significantly less than $\textsc{Rexp}_r$.
The pool $\mP_r$ is updated at the end of each epoch by removing experts that live for $\eps\sqrt{n}$ epochs.

%It samples $S = \Theta(\sqrt{n}/\eps)$ experts at the beginning of each epoch.
%It observes the loss of experts in $S$ and at the end of epoch, it compares with the performance of {\sc RandomExpert}. Intuitively, it maintains experts with high margin for long time. In particular, suppose an expert $i \in S$ average has margin $O(\eps_r)$, then the expert $i$ would be kept for $\eps_r\sqrt{n}$ epochs.

\begin{algorithm}[!htbp]
\caption{Algorithm against adaptive adversary}
\label{algo:adaptive}
\begin{algorithmic}[1]
%\For{$t=1,2,\ldots, \eps^2 T$}\Comment{Epoch $t$}
\State Run {\sc IntervalRegret} over $\textsc{RandomExpert}$ and $\textsc{LongExpert}$
%\EndFor
\end{algorithmic}
\end{algorithm}

\begin{algorithm}[!htbp]
\caption{$\textsc{RandomExpert}$}
\label{algo:random-expert}
\begin{algorithmic}[1]
\State Initialize $\textsc{Rexp}_r$, $\eps_r \leftarrow 2^{r-1}\eps$, $N_{r}\leftarrow \eps_r^2 \sqrt{n}\log^2(nT)/\eps$ ($\forall\,r \in [R]$) \Comment{$R = \lceil\log_2(1/\eps)\rceil$}
\For{$t=1,2,\ldots, \eps^2 T$} \Comment{Epoch $t$}
\State Run {\sc IntervalRegret} over $\textsc{Rexp}_{r}$ ($r \in [R]$) every $1/\eps^2$ days. \label{line:interval-random}
\EndFor\\
\Procedure{$\textsc{Rexp}_r$}{} \Comment{$r\in [R]$}
\For{$s= 1, 2, \ldots, \eps_r^2/\eps^2$} \Comment{Restart every $1/\eps_r^2$ days}
\State Sample $N_r \times 1/\eps_r^2$ experts $\{i_{\alpha, \beta}\}_{\alpha \in [N_r], \beta \in  [1/\eps_r^2] }$ 
\For{$b = 1,2, \ldots, 1/\eps_{r}^2$}
\State Using MWU to maintain weights over $\{i_{\alpha, \beta}\}_{\alpha \in [N_{r}]}$ ($\forall\beta \in [1/\eps_r^2]$)
\State Play $i_{\alpha(b), b}$  \Comment{The decision of the $b$-th pool}
\EndFor
\EndFor
\EndProcedure
\end{algorithmic}
\end{algorithm}

\begin{algorithm}[!htbp]
\caption{$\textsc{LongExpert}$}
\label{algo:long}
\begin{algorithmic}[1]
\State Run $\textsc{MaintainPool}_{r}$ ($r\in [R]$)
\State Run {\sc MonocarpicExpert} over $\mP$ \Comment{$\mP = \mP_1 \cup \cdots \cup \mP_R$}

%\Procedure{$\textsc{MaintainPool}$}{}
%\State Initiate $\mP \leftarrow \emptyset$
%\For{$t = 1,2,\ldots, \eps^2 T$} \Comment{Epoch $t$}
%\State Run $\textsc{MaintainPool}_r$ ($r\in [R]$)
%\State $\mP \leftarrow \mP \cup \mQ_1 \cup \cdots \cup \mQ_R$ 
%\State Evict expired experts in $\mP$. \Comment{An expert lives in $\mP$ for $\eps\sqrt{n}$ epochs }
%\EndFor \\
%\EndProcedure

\Procedure{$\textsc{MaintainPool}_r$}{} \Comment{$r\in [R]$}
%\State  $\mQ_r \leftarrow \emptyset$
\State $\mP_r \leftarrow \emptyset$ 
\For{$t=1,2,\ldots, \eps^2 T$} \Comment{Epoch $t$}
\For{$s = 1,2,\ldots, \eps_r^2/\eps^2$} \Comment{Restart $s$}
\State $O_{r} \leftarrow \textsc{Sample}(\mathcal{N}, \frac{1}{\eps\sqrt{n}})$ \Comment{Sample $\Theta(\sqrt{n}/\eps)$ experts}\label{line:long-sample}
\State Observe $O_r$ for $1/\eps^2_r$ days
\For{each expert $i \in O_r$}
\State $\mP_r \leftarrow \mP_r \cup \{i\}$ if $\mathcal{L}_{s}(i) <  \mathcal{L}_s(\textsc{Rexp}_r) - 2\log(nT)/\eps_r$. \label{line:keep}
\EndFor 
\EndFor
\State Evict expired experts in $\mP_r$ \Comment{An expert lives in $\mP$ for $\eps\sqrt{n}$ epochs }
\EndFor
%\State $P \leftarrow P \cup Q$
%\State Update membership of $P$
\EndProcedure

%\State $S \leftarrow \textsc{Sample}(\mathcal{N}, \frac{1}{\eps\sqrt{n}})$
%\For{$b=1,2,\ldots, 1/\eps^2$} 
%\State Run MWU over $\mP$ 
%\State Observe the loss of experts in $S$
%\EndFor
%\For{each expert $i \in S$}
%\State Let $r \in [0: R]$ be the small integer such that $\mathcal{L}_{t}(i) \geq  \mathcal{L}_t(\textsc{RandomExpert}) - 10\eps_r $
%\State Add $i$ to $\mP$ for $\eps_r\sqrt{n}$ epochs \Comment{Only for $r\geq 1$}
%\EndFor
%\State Update $\mP$
%\EndFor
\end{algorithmic}
\end{algorithm}

\paragraph{Notation} For any epoch $t \in [\eps^2 T]$, thread $r\in [R]$, restart $s \in [\eps_r^2/\eps^2]$ and day $\beta \in [1/\eps_r^2]$, let $i_{t, r, s, \beta}$ be the action taken by $\textsc{Rexp}_r$, and $\ell_{t, r, s, \beta}$ be the loss vector on the $\beta$-th day, the $s$-th restart, $t$-th epoch of thread $r$.
For any $k\in [n]$, let $i_{t, r, s, k}^{*} \in [n]$ be the top $k$-th expert during $s$-th restart, $t$-th epoch of thread $r$, with ties breaking arbitrarily and consistently.

\subsubsection*{Regret guarantee of {\sc RandomExpert}} 
We first provide the regret guaranteee of {\sc RandomExpert} and we aim for a high probability bound.

\begin{lemma}
\label{lem:random}
With probability at least $1 - 1/(nT)^{\omega(1)}$, for each thread $r \in [R]$, each epoch $t \in [\eps^2 T]$ and restart $s \in [\eps_r^2/\eps^2]$, one has
\begin{align}
\sum_{\beta=1}^{1/\eps_r^2}\ell_{t,r,s,\beta}(i_{t,r,s, \beta}) - \sum_{\beta=1}^{1/\eps_r^2}\ell_{t,r, s,\beta}(i_{t,r,s, k}^{*}) \leq \eps_r^{-1}\log(nT) \label{eq:compare-k}
\end{align}
holds for any $k \geq \eps \sqrt{n} /\eps_r^2$, against an adaptive adversary.
\end{lemma}

\begin{proof}
For an epoch $t\in [\eps^2 T]$, a thread $r\in [R]$ and a restart $s\in [\eps_r^2/\eps^2]$, let $B_r = 1/\eps_r^2$. It would be convenient to simplify the notation and remove the subscript on $t,r,s$ when there is no confusion.

The high level idea has appeared in previous work, e.g. \cite{cesa2006prediction, gonen2019private}. By taking decisions from different copies on different days, the decision satisfies 
\begin{align}
\Pr[i_\beta = i | i_1, \ldots, i_{\beta-1}, \ell_{1}, \ldots, \ell_{\beta-1}] = \Pr[i_{\beta} = i | \ell_1, \ldots, \ell_{\beta-1}]. \label{eq:not-adaptive}
\end{align}That is, the dependence of $i_\beta$ on previous actions is explained away by the dependence based on $\ell_{1}, \ldots, \ell_{\beta-1}$. This intuitively means that adaptivity (adaptively chosen loss sequence on previous actions) does not help. 
Let the distribution of Eq.~\eqref{eq:not-adaptive} be $P(\ell_{1}, \ldots, \ell_{\beta -1})$ and we abbreviated as $p_\beta$ when there is no confusion.

\paragraph{Step 1. Reduction to expected regret.} 
Let 
\[
X_\beta = \ell_\beta(i_\beta) - \E_{i_\beta \sim p_\beta}[\ell_\beta(i_\beta)] \quad \text{and} \quad Y_{\beta} = X_1 + \cdots + X_\beta, \forall \beta \in [B]
\]
$\{Y_\beta\}_{\beta \in [B]}$ forms a martingale with respect to the randomness of $\{\ell_\beta\}_{\beta\in [B]}$ and $\{i_\beta\}_{\beta\in [B]}$, and it has bounded difference of at most $2$ (i.e., $|X_{\beta}| \leq 2$), by Azuma-Hoeffding bound, one has
\begin{align*}
\Pr\left[|Y_B| \geq \frac{1}{2}\sqrt{B}\log(nT)\right] \leq 2\exp(-B\log^2(nT)/32B) \leq 1/(nT)^{\omega(1)}.
\end{align*}
Hence, with probability at least $1 - 1/(nT)^{\omega(1)}$ over the choice of $i_t$, one has
\begin{align}
\sum_{\beta=1}^{B}\left(\ell_\beta(i_\beta) - \E_{i_\beta \sim p_\beta}[\ell_t(i_\beta)]\right) \leq \frac{1}{2}\sqrt{B}\log(nT).
\label{eq:step-1}
\end{align}

\paragraph{Step 2. Reduction to oblivious adversary.}  
We next show it suffices to consider the regret w.r.t. an oblivious adversary.
% Note that the decision $i_{\beta}$ depends only on $\ell_{1}, \ldots, \ell_{\beta-1}$, i.e., 
%\[
%\Pr[i_{\beta} = i | i_{1}, \ldots, i_{\beta-1}, \ell_{1}, \ldots, \ell_{\beta -1}] = \Pr[i_{\beta} = i | \ell_{1}, \ldots, \ell_{\beta -1}].
%\]
%Denote the distribution of RHS as $P(\ell_{1}, \ldots, \ell_{\beta -1})$, one has
Recall that $i^{*}_k$ is the top $k$-th expert.
Condition on the high probability event of Step 1, one has
\begin{align}
&~ \sum_{\beta=1}^{B}\left(\ell_\beta(i_\beta) - \ell_\beta(i_k^{*})\right) \notag \\
\leq &~\sum_{\beta=1}^{B}\left( \E_{i_\beta \sim p_\beta}[\ell_\beta(i_\beta)] - \ell_\beta(i_k^{*}) + \frac{1}{2}\sqrt{B}\log(nT)\right) \notag\\
\leq &~ \sup_{\ell_1} \E_{i_1 \sim p_1}\cdots \sup_{\ell_{B-1}}\E_{i_{B-1} \sim p_{B-1}} \sup_{\ell_{B}}\E_{i_B \sim p_{B}} \left[\sum_{\beta=1}^{B}\ell_\beta(i_\beta) - \sum_{\beta=1}^{B}\ell_\beta(i_{k}^{*})\right] + \frac{1}{2}\sqrt{B}\log(nT) \notag\\
= &~ \sup_{\ell_1} \E_{i_1 \sim p_1}\cdots \sup_{\ell_{B-1}} \sup_{\ell_B}\E_{i_{B-1} \sim p_{B-1}} \E_{i_{B} \sim p_{B}} \left[\sum_{\beta=1}^{B}\ell_\beta(i_\beta) - \sum_{\beta=1}^{B}\ell_\beta(i_{k}^{*})\right] + \frac{1}{2}\sqrt{B}\log(nT) \notag \\
\vdots \notag \\
= &~ \sup_{\ell_1} \cdots \sup_{\ell_{B}}\E_{i_1 \sim p_1} \cdots \E_{i_B\sim p_B} \left[\sum_{\beta=1}^{B}\ell_\beta(i_\beta) - \sum_{\beta=1}^{B}\ell_\beta(i_k^{*})\right] + \frac{1}{2}\sqrt{B}\log(nT). \label{eq:adaptive-oblivious}
\end{align}
The first step follows from Eq.~\eqref{eq:step-1}. 
%We expand the definition in the second step and use the fact that $i_{t}$ does not depend on $i_1, \ldots, i_{t-1}$ when conditioned on $\ell_1, \ldots, \ell_{t-1}$.
In the third step, we can exchange the expectation and the supreme because $i_\beta$ does not depend on $i_1, \ldots,i_{\beta-1}$ when conditioned on $\ell_1, \ldots, \ell_{\beta-1}$.

\paragraph{Step 3. Reduction to one copy}
Now fix a loss sequence $\ell_1,\ldots, \ell_{B}$, we bound the expected regret. Consider an algorithm that maintains only one copy of experts: It samples $N_r = \eps_r^2 \sqrt{n}\log^2(nT)/\eps$ experts at the beginning and runs MWU over them for $B_r = 1/\eps_r^2$ days. 
Let $\hat{i}_\beta$ be the action taken by this (single copy) algorithm at day $\beta$ and it is easy to see 
\[
\Pr\left[\hat{i}_\beta = i | \ell_1, \ldots \ell_{\beta-1}\right] =  \Pr[i_\beta = i | \ell_1, \ldots \ell_{\beta-1}] .
\]
Hence, fixing the loss sequence, it suffices to bound
\begin{align}
\textsf{reg} = \E_{i_t \sim p_1} \cdots \E_{i_B \sim p_{B}} \left[\sum_{\beta=1}^{B}\ell_t(i_t) - \sum_{\beta=1}^{B}\ell_\beta(i_k^{*})\right] = \E\left[\sum_{\beta=1}^{B}\ell_\beta(\hat{i}_\beta) - \sum_{\beta=1}^{B}\ell_\beta(i_k^{*})\right].\label{eq:regret-one-copy}
\end{align}

\paragraph{Step 4. Bounding the regret.} Finally we bound the regret (i.e. RHS of Eq.~\eqref{eq:regret-one-copy}). Note that the loss sequence is fixed and so are the experts $i_1^{*}, \ldots, i_{n}^{*}$. The single copy algorithm samples $N_r$ experts at the beginning, and therefore, with probability at least 
\[
1 - \left(1 - \frac{\eps}{\eps_r^2\sqrt{n}}\right)^{N_r} = 1 - \left(1 - \frac{\eps}{\eps_r^2\sqrt{n}}\right)^{\eps_r^2\sqrt{n}\log^2(nT)/\eps} \geq 1 - 1/(nT)^{\omega(1)},
\]
the algorithm samples an expert $i \in \{i_1^{*}, \ldots, i^{*}_{\eps \sqrt{n}/\eps_r^2}\}$. Condition on this, for any $k \geq \eps\sqrt{n}/\eps_r^2$, one has
\begin{align}
\E\left[\sum_{\beta=1}^{B}\ell_{\beta}(\hat{i}_{\beta}) - \sum_{\beta=1}^{B}\ell_{\beta}(i_{k}^{*})\right] = &~ 
\E\left[\sum_{\beta=1}^{1/\eps_r^2}\ell_{\beta}(\hat{i}_{\beta}) - \sum_{\beta=1}^{1/\eps_r^2}\ell_{\beta}(i_{k}^{*})\right]\notag\\
\leq &~  \E\left[\sum_{\beta=1}^{1/\eps_r^2}\ell_{\beta}(\hat{i}_{\beta}) - \sum_{\beta=1}^{1/\eps_r^2}\ell_{\beta}(i)\right] \leq O(\eps_r^{-1}\sqrt{\log(nT)}).\label{eq:oblivious-sample}
\end{align}
Here we plug in $B = B_r = 1/\eps_r^2$ in the first step. As the loss can be at most $1/\eps_r^2$, one can prove the expected regret is at most $O(\eps_r^{-1}\sqrt{\log(nT)})$ (without conditioning on the high probability event).

Combining the above four steps, i.e. Eq.~\eqref{eq:step-1}\eqref{eq:adaptive-oblivious}\eqref{eq:regret-one-copy}\eqref{eq:oblivious-sample}, we conclude the proof.

%Finally, we bound the regret and prove
%\begin{align*}
%\E[\sum_{t=1}^{1/\eps^2_r} \ell_t(i_t) - \frac{1}{N_r}\sum_{s=1}^{N_r}\sum_{t=1}^{1/\eps^2_r}\ell_t(i_s)] \leq O(\sqrt{\eps_r \log n}).
%\end{align*}
%The expectation is taken over the random choice of set $N_r$ experts. We prove with probability at least $s\eps_r^2/\eps\sqrt{n}$, $N_r$ contains one of $i_1, \ldots, i_s$. The failure is computated as
%\begin{align*}
%(1 - \frac{s}{n})^{N_r} = (1 - \frac{s}{n})^{\eps_r^2 \sqrt{n}/\eps} \approx
%\end{align*}
%Hence, we have
%\begin{align*}
%\sum_{t=1}^{1/\eps^2_r} \ell_t(i_t) \leq &~ \sum_{s=1}^{N_r}\sum_{t=1}^{1/\eps^2_r} \ell_t(i_t) p_{s} + O(\sqrt{\eps_r \log n})\\
%\leq &~\frac{1}{N_r}\sum_{s=1}^{N_r}\sum_{t=1}^{1/\eps^2_r}\ell_t(i_s) + O(\sqrt{\eps_r \log n})
%\end{align*}
\end{proof}

\subsubsection*{Memory bounds}
We first provide an upper bound on the memory of our main algorithm.
\begin{lemma}
\label{lem:memory-adaptive}
With probability at least $1 - 1/(nT)^{\omega(1)}$, the memory used by Algorithm \ref{algo:adaptive} is at most $ O(\sqrt{n}\log^2(nT)\log(1/\eps)/\eps)$ words.
\end{lemma}
\begin{proof}
We bound the memory usage of {\sc RandomExpert} and {\sc LongExert} separately.
The {\sc RandomExpert} runs {\sc IntervalRegret} over $R = O(\log(1/\eps))$ threads $\textsc{Rexp}_r$, where each thread $\textsc{Rexp}_r$ samples $N_r \cdot 1/\eps_r^2 = \eps_r^2\sqrt{n}\log^2(nT)/\eps \cdot 1/\eps_r^2 = \sqrt{n}\log^2(nT)/\eps$ experts in $[n]$ and maintain their weight during each restart. Hence the memory usage of {\sc RandomExpert} is bounded by $O(\sqrt{n}\log^2(nT)\log(1/\eps)/\eps)$.

The crucial part is bounding the memory of {\sc LongExpert} procedure. For any epoch $t \in [\eps^2T]$, thread $r \in [R]$, restart $s\in [\eps_r^2/\eps^2]$, let $O_{t,r,s}$ be the sample set (Line \ref{line:long-sample} of Algorithm \ref{algo:long}).
We condition on the high probability event of Lemma \ref{lem:random} and bound the number of experts that survive after Line \ref{line:keep}. By Lemma \ref{lem:random}, any expert $i \in [n]\setminus\{i_{t,r, s,1}^{*}, \ldots, i_{t,r, s,\eps\sqrt{n}/\eps_r^2}^{*}\}$ satisfies 
\[
\sum_{\beta=1}^{1/\eps_r^2}\ell_{t,r, s,\beta}(i_{t,r, s,\beta}) - \sum_{\beta=1}^{1/\eps_r^2}\ell_{t,s,\beta, r}(i) \leq \eps_r^{-1}\log(nT),
\]
hence it would not survive due to the eviction rule of {\sc LongExpert} (see Line \ref{line:keep} of Algorithm \ref{algo:long}). It remains to bound the size of the intersection $|O_{t,r, s} \cap \{i_{t,r,s,1}^{*}, \ldots, i_{t,r, s,\eps\sqrt{n}/\eps_r^2}^{*}\}|$. Note that each expert is sampled independently with probability $\frac{1}{\eps\sqrt{n}}$, i.e.,
\[
\E\left[\left|O_{t,r, s} \cap \{i_{t,r,s,1}^{*}, \ldots, i^{*}_{t,r,s,\eps\sqrt{n}/\eps_r^2}\}\right|\right] = \eps\sqrt{n}/\eps_r^2 \cdot \frac{1}{\eps\sqrt{n}} = 1/\eps_r^2,
\]
therefore, by Chernoff bound, one has
\begin{align*}
\Pr\left[|O_{t,r, s} \cap \{i_{t,r, s,1}^{*}, \ldots, i^{*}_{t,r, s,\eps\sqrt{n}/\eps_r^2}\}| \geq \log^{2}(nT)/\eps_r^2\right] \leq \exp(-\log^{2}(nT)/3\eps_r^2) \leq 1/(nT)^{\omega(1)}.
\end{align*}
Taking a union bound over $s \in [\eps_r^2/\eps^2]$, the number of survived experts at one epoch is at most $\log^2(nT)/\eps^2$. Taking a union bound over $t \in [\eps^2 T]$ and note an expert stays at $\mP_r$ for $\eps\sqrt{n}$ epochs, the size of $\mP_r$ satisfies $|\mP_r| \leq \log^2(nT)/\eps^2\cdot \eps \sqrt{n} = \sqrt{n}\log^2(nT)/\eps$. The size of $\mP$ satisfies
\[
|\mP| = |\mP_1 \cup \cdots \mP_{R}| \leq \sqrt{n}\log^2(nT)/\eps \cdot R =\sqrt{n}\log^2(nT)\log(1/\eps)/\eps. 
\]
\end{proof}

\subsubsection*{Regret analysis}

We next bound the regret with respect to each expert separately. For any day $t \in [T]$, let $i_t^{(L)}$ be the action of {\sc LongExpert} at day $t$ and $i_t^{(R)}$ be the action of {\sc RandomExpert}. Let $i_t^{(r)}$ be the action of $\textsc{Rexp}_r$ at day $t$.

%For each epoch $t \in [\eps^2 T]$, day $b \in [1/\eps^2]$, let $i_{t, b}$ be the action of Algorithm \ref{algo:adaptive}, $i_{t, b}^{(L)}$ be the action of {\sc LongExpert} and $i_{t, b}^{(R)}$ be the action taken by {\sc RandomExpert}. 

\begin{lemma}
\label{lem:regret-adaptive}
With probability at least $1 - 1/(nT)^{\omega}$, for any expert $i \in [n]$, the regret is bounded as
\begin{align*}
%\sum_{t=1}^{\eps^2 T}\sum_{b=1}^{1/\eps^2}\ell_{t, b}(i_{t, b}) - \sum_{t=1}^{\eps^2 T}\sum_{b=1}^{1/\eps^2}\ell_{t, b}(i) \leq O(\eps T).
\sum_{t=1}^{T}(\ell_{t}(i_{t}) -\ell_{t}(i)) \leq O(\eps T\log^2(nT)).
\end{align*}
\end{lemma}
\begin{proof}
We fix the randomness of {\sc RandomExpert}.
%and condition on the high probability event of Lemma~\ref{lem:random}. 
The proof proceeds by assigning each day $[T]$ to either {\sc RandomExpert} or {\sc LongExpert}. 
In particular, we partition the entire sequence $[T]$ into intervals as follows. The first interval lasts till the first time that expert $i$ is added to $\mP$, the second interval terminates when $i$ gets expired and removed from $\mP$, and we repeat the above process till the end of sequence. We use $E_R$ to denote the collections of odd intervals ($\mP$ does not contain $i$ in odd intervals), and $E_{L}$ to denote the collections of even intervals ($\mP$ contains $i$ in an even interval).
%we partition all times into intervals as follows: the first lasts until the first time that $\mP$ contains expert $i$ during $\mI$; the second  lasts until the first time that the expert expires; etc., with odd intervals (denoted $E_R$) not having $i$, and even intervals (denoted $E_L$) having $i$.}
%Let $E_{L}$ contain all time intervals $\mI \subseteq [T]$ such that $\mP$ contains expert $i$ during $\mI$. Note that a time interval $\mI\in E_{L}$ is of length at least $\eps\sqrt{n} \cdot 1/\eps^2 = \sqrt{n}/\eps$, since an expert is expired after $\eps \sqrt{n}$ epochs. 
%Let $E_{R}$ contain the rest of time intervals in $[T]\setminus E_{L}$. 
For proof convenience, if an interval $\mI\in E_{R}$ spans more than one epoch, we split it according to the epoch boundary (so the length of $\mI$ never exceeds one epoch). 
For an interval $\mI \in E_{R}$, suppose it is contained in epoch $t(\mI) \in [\eps^2 T]$, note that $\mI$ is not necessarily an entire epoch since expert $i$ could join $\mP$ in the middle of epoch $t(\mI)$.

\paragraph{Notation} We first introduce some definitions. 
\[
K_1 = [\eps^{2} T], K_2 = [\eps^2 T] \times [4], \ldots, K_R = [\eps^2 T] \times \underbrace{[4] \times \cdots \times [4]}_{R-1}  \quad \text{and}\quad K = K_1 \cup \cdots \cup K_R.
\]
For a timestep $a \in K$, let $r(a)$ be defined as $a \in K_{r(a)}$. The timestep $a$ uniquely identifies a restart of $\textsc{Rexp}_{r(a)}$, i.e., the $\sum_{r=2}^{r(a)}a_{r}4^{r-2}$-th restart at the $a_1$-th epoch.
Let $\mL_a(i) = \sum_{t\in a}\ell_t(i)$ be the loss of timestep $a$, here we slightly abuse of notation and also use $a \in K$ to denote the entire period of the timestep/restart $a$.
For any interval $\mI \in E_R$, it starts from the beginning of epoch $t(\mI) \in [\eps^2 T]$ and ends at some timestep $a(\mI) \in K$. 
We further split $\mI$ according to the 4-base representation of $a(\mI)$: $\mI = a_1 \cup \cdots \cup a_{\tau(\mI)}$.

\paragraph{Split interval} The key step is to split the interval $\mI \in E_{R}$ according to the performance of expert $i$ and $\textsc{Rexp}_{r}$.
For any $\mI \in E_{R}$ with $\mI = a_1 \cup \cdots \cup a_{\tau(\mI)}$, we call {\sc AssignInterval}($a_\tau, i$) for each sub-interval $a_{\tau}$ ($\tau \in [\tau(\mI)]$). We note the {\sc AssignInterval} procedure is only used in the analysis. It returns a collection of timesteps  $M(a_{\tau}) \subseteq K$ that exactly spans $a_{\tau}$. 
Let $M(I) = \cup_{\tau \in [\tau(\mI)]}M(a_{\tau})$ be the union of the collections, we have
\begin{align}
\bigcup_{a \in M(\mI)}a = \mI.\label{eq:interval-cover}
\end{align}

The {\sc AssignInterval} procedure (pseudocode in Algorithm \ref{algo:assign-interval}) works as follow. It receives a timestep $a \in K$ and a target expert $i$, and it outputs a collection of disjoint timesteps $M(a)$ that exactly covers $a$.
{\sc AssignInterval} is a recursive procedure. It first tries to return the entire interval $a$ by checking if $\textsc{Rexp}_{r(a)}$ obtains good performance (Line \ref{line:good} of Algorithm \ref{algo:assign-interval}). 
If not, it splits the timestep $a$ into four parts $\{(a,j) \in K_{r(a) + 1}: j \in [4]$\} and seeks help from $\textsc{Rexp}_{r(a)+1}$, by recursively calling $\textsc{AssignInterval}((a, j), i)$ ($j \in [4]$).
Note it is guaranteed to terminate since $\mL_{a}(i) \geq 0 > \mL_a(i^{(R)}) - 2\log(nT)/\eps_R$ for any $a \in K_R$.
%For any $t \in [A]$, the set $\mI(t)$ contains a set of interval and 

\begin{algorithm}[!hbtp]
\caption{$\textsc{AssignInterval}(a, i)$ (Note: only used in analysis) \Comment{Timestep $a$, expert $i$}}
\label{algo:assign-interval}
\begin{algorithmic}[1]
\If{$\mL_{a}(i) \geq \mL_a(i^{(r(a))}) - 2\log(nT)/\eps_{r(a)}$}\label{line:good}
\State Return $a$
\Else 
\State Call and return the union of $\textsc{AssignInterval}((a, j), i)$ for $j \in [4]$ 
\EndIf
\end{algorithmic}
\end{algorithm}

Now we can bound the regret, first we split the entire sequence according to $E_{L}$ and $E_R$. With probability at least $1-1/(nT)^{\omega(1)}$ over the randomness of {\sc IntervalRegret} of the main algorithm, one has 
\begin{align}
&~\sum_{t=1}^{T}(\ell_{t}(i_{t}) -\ell_{t}(i))\notag \\
= &~ \sum_{\mI\in E_{R}}\sum_{t\in \mI}(\ell_t(i_t) - \ell_t(i)) + \sum_{\mI\in E_L}\sum_{t\in \mI}(\ell_t(i_t) - \ell_t(i)) \notag\\
\leq &~ \sum_{\mI\in E_R}\sum_{t\in \mI}(\ell_{t}(i^{(R)}_{t}) -\ell_{t}(i))+ \sum_{\mI\in E_L}\sum_{t\in \mI}(\ell_{t}(i_{t}^{(L)}) -\ell_{t}(i)) + \sum_{\mI \in E_{R}\cup E_L}O(\sqrt{|\mI|}\log(nT))\notag \\
\leq &~ \underbrace{\sum_{\mI\in E_R}\sum_{t\in \mI}(\ell_{t}(i^{(R)}_{t}) -\ell_{t}(i))}_{\text{($R$-regret)}}+ \underbrace{\sum_{\mI\in E_L}\sum_{t\in \mI}(\ell_{t}(i_{t}^{(L)}) -\ell_{t}(i))}_{\text{($L$-regret)}} + 
 O(\eps T\log(nT)) \label{eq:split}.
\end{align}
The second step holds since Algorithm \ref{algo:adaptive} runs {\sc IntervalRegret} over {\sc LongExpert} and {\sc RandomExpert}\, the last step follows from Cauchy-Schwarz and the total number of interval in $E_{R} \cup E_{L}$ is at most $\eps^2 T$.

We bound the two sums in the RHS of~\eqref{eq:split} separately. ($L$-regret) is easy to deal with as $i \in \mP$ during any interval $\mI \in E_L$ (this holds by definition), hence, with probability at least $1- 1/(nT)^{\omega(1)}$,
\begin{align}
\sum_{\mI\in E_{L}}\sum_{t\in \mI}(\ell_{t}(i_{t}^{(L)}) -\ell_{t}(i)) \leq O\left(\sum_{\mI\in E_{L}}\sqrt{|\mI|}\log(nT)\right) \leq O(\eps T\log(nT)). \label{eq:long}
\end{align}
The first step follows from the regret guarantee of {\sc MonocarpicExpert} and the second step follows from Cauchy-Schwarz.

To bound ($L$-regret), we split it according to $M(\mI)$. For each thread $r\in [R]$, epoch $t\in [\eps^2 T]$ such that there exists an interval $\mI \in E_{R}$ with $t(I) = t$, define $A_{t, r}$ to be the collection of $r$-th level timestep/restart in $M(\mI)$, i.e.,
\[
A_{t, r} = \{a: a \in M(\mI), a\in K_r \},
\]
and $N_{t,r} = |A_{t,r}|$. 
%That is, $N_{t,r}$ is the total number of $r$-th level timestep/restart in $M(\mI)$, where epoch $t = t(\mI)$ includes some interval $\mI \in E_{R}$. 
We set $A_{t, r} = \emptyset$ and $N_{t, r} = 0$ if epoch $t$ is contained in $E_{L}$.

With probability at least $1- 1/(nT)^{\omega(1)}$, we have
\begin{align}
\sum_{\mI\in E_{R}}\sum_{t\in \mI}(\ell_{t}(i^{(R)}_{t}) -\ell_{t}(i)) = &~\sum_{\mI\in E_{R}}\sum_{a \in M(\mI)}\sum_{t\in a}(\ell_{a}(i^{(R)}_{t}) -\ell_{t}(i)) \notag \\
= &~ \sum_{\mI\in E_{R}}\sum_{a \in M(\mI)}\sum_{t\in a}(\ell_{t}(i^{(R)}_{t}) -\ell_{t}(i^{(r(a))}_t) + \ell_{t}(i^{(r(a))}_t) -\ell_{t}(i)) \notag \\
\leq &~ \sum_{\mI\in E_R}\sum_{a \in M(\mI)} O(\log(nT)/\eps_{r(a)})\notag \\
= &~ \sum_{t=1}^{\eps^2 T} \sum_{r=1}^{R}\frac{N_{t,r}}{\eps_{r}} \cdot O(\log(nT)) \label{eq:random}
\end{align}
The first step holds since $M(\mI)$ covers $\mI$ (see Eq.~\eqref{eq:interval-cover}).
The third step follows from the interval regret guarantee of {\sc RandomExpert} (Line \ref{line:interval-random} of Algorithm \ref{algo:random-expert}) and the {\sc AssignInterval} procedure (Line \ref{line:good} of Algorithm \ref{algo:assign-interval}). 
%and in the last step, 
%we define $N_{t,r}$ as the number of segment $\mI' \in M(\mI)$ of epoch $t$. 

It remains to bound the RHS of Eq.~\eqref{eq:random} and we bound $\sum_{t=1}^{\eps^2 T} \frac{N_{t,r}}{\eps_{r}}$ ($r \in [2:R]$) separately. 

For any epoch $t \in [\eps^2 T]$ and thread $r \in [R]$, define
\begin{align}
\label{eq:def-b}
B_{t, r} = \{a \in K_r: a_1 = t, i \notin \mP_{r, a} \text{ and } \mL_a(i) < \mL_a(i^{(r)}) - 2\log(nT)/\eps_r \}
\end{align}
where $\mP_{r, a}$ is the pool maintained by thread $r$ at the beginning of timestep $a$.
In other words, $B_{t,r}$ contains all restarts $a$ in epoch $t$ such that the performance of expert $i$ is much better than $\textsc{Rexp}_r$, and expert $i$ is not at pool $\mP_{r}$ yet. Let $M_{t, r} = |B_{t,r}|$. The random variable $M_{t,r}$ is easily to handle with than $N_{t,r}$. First, we observe
\begin{claim}
\label{claim:random-1}
For any $t\in [\eps^2 T], r \geq 2$, $M_{t,r -1} \geq \frac{N_{t,r} - 3}{4}$.
\end{claim}
\begin{proof}
For any timestep $a = a_1\ldots a_{r} \in A_{t,r}$, we have $a'=a_{1}\ldots a_{r-1}\in B_{t,r-1}$ due to Line \ref{line:good} of Algorithm \ref{algo:assign-interval}. The only exception is that $\textsc{AssignInterval}$ is initiated with $a$, this happens only $3$ times per epoch. Hence we have $M_{t,r -1} \geq \frac{N_{t,r} - 3}{4}$.
\end{proof}

We next bound the term $\sum_{t=1}^{\eps^2 T}M_{t,r}$.
\begin{claim}
\label{claim:random-2}
For any $t\in [\eps^2 T], r \in [R]$, with probability at least $1-1/(nT)^{\omega}$, one has 
$\sum_{t=1}^{\eps^2 T}M_{t,r} \leq  O(\eps^2 T\log(nT))$.
\end{claim}
\begin{proof}
For any timestep $a \in B_{t,r}$, with probability $\frac{1}{\eps\sqrt{n}}$ (Line \ref{line:long-sample} of {\sc LongExpert}), {\sc LongExert} samples the expert $i$ into $O_{r}$ (which is not known to the adversary). By Eq.~\eqref{eq:def-b} and Line~\ref{line:keep} of {\sc LongExpert}, the expert $i$ would then be added to $\mP$ and live for $\eps\sqrt{n}$ epochs. Therefore, with probability at least $1 - 1/(nT)^{\omega(1)}$, it takes no more than $O(\eps\sqrt{n}\log(nT))$ timesteps in $\bigcup_{t\in [\eps^2 T]}B_{t,r}$ to add expert $i$ into the pool for $\eps\sqrt{n}$ epochs. The later happens at most $(\eps^2 T + \eps\sqrt{n})/\eps\sqrt{n} = O(\frac{\eps T}{\sqrt{n}})$ times (note we require $\eps^2 T \geq \eps\sqrt{n}$ here), and therefore, 
\[
\sum_{t=1}^{\eps^2 T}M_{t,r}  = \sum_{t=1}^{\eps^2 T}|B_{t,r}| \leq O(\eps\sqrt{n}\log(nT)) \cdot O(\frac{\eps T}{\sqrt{n}}) = O(\eps^2 T\log(nT)).\qedhere
\]
\end{proof}

Combining Claim \ref{claim:random-1} and Claim \ref{claim:random-2}, we conclude
\begin{align}
\label{eq:bound-ntr}
\sum_{t=1}^{\eps^2 T}\frac{N_{t,r}}{\eps_r} \leq\sum_{t=1}^{\eps^2 T}\frac{N_{t,r}}{\eps} = 3\eps T +  \sum_{t=1}^{\eps^2 T}\frac{4M_{t,r}}{\eps} \leq O(\eps T\log(nT)).
\end{align}
holds for any $r \geq 2$, and for $r = 1$, $N_{t,1} \leq 1$. Combining Eq.~\eqref{eq:split}\eqref{eq:long}\eqref{eq:random}\eqref{eq:bound-ntr}, one has
\[
\sum_{t=1}^{T}(\ell_{t}(i_{t}) -\ell_{t}(i)) \leq O(\eps T\log^2(nT)).
\]
We complete the proof here.\qedhere

%Recall we fix the randomness of {\sc RandomExpert}, the random variable $N_{t,r}$ depends only on the randomness of Line \ref{line:long-sample} of {\sc LongExpert} as well as the loss sequence chosen by adversary.
%For any timestep $a \in A_{t,r}$, we 

%Let $t\in [\eps^2 T]$, $r\geq 2$ and $N_{t,r} \geq 2$, then there are at least $N_{t,r}/4$ restart of $(r-1)$-thread, such that once $i$ is sampled into $O_{r}$, it would be included and survive for $\eps\sqrt{n}$ epochs. Since each expert are sampled with probability $1/\eps\sqrt{n}$, we have that
%\[
%\frac{1}{\sqrt{n}} \cdot \frac{\sum_{t}N_{t,r}}{\eps_{r}} \leq \frac{1}{\eps\sqrt{n}} \cdot \sum_{t}N_{t,r} \lesssim \frac{\eps^2 T}{\eps \sqrt{n}} = \frac{\eps T}{\sqrt{n}}
%\]
\end{proof}

Combining the regret analyse (Lemma \ref{lem:regret-adaptive}) and the memory bound (Lemma \ref{lem:memory-adaptive}), we have proved Proposition \ref{prop:adaptive-upper}.

\begin{proof}[Proof of Theorem \ref{thm:adaptive-upper-main}]
The Proposition \ref{prop:adaptive-upper} guarantees $\wt{O}(\frac{\sqrt{n}T}{S})$ regret when $T \geq S$. For the corner case of $T \in [\frac{n}{S}, S]$ (note it suffices to consider $\wt{S} \geq \Omega(\sqrt{n})$ due to our lower bound), our goal is to obtain a total regret of $\wt{O}(\sqrt{\frac{nT}{S}})$. 
This can be obtained by the grouping trick. 
In particular, we partition the experts $[n]$ into $G = O(\frac{S}{T})$ group, with $n' = O(\frac{n}{G}) = O(\frac{nT}{S}) \in [1, n]$ experts per group.
We run Algorithm \ref{algo:adaptive} over each group $j \in [G]$ using space $S' = S/G = O(T)$. By viewing each group as a meta expert, we run MWU over the set of meta experts $[G]$.  
%Note each group contains $n' = O(\frac{n}{G}) = O(\frac{nT}{S}) \in [1, n]$ experts, and the space available is $S' = O(\frac{S}{G}) = O(T)$.
For any expert $i \in [n]$, suppose it is contained in the $j(i)$-th group ($j(i) \in [G]$). After $T$ days, one has
\begin{align*}
\sum_{t=1}^{T}\ell_t(i_t) - \sum_{t=1}^{T}\ell_t(i) = &~ \sum_{t=1}^{T}\ell_t(i_t) - \sum_{t=1}^{T}\ell_t(j(i)) + \sum_{t=1}^{T}\ell_t(j(i)) - \sum_{t=1}^{T}\ell_t(i)  \\
\leq &~ \wt{O}\left(\sqrt{T}\right) + \wt{O}\left(\frac{\sqrt{n'}T}{S'}\right) = \wt{O}\left(\sqrt{T}\right) + \wt{O}\left(\frac{\sqrt{nT/S}\cdot T}{T}\right) = \wt{O}\left(\sqrt{\frac{nT}{S}}\right).
\end{align*}
Here the second step follows from the regret guarantee of MWU and Proposition \ref{prop:adaptive-upper}.
We complete the proof here.

%Therefore, apply Proposition \ref{prop:adaptive-upper},  after $T$ days, one obtains a total regret of 
%\[
%\wt{O}\left(\frac{\sqrt{n'}T}{S'}\right) = \wt{O}\left(\frac{\sqrt{nT/S}\cdot T}{T}\right)= \tilde{O}\left(\sqrt{\frac{nT}{S}}\right).
%\]
\end{proof}

%% file: adaptive-lower.tex
\section{Lower bound of adaptive adversary}
\label{sec:adaptive-lower}

\begin{theorem}[Lower bound, formal version of Theorem \ref{thm:lower-main}]
\label{thm:lower}
Let $n, T \geq 1$ and $\eps \in (\frac{\log^4 n}{\sqrt{n}}, \frac{1}{2})$, there is no online learning algorithm with space $o(\frac{\sqrt{n}}{\eps \log^3 n})$ can achieve $\frac{\eps^2}{20} T$ regret against an adaptive adversary.
\end{theorem}

\paragraph{Hard instance}

Let $M = \eps\sqrt{n}$, divide experts $[n]$ into $N = n/M = \sqrt{n}/\eps$ blocks and each block contains $M$ experts.  
At the beginning, the adversary samples $N$ independent copies of \textsc{SetDisjointness} instance (each has size $M$), where each instance is drawn from the following distribution $\mu = (\mu_A, \mu_B)$
\begin{itemize}
\item Sample an index $i^{*} \in [M]$ uniformly at random and set $(x_{A, i^{*}},x_{B, i^{*}}) = (1,1)$
\item For the rest of coordinates, $\frac{M-1}{3}$ of them to be $(0,0)$, $(1, 0)$, $(0,1)$.
\end{itemize}
In other words, the adversary draws $x_A = (x_{A, 1}, \ldots, x_{A, N}) \in \{0, 1\}^n, x_{B} = (x_{B, 1}, \ldots, x_{B, N}) \in \{0,1\}^n$ from $\mu^{N}$ where $(x_{A, \alpha}, x_{B, \alpha}) \in \{0,1\}^{2M}$ is drawn from $\mu$ for each $\alpha \in [N]$.
Let $i_{\alpha}^{*}$ be the intersecting index of the $\alpha$-th copy and $I^{*} = \{i_{\alpha}^{*}\}_{\alpha \in[N]}$ be the collection of intersecting indices.

%Let $\mu_{0}$ be the probability distribution of $B$ coins with mean $1/2$, $\mu_{1}$ be the probability distribution over $B$ coins such that one (random) coin has mean $1/2 - \eps$ and other $(B-1)$ coins have mean $1/2$.
%Let $\mu_{1, b}$ ($b\in [B]$) be the probability distribution such that the $b$-th coin has mean $1/2-\eps$ while other has mean $1/2$.

\paragraph{Adaptive loss sequence}
The entire loss sequence is divided into $\frac{10T}{N}$ epochs and each epoch contains $\frac{N}{10}$ days.
The loss sequence of each epoch is constructed separately. 
For any epoch $t \in [\frac{10 T}{N}]$ and day $b \in [\frac{N}{10}]$, let $N_{t,b} \subseteq [N]$ be the set of blocks where the online algorithm have played an action in during the first $(b-1)$ days.
The loss vector $\ell_{t, b} = (\ell_{t,b,1}, \ldots, \ell_{t, b,N}) \in [0, 1]^{n}$ is then constructed as 
\begin{itemize}
\item With probability $1 - \eps^2$, $\ell_{t, b, \alpha} = \vec{0}$ for any block $\alpha \in [M]$;
\item With probability $\eps^2$,
\begin{itemize}
\item For any played block $\alpha \in N_{t,b}$, $\ell_{t, b, \alpha} = \frac{1}{2} \cdot \vec{1}$;
\item For the rest of blocks,
\begin{itemize}
\item With probability $1/2$, $\ell_{t, b, \alpha} = \vec{1} - x_{A, \alpha}$ for all $\alpha \in [N]\backslash N_{t,b}$, 
\item With probability $1/2$, $\ell_{t, b, \alpha} = \vec{1} - x_{B, \alpha}$ for all $\alpha \in [N]\backslash N_{t,b}$, 
\end{itemize}
\end{itemize}
\end{itemize}
In other words, the nature sets $\ell_{t, b}$ to be zero loss $\vec{0}$ with probability $1-\eps^2$, and with probability $\eps^2$, it sets $\ell_{t, b}$ to be $\vec{1} - x_{A}$ or $\vec{1} - x_{B}$ with probability $1/2$, except for those blocks that have been by ``explored'' by the online learning algorithm, where the loss is set to $1/2$.

It is fairly straightforward to verify that the optimal expert has expected loss at most $\frac{\eps^2}{20}$.
\begin{lemma}[Optimal expert]
\label{lem:optimal-expert}
The optimal expert has expected average loss at most $\frac{\eps^2}{20}$.
\end{lemma}
\begin{proof}
Consider the set of $I^{*} = \{i_{1}^{*}, \ldots, i_{N}^{*}\}$. For each epoch $t \in [\frac{10 T}{N}]$, since there are $\frac{N}{10}$ days during the epoch, we know that $\frac{9}{10}$-fraction of indices in $I^{*}$ receive $0$ loss in the entire epoch, and the rest $\frac{1}{10}$-fraction of indices receive an (average) loss at most $1/2 \cdot \eps^2$.
Hence there exists at least one expert in $I^{*}$ that receives an average loss of at most $\frac{\eps^2}{20}$ during the entire $T$ days.
\end{proof}

The major work is devoted to bounding the expected loss of an online learning algorithm with memory $S = o(\frac{\sqrt{n}}{\eps\log^3 n})$. 
\begin{lemma}
\label{lem:lower-algo}
For any epoch $t\in [T/B]$, an online learning algorithm with memory $S = o(\frac{\sqrt{n}}{\eps\log^3 n})$ obtains an expected (average) loss of at least $\frac{\eps^2}{10}$.
\end{lemma}

The proof splits into three steps. 
We first formulate a three-party communication problem, \textsc{SetDisjointnessAd} (set disjointness with advice), and prove that a low memory online learning algorithm implies a low communication cost protocol for \textsc{SetDisjointnessAd}.
We then remove the ``advice'' and prove that a low communication cost protocol of \textsc{SetDisjointnessAd} implies a low communication cost protocol for multiple copies of \textsc{SetDisjointness} with non-negligible success probability.
Finally we use the direct-product theorem of \cite{braverman2015interactive} and the lower bound on internal information cost to prove an upper bound on the success probability.

\paragraph{Step 1} We first introduce the three-party communication problem of \textsc{SetDisjointnessAD}.

\begin{definition}[\textsc{SetDisjointnessAd}]
In the communication problem of \textsc{SetDisjointnessAd}, there are three parties -- Alice, Bob and Charlie.
The input of Alice and Bob, $x_A$ and $x_B$, are drawn from $\mu^{N}$. Charlie receives both $x_A$ and $x_B$.
The communication proceeds as follow. 
Charlie first sends a transcript $\Pi_C$ (of size at most $S$) to both Alice and Bob, and it is not allowed to speak anymore. 
Alice and Bob then perform multi-rounds of communication and output a set of indices $I_{\out} \subseteq I^{*}$ as large as possible.
\end{definition}

Let $\Pi_{A}$ and $\Pi_{B}$ be the transcripts of Alice and Bob, and $\Pi = (\Pi_A, \Pi_B)$ be the entire transcript between them. The following reduction comes from the fact that Alice, Bob and Charlie can simulate a low space online learning algorithm.

%The reduction from \textsc{FindHeavyBlock} to the online learning algorithm is fairly straightforward. The only (minor) difference is that in the online learning problem, once a player select an expert $i$ during epoch $t$, its performance become a fair coin. Intuitively, it only makes the online learning problem harder.

\begin{lemma}
\label{lem:online-simulate}
Suppose there is an online learning algorithm that uses at most $S$ space and achieves at most $\frac{\eps^2}{10}$ loss in expectation during the $t$-th epoch. 
Then there is a randomized public coin communication protocol for \textsc{SetDisjointnessAd}, such that 
\begin{enumerate}
\item $|\Pi_C|\leq S$; 
\item $|\Pi| \leq \eps \sqrt{n}\log n\cdot (S + \log n) + 2\eps^{-1}\sqrt{n}\log n = o(\frac{n}{\log^2 n})$; and
\item $\E[|I_{\out}|] \geq \frac{1}{25}N$.
\end{enumerate}
\end{lemma}
\begin{proof}
Given a \textsc{SetDisjointnessAd} instance, Charlie simulates the online learning algorithm on its input $(x_A, x_B)$ for $(t-1)$ epochs. Its transcript $\Pi_C$ is set to be the memory state at the beginning of epoch $t$. 
Alice and Bob then simulate the online learning algorithm. 
They proceed in $\frac{N}{10}$ rounds and at round $b \in [\frac{N}{10}]$, they use public randomness to determine the loss vector. That is, 
\begin{itemize}
\item {\bf Case 1.} With probability $1-\eps^2$, the loss vector is $\vec{0}$. Alice and Bob then simulate the online learning algorithm locally (both of them possess the memory state up to round $(b-1)$). This incurs zero communication cost.
\item {\bf Case 2.} With probability $\eps^2$, the loss vector depends on Alice or Bob's input. In particular, with probability $1/2$, the loss vector is set to be $\vec{1} - x_{A}$ for blocks that have not been played, and $\frac{1}{2}\cdot \vec{1}$ for played blocks. Alice simulates the online learning algorithm, adds its memory state as well as the action $i_{t, b}$ to the transcript. This incurs $S + \log n$ communication cost. With the rest $1/2$ probability , Bob performs the same operation.
\end{itemize}
Alice and Bob stop simulating the online learning algorithm if the total number of Case 2 exceeds $\eps \sqrt{n}\log n = 10 \log n \cdot \frac{\eps^2N}{10}$. Note this happens with probability at most $1/n^{\omega(1)}$ by Chernoff bound.

Let $I = \{i_{t, 1}, \ldots, i_{t, \frac{N}{10}}\}$ be the collection of actions played in these $\frac{N}{10}$ rounds, the final output set $I_{\out} = I\cap I^{*}$ is determined using two extra rounds of communication.

The total communication cost can be bounded as 
\[
|\Pi_C|\leq S \quad  \text{and} \quad  |\Pi| \leq \eps \sqrt{n}\log n\cdot (S + \log n) + 2\eps^{-1}\sqrt{n}\log n.
\]
To bound the size of $I_{\out}$, note Alice, Bob and Charlie simulate the online learning algorithm, which gets at most $\frac{\eps^2}{10}$ (average) loss in expectation.
In particular, at day $b \in [N/10]$, if the algorithm commits an action $i_{t, b} \in I^{*}\backslash N_{t, b}$, then it receives expected loss $0$, otherwise, it receives at least $\frac{1}{2}\eps^2$ loss in expectation. The later can happen at most 
$
\frac{(\eps^2/10) \cdot (N/10)}{\eps^2/2} = \frac{N}{50}
$
days. This implies the online learning algorithm commits at least $\frac{N}{10} - \frac{N}{50} = \frac{2}{25}N$ intersecting indices in $I^{*}$ in expectation. Note Alice and Bob fail to simulate the online learning protocol with probability no more than probability $1/n^{\omega(1)}$, hence their output satisfies $\E[|I_{\out}|] \geq \frac{N}{25}$. This finishes the proof.
\end{proof}

\paragraph{Step 2} We next define the communication problem of multi-copy \textsc{SetDisjointness}.

\begin{definition}[Multi-copy \textsc{SetDisjointness}]
In the communication problem of multi-copy \textsc{SetDisjointness}, Alice and Bob receive input $x_{A}, x_{B}$ separately, drawn from $(x_A, x_B) \sim  \mu^{N}$. They perform multi-rounds communication and the goal is to output the index set $I^{*}$.
\end{definition}

Let $\suc(\Pi, \mu^{N})$ be the success probability of a communication protocol $\Pi$ over the distribution $\mu^{N}$.
The key observation is that a low communication cost protocol of \textsc{SetDisjointnessAd} implies a low communication cost protocol for multi-copy \textsc{SetDisjointness} with non-negligible success probability.
Formally,

\begin{lemma}
\label{lem:suc}
Suppose there is a communication protocol $(\Pi_{C}, \Pi)$ for \textsc{SetDisjointnessAd} with $|\Pi_{C}| \leq S$ and $|\Pi| \leq S_2$, and its output $I_{\out} \subseteq I^{*}$ satisfies $\E[|I_{\out}|] \geq \frac{N}{25}$. Then there exists a randomized public coin communication protocol $\Pi_{N}$ that solves the Multi-copy \textsc{SetDisjointness} with communication cost $O(S_2 \log n)$ and success probability at least $2^{-O(S \log n)}$.
\end{lemma}
\begin{proof}
It is w.l.o.g.~to assume the protocol of \textsc{SetDisjointnessAd} is deterministic as it is a distributional communication problem. Let $f_C: \{0,1\}^{2n}\rightarrow \{0,1\}^{S}$ be Charlie's (deterministic) strategy that maps an input $(x_A, x_B) \in \{0,1\}^{2n}$ to a transcript $\Pi_C \in \{0,1\}^S$. 
For any feasible transcript $\Pi_C \in \{0,1\}^S$, define
\[
f_{C}^{-1}(\Pi_C) = \left\{(x_A, x_B): f_C(x_A, x_B) = \Pi_C  \right\}
\]

That is, $f_{C}^{-1}(\Pi_C)$ is the pre-image of $\Pi_C$. 
Let $\mu_{C}$ be the distribution of $\Pi_C$.

\paragraph{Step 2-1}  Consider the following communication protocol $\Pi_{N,1}$ for multi-copy \textsc{SetDisjointness}.
\begin{itemize}
\item Alice and Bob first sample a possible advice $\hat{\Pi}_C \sim \mu_{C}$ using public randomness;
\item Alice and Bob then simulate the protocol $\Pi$ of \textsc{SetDisjointnessAd} given $\hat{\Pi}_C$.
\end{itemize}

The communication cost is $|\Pi_{N, 1}| = |\Pi| = S_2$. We prove with probability at least $\frac{1}{2500} \cdot 2^{-S}$, the output index set $I_{\out}^{\Pi, \hat{\Pi}_C}\subseteq I^{*}$ satisfies $|I_{\out}| \geq \frac{N}{50}$. 
In particular, one has
\begin{align*}
&~ \Pr\left[|I_{\out}^{\Pi, \hat{\Pi}_C}| \geq  \frac{N}{50}\right]\\
\geq &~ \sum_{\Pi_C \in \{0,1\}^S}\Pr[(x_A, x_{B}) \in f_{C}^{-1}(\Pi_{C})]\cdot \Pr[\hat{\Pi}_C =\Pi_{C}] \cdot \Pr\left[|_{\out}^{\Pi, \hat{\Pi}_C} | \geq \frac{N}{50} \mid  (x_A, x_B) \in f_{C}^{-1}(\hat{\Pi}_C) \right] \\
= &~ \sum_{\Pi_C \in \{0,1\}^S}\Pr[(x_A, x_{B}) \in f_{C}^{-1}(\Pi_{C})]^2 \cdot \Pr\left[|I_{\out}^{\Pi, \Pi_C} | \geq \frac{N}{50} \mid  (x_A, x_B) \in f_{C}^{-1}(\Pi_C) \right] \geq \frac{1}{2500} \cdot 2^{-S}.
\end{align*}
We restrict the success event to the case of $\hat{\Pi}_C = \Pi_{C}$ in the first step, the second step holds since Alice and Bob draw $\hat{\Pi}_{C}$ from $\mu_C$, and therefore $\Pr[\hat{\Pi}_C = \Pi_C] = \Pr[(x_A, x_B) \in f_{C}^{-1}(\Pi_C)]$. 
The last step follows from Cauchy-Schwartz and the fact that 
\begin{align*}
&~\sum_{\Pi_C \in \{0,1\}^S}\Pr[(x_A, x_{B}) \in f_{C}^{-1}(\Pi_{C})] \cdot \Pr\left[|I_{\out}^{\Pi, \Pi_C} \mid \geq \frac{N}{50} \mid  (x_A, x_B) \in f_{C}^{-1}(\Pi_C) \right] \\
= &~ \Pr\left[|I_{\out}^{\Pi, \Pi_{C}} | \geq \frac{N}{50}\right] > \frac{1}{50},
\end{align*}
we use the fact that $\E[|I_{\out}^{\Pi, \Pi_C}|] \geq \frac{N}{25}$ in the last step.

\paragraph{Step 2-2} The above protocol $\Pi_{N,1}$ guarantees that Alice and Bob find a set of intersecting indices of size at least $\frac{N}{50}$ with probability at least $2^{-O(S)}$. We next boost the size of set by repeating $\Pi_{N,1}$ for $R = 100\log n$ times. Our final protocol $\Pi_N$ works as follow.
\begin{itemize}
\item For $r = 1,2, \ldots, R$:
\begin{itemize}
\item Using public randomness, Alice and Bob sample permutation functions $\gamma_{r}: [N]\rightarrow [N]$  and $\theta_{r, \alpha}: [M] \rightarrow [M]$ ($\alpha \in [M]$)
\item Alice and Bob reset their input $x_{A}^{(r)}, x_{B}^{(r)}$ as $x_{A,\alpha, \beta}^{(r)} = x_{A, \gamma_r(\alpha), \theta_{r, \alpha}(\beta)}$ and $x_{B,\alpha, \beta}^{(r)} = x_{B, \gamma_r(\alpha), \theta_{r, \alpha}(\beta)}$
\item Alice and Bob run the communication protocol $\Pi_{N, 1}$ over input $x_{A}^{(r)}$ and $x_{B}^{(r)}$ and perform the same permutation to the output set, i.e., $I_{\out}^{(r)} \leftarrow \{(\gamma_r(\alpha), \theta_{r, \alpha}(\beta)): (\alpha, \beta) \in I_{\out}^{(r)} \}$
\end{itemize}
\item Return $I_{\out} \leftarrow I_{\out}^{(1)} \cup \cdots \cup I_{\out}^{(r)}$.
\end{itemize}

The key observation is that the collection of inputs over $R$ iterations $(x_A^{(1)}, x_{B}^{(1)} \ldots, x_{A}^{(R)}, x_{B}^{(R)})$ is drawn from the product distribution $(\mu^N)^{R}$, this crucially uses the symmetric property of input distribution $\mu$. 

Let $\mE_r$ be the success event of $r$-th iteration, i.e., $|I_{\out}^{(r)}| \geq \frac{N}{50}$, then we have
\begin{align}
\Pr[\mE_1 \cup \cdots \cup \mE_{R}] = \Pr[\mE_1] \cdots \Pr[\mE_R] = 2^{-SR} \label{eq:independent1}
\end{align}
and for any $\alpha \in [n]$
\begin{align}
\Pr\left[i_{\alpha}^{*} \notin \bigcup_{r\in [R]} I_{\out}^{(r)} \,\mid\, \bigcup_{r\in [R]}\mE_{r}\right] = \prod_{r\in [R]}\Pr\left[i_{\alpha}^{*} \notin  I^{(r)}_{\out} | \mE_{r}\right] \leq \left(1 - \frac{1}{50}\right)^{R} \leq \frac{1}{n^2}. \label{eq:independent2}
\end{align}
Both equations use the independent property of $(x_A^{(1)}, x_{B}^{(1)} \ldots, x_{A}^{(R)}, x_{B}^{(R)})$. The second step of Eq.~\eqref{eq:independent2} holds due to the fact that each index of $I^{*}$ are contained in $I_{\out}^{(r)}$ with equal probability (because of the shuffling step) and the last step holds due to the choice of $R = 100 \log n$  

Taking an union bound over $\alpha \in [N]$, we obtain
\[
\Pr\left[I^{*} \neq \bigcup_{r\in [R]} I^{(r)}_{\out} \,\mid\, \bigcup_{r\in [R]}\mE_{r}\right] \leq \frac{1}{n}
\]
and therefore 
\[
\Pr[I_{\out} = I^{*}] \geq \Pr[\mE_1 \cup \cdots \cup \mE_{R}] \cdot \Pr\left[I^{*} =  \bigcup_{r\in [R]} I^{(r)}_{\out} \,\mid\, \bigcup_{r\in [R]}\mE_{r}\right] \geq 2^{-O(S\log n)}
\]
where the second step comes from Eq.~\eqref{eq:independent1}\eqref{eq:independent2}. Finally, we note the communication cost equals $O(S_2 R)= O(S_2\log n)$.
\end{proof}

\paragraph{Step 3}

Finally, we prove a lower bound on the success probability of multi-copy {\sc SetDisjointness}. In particular, we prove
\begin{lemma}
\label{lem:product}
Suppose a communication protocol $\Pi$ of multi-copy \textsc{SetDisjointness} succeeds with probability at least $\exp(-O(N))$, then $|\Pi| \geq \Omega(n/\log n)$.
\end{lemma}

We make use of the following direct product theorem due to \cite{braverman2015interactive}
\begin{theorem}[Direct product, Theorem 2 of \cite{braverman2015interactive}]
\label{thm:direct-product}
There is a global constant $a > 0$ such that for any two party communication problem $f$ defined over distribution $\nu$. If $\suc(\nu, f, I) \leq \frac{2}{3}$, and $K \log (K) \leq \alpha N \cdot I$, it holds that $\suc(\mu^{N}, f^{N}, K) \leq \exp(-\Omega(N))$. Here $\suc(\nu, f, I)$ is the maximal success probability of a protocol with internal information cost $I$ of computing $f$ under $\mu$.
\end{theorem}

We prove that a communication protocol for {\sc SetDisjointness} must have $\Omega(M) = \Omega(\eps\sqrt{n})$ internal information cost under $\mu$, in order to succeed with probability at least $2/3$. 
\begin{lemma}
\label{lem:internal}
The internal information cost of any communication protocol for {\sc Disjointness} (with success probability at least $2/3$) under distribution $\mu$ is at least $\Omega(M) = \Omega(\eps\sqrt{n})$.
\end{lemma}

Let alone the (standard) proof of Lemma \ref{lem:internal} (see Appendix \ref{sec:lower-app}), together with Theorem \ref{thm:direct-product}, we have proved Lemma \ref{lem:product}. Combine Lemma \ref{lem:product}, Lemma \ref{lem:suc} and Lemma \ref{lem:online-simulate}, we have proved Lemma \ref{lem:lower-algo}, i.e. the online learning algorithm obtains an average regret at least $\frac{\eps^2}{10}$. This finishes the proof of Theorem \ref{thm:lower} since the optimal expert has expected average loss at most $\frac{\eps^2}{20}$.

\subsection{A \texorpdfstring{$\Omega(n)$}{Lg} memory lower bound for strong adaptive adversary}
\label{sec:strong-adap-app}

Finally, we adopt the methodology of Theorem \ref{thm:lower-main} and prove a linear memory lower bound $\Omega(n)$ against against {\em a strong adaptive adversary}, this improves the result of \cite{peng2022online}.
A strong adaptive adversary can choose the loss vector based on the algorithm's strategy (a probability distribution over $[n]$) on the current round (instead of only past actions).
We refer readers to \cite{peng2022online} for a detailed discussion on this notion.

\begin{theorem}
\label{thm:strong-adaptive}
Any online learning algorithm that can obtain $o(T)$ regret requires $\Omega(n)$ space when facing a strong adaptive adversary.
\end{theorem}
\begin{proof}
The construction of hard instance follows from the high-level methodology of Theorem \ref{thm:lower-main}, despite the proof uses a simple counting argument.
We prove by contradiction and suppose $S < \frac{n}{2000}$. 
Let $I$ be a set of special experts and $|I| = 10S$. Given a strategy $p\in \Delta_{n}$, the loss vector $\ell_{I, p}$ is constructed as follow.
\begin{itemize}
\item For any non-special expert $i \in [n]\backslash I$, $\ell_{I, p}(i) = 1$;
\item For any special expert $i \in [I]$, we set $\ell_{I, p}(i) = 0$ if $i$ is not among the top-$2S$ heavy experts of $I$, i.e., $|\{i'\in I: p(i')\geq p(i)\}| \geq 2 S$; we set $\ell_{I, p}(i) = 1$ otherwise (i.e., $i$ is among the top-$2S$ heavy experts in $I$).
\end{itemize}

We immediately have
\begin{lemma}\label{lem:count}
For any fixed strategy $p \in \Delta_{n}$, there are at most $2^{10S} \cdot \binom{n}{8S}$ number of sets $I \subseteq [n]$ ($|I| = 100S$) such that $\langle p, \ell_{I, p}\rangle < \frac{2}{5}$.
\end{lemma}
\begin{proof}
We assume $p_1\geq \cdots \geq p_n$, this is wlog. 
If $\sum_{i=10S+1}^{n}p_{i} \geq \frac{2}{5}$, then it is easy to verify that $\langle p, \ell_{I, p}\rangle \geq \frac{2}{5}$ for any $I$. 
Consider the case of $\sum_{i=10S+1}^{n}p_{i} < \frac{2}{5}$, if $|I \cap [10S]| \leq 2S$, then $\ell_{I, p}(i) = 1$ for any $i \in [10S]$, and therefore,
\[
\langle p, \ell_{I, p}\rangle \geq \sum_{i=1}^{10S} p_i \ell_{I, p}(i) = 3/5.
\]
Hence, the total number of $I$ is at most 
\[
\sum_{k = 2S +1}\binom{10S}{k}\binom{n - 10 S}{10S - k} \leq  2^{10S} \cdot \binom{n}{8S}.
\]
\end{proof}

Back to the proof of Theorem \ref{thm:strong-adaptive}.
The loss sequence is generated as follow.
\begin{itemize}
\item The adversary chooses the set of special experts $I$ ($|I| = 100 S$) uniformly at random from $[n]$. The set is chosen at the beginning and fixed afterwards.
\item At day $t$, given the strategy $p_t$ of the algorithm, the adversary constructs the loss vector $\ell_{I, p_t}$.
\end{itemize}
At each day, there are at $2S$ special experts receives loss $1$ and other $8S$ special experts receive $0$ loss. Hence the best expert receives loss at most $\frac{1}{5}$.
We aim to prove that Alice has expected loss at least $\frac{3}{10}$ for every day, with high probability over the choice of $I$. 

The algorithm has $2^{S}$ possible memory states $X$ in total. 
At each memory state $x \in X$, suppose the algorithm outputs a strategy $p_x \in \Delta_n$ (note $p_x$ could be a random variable). Define 
\begin{align*}
T_x := \left\{I: \Pr[\langle p_x, \ell_{I, p_x}\rangle < \frac{2}{5}] \geq 0.1, I\subseteq [n], |I| = 10S       \right\}.
\end{align*}

By Lemma \ref{lem:count}, we know that $|T_x| \leq 10\cdot 2^{10S} \binom{n}{8S}$. Taking an union over all memory states $x\in X$, one has
\[
\left|\bigcup_{x\in X} T_x\right| \leq 2^{S} \cdot 10 \cdot 2^{10S} \cdot \binom{n}{8S} \leq \frac{1}{10}\binom{n}{10S}.
\]
The last step follows from the choice of parameters.

Hence, we conclude that with probability at least $0.9$, the adversary draws a set of special experts $I$ such that $I \notin \bigcup_{x\in X} T_x$.
This means that at each day, the algorithm receives at least $0.9 \cdot \frac{2}{5} \geq \frac{3}{10}$ loss in expectation. We conclude the proof here.
\end{proof}

%% file: appendix-prob.tex
\section{Probabilistic tool}

We state the version of Chernoff bound and Azuma-Hoeffding bound used in this paper.
\begin{lemma}[Chernoff bound]
Let $X = \sum_{i=1}^n X_i$, where $X_i=1$ with probability $p_i$ and $X_i = 0$ with probability $1-p_i$, and all $X_i$ are independent. Let $\mu = \E[X] = \sum_{i=1}^n p_i$. Then \\
1. $ \Pr[ X \geq (1+\delta) \mu ] \leq \exp ( - \delta \mu / 3 ) $, $\forall\, \delta > 1$ ; \\
2. $ \Pr[ X \leq (1-\delta) \mu ] \leq \exp ( - \delta^2 \mu / 2 ) $, $\forall\, 0 < \delta < 1$. 
\end{lemma}

\begin{lemma}[Azuma-Hoeffding bound]
Let $X_0, \ldots, X_n$ be a martingale sequence with respect to the filter $F_0 \subseteq F_1 \cdots \subseteq F_n$ such that for $Y_i = X_i - X_{i-1}$, $i \in [n]$, we have that $|Y_i| = |X_i - X_{i-1}| \leq c_i$. Then
\begin{align*}
    \Pr[|X_t - X_0 | \geq t ] \leq 2\exp\left(-\frac{t^2}{2\sum_{i=1}^{n}c_i^2}\right).
\end{align*}
\end{lemma}

\section{Missing proof from Section \ref{sec:adaptive-lower}}
\label{sec:lower-app}
\begin{proof}[Proof of Lemma \ref{lem:internal}]
Consider the following distribution $\mu'$ defined over $\{0, 1\}^{2M}$
\begin{itemize}
\item Sample an index $i^{*} \in [M]$ uniformly at random and set $(x_{A, i}, x_{B, i}) = (1,1)$
\item For the rest of coordinates, set $(x_{A, i}, x_{B, i})$ to be $(0,1), (1, 0), (0,0)$ with equal probability.
\end{itemize}
The only (minor) difference between $\mu$ and $\mu'$ is that $(x_{A}, x_{B})$ contains exactly $\frac{M-1}{3}$ pairs of $(0,0), (0,1), (1,1)$ under $\mu$.

An internal information cost lower bound of $\mu'$ is known, in particular, we have
\begin{lemma}[Theorem 3 of \cite{assadi2019polynomial}]
\label{lem:disjoint}
The internal information cost of any communication protocol for {\sc Disjointness} (with success probability at least $2/3$) under distribution $\mu'$ is at least $\Omega(M)$.
\end{lemma}

Given an instance of {\sc Setdisjointness} $(x_A', x_{B}')$ under distribution $\mu'$, we turn it into an input from $\mu$ by augmenting a few coordinates.
%Consider the following protocol $\Pi'$. Alice and Bob first compute the number of $(0,1)$, $(1, 0)$ and $(0,0)$ coordinates, this can be done by simply exchanging $|x_{A}'|$ and $|x_{B}'|$, i.e., the number of non-zero entry of $x_{A}'$ and $x_{B}'$.
Alice and Bob augment their input to $x_{A}, x_{B}$ ($x_A, x_B \in \{0, 1\}^{3M+1}$ such that there are exact $M$ pairs of $(0,1), (1,0), (0,1)$. This can be done by (1) Alice appends $M - |x_A'| + 1$ coordinates of $1$ and (2) Bob appends $M - |x_B'| +1$ (different) coordinates of $1$. They then use public randomness to shuffle the coordinates and then run the protocol $\Pi$ (designed for input from $\mu$) on $x_{A}, x_{B}$. Let $\gamma$ be the public randomness used for shuffling.

It is clearly that Alice and Bob would success with probability at least $2/3$, to bound the internal information cost, we note
\begin{align*}
\IC_{\mu'}(\Pi') = &~ I(X_A'; \Pi'|X_B') + I(X_B'; \Pi'|X_A') \\
= &~ I(X_A'; \Pi, \gamma |X_B') + I(X_B'; \Pi, \gamma|X_A') \\
= &~ I(X_A'; \Pi |X_B', \gamma) + I(X_B'; \Pi|X_A', \gamma) \\
= &~ I(X_A; \Pi |X_B, \gamma) + I(X_B; \Pi |X_A, \gamma)\\
\leq &~ I(X_A; \Pi |X_B) + I(X_B; \Pi |X_A)\\
= &~ \IC_{\mu}(\Pi).
\end{align*}
The third step follows from the chain rule of mutual information and $I(X_A'; \gamma| X_{B}') = 0$.
The fourth step holds since $(X_B, \gamma)$ are one to one $(X_B', \gamma)$, and  $(X_A, \gamma)$ are one to one $(X_A', \gamma)$.
The fifth step holds since $I(\gamma; \Pi | X_{A}, X_{B}) = 0$, i.e., the transcript $\Pi$ is independent of the  permutation function given $X_A$ and $X_B$. We finish the proof here.
\end{proof}